\documentclass[12pt]{article}

\usepackage{amsmath, amssymb, amsthm}
\usepackage{graphicx}
\usepackage{enumerate}
\usepackage{natbib}
\usepackage{bm}
\usepackage{dsfont}
\usepackage{url}  
\usepackage{microtype} 
\usepackage{tikz}
\usepackage{mathtools}
\usepackage{paralist}
\usepackage{booktabs} 
\usepackage{wasysym}
\usepackage{subcaption}
\usepackage{fourier}
\usepackage{enumitem} 

\makeatletter
\newcommand*\rel@kern[1]{\kern#1\dimexpr\macc@kerna}
\newcommand*\widebar[1]{%
  \begingroup
  \def\mathaccent##1##2{%
    \rel@kern{0.8}%
    \overline{\rel@kern{-0.8}\macc@nucleus\rel@kern{0.2}}%
    \rel@kern{-0.2}%
  }%
  \macc@depth\@ne
  \let\math@bgroup\@empty \let\math@egroup\macc@set@skewchar
  \mathsurround\z@ \frozen@everymath{\mathgroup\macc@group\relax}%
  \macc@set@skewchar\relax
  \let\mathaccentV\macc@nested@a
  \macc@nested@a\relax111{#1}%
  \endgroup
}
\makeatother
\newcommand{\prob}{\operatorname{\mathsf{P}}}

\newcommand{\diff}{\mathrm{d}}

\newcommand{\bV}{\bm{V}}
\newcommand{\bw}{\bm{w}}

\newcommand{\bx}{\bm{x}}
\newcommand{\bX}{\bm{X}}
\newcommand{\bY}{\bm{Y}}

\newcommand{\bone}{\bm{1}}

\newcommand{\xw}{\frac{x}{x+y}}
\newcommand{\yw}{\frac{y}{x+y}}

\newcommand{\reals}{\mathbb{R}}

\newcommand{\1}{\operatorname{\mathds{1}}}

\newcommand{\cbr}[1]{\left\{ {#1} \right\}}
\newcommand{\rbr}[1]{\left( {#1} \right)}
\newcommand{\sbr}[1]{\left[ {#1} \right]}

\newcommand{\hatc}{\hat{c}}

\renewcommand{\ge}{\geqslant}
\renewcommand{\le}{\leqslant}
\renewcommand{\geq}{\geqslant}
\renewcommand{\leq}{\leqslant}

\NewDocumentCommand{\evalat}{sO{\big}mm}{%
  \IfBooleanTF{#1}
   {\mleft. #3 \mright|_{#4}}
   {#3#2|_{#4}}%
}

\DeclareRobustCommand{\bigO}{%
  \text{\usefont{OMS}{cmsy}{m}{n}O}%
}
\DeclareRobustCommand{\smallo}{\ensuremath{o}}

\usepackage{xcolor}
\definecolor{orange-red}{rgb}{1.0, 0.27, 0.0}

\usepackage[utf8]{inputenc}
\usepackage[T1]{fontenc}
\usepackage[a4paper,left=2.25cm, right = 2.25cm, top = 1.75cm, bottom = 2.25cm]{geometry}
\usepackage[colorlinks=true,linkcolor=black,citecolor=black,urlcolor=black]{hyperref}

\newtheorem{lemma}{Lemma}
\newtheorem{proposition}{Proposition}

\newtheorem{assumption}{Assumption}

\theoremstyle{definition}
\newtheorem{definition}{Definition}
\newtheorem{example}{Example}

\theoremstyle{remark}
\newtheorem{remark}{Remark}

\numberwithin{equation}{section}

\begin{document}

\title{
Geometric criteria for identifying extremal dependence and flexible modeling via additive mixtures}

\author{Jeongjin Lee\thanks{School of Mathematical Sciences, Lancaster University, Fylde College, Lancaster, LA1 4YF, United Kingdom. Corresponding author. E-mail: j.lee58@lancaster.ac.uk}\footnotemark[1]
\and Jennifer Wadsworth\thanks{School of Mathematical Sciences, Lancaster University, Fylde College, Lancaster, LA1 4YF, United Kingdom. E-mail: j.wadsworth@lancaster.ac.uk}}

\date{\today}

\maketitle

\begin{abstract}
\smallskip
The framework of geometric extremes is based on the convergence of scaled sample clouds onto a limit set, characterized by a gauge function, with the shape of the limit set determining extremal dependence structures.
While it is known that a blunt limit set implies asymptotic independence, the absence of bluntness can be linked to both asymptotic dependence and independence.
Focusing on the bivariate case, under a truncated gamma modeling assumption with bounded angular density, we show that a ``pointy'' limit set implies asymptotic dependence, thus offering practical geometric criteria for identifying extremal dependence classes.
Suitable models for the gauge function offer the ability to capture asymptotically independent or dependent data structures, without requiring prior knowledge of the true extremal dependence structure.
The geometric approach thus offers a simple alternative to various parametric copula models that have been developed for this purpose in recent years.
We consider two types of additively mixed gauge functions that provide a smooth interpolation between asymptotic dependence and asymptotic independence.
We derive their explicit forms, explore their properties, and establish connections to the developed geometric criteria.
Through a simulation study, we evaluate the effectiveness of the geometric approach with additively mixed gauge functions, comparing its performance to existing methodologies that account for both asymptotic dependence and asymptotic independence.
The methodology is computationally efficient and yields reliable performance across various extremal dependence scenarios.

\end{abstract}


\section{Introduction}
\label{sec:Intro}

Classical extreme value modeling primarily focuses on scenarios where all variables exhibit extreme behavior simultaneously.
However, in many applications, it is common to observe cases where some variables exhibit extreme behavior while others remain non-extreme. 
For a bivariate random vector, two main scenarios are possible: simultaneous extremes (asymptotic dependence) or non-simultaneous extremes (asymptotic independence). For a random vector $\bX:=(X,Y)^\top$ with marginal distributions $F_X(x)$ and $F_Y(y)$, the class of extremal dependence can be determined by the tail dependence coefficient $\chi=\lim_{u\uparrow 1}\prob\rbr{F_X(x)\ge u \mid F_Y(Y)\ge u}.$
A positive limit $\chi>0$ indicates asymptotic dependence (AD) for the pair $(X,Y)$, while $\chi=0$ corresponds to asymptotic independence (AI).

To accommodate both AD and AI without prior knowledge of the true extremal dependence class, we adopt a geometric framework, which has recently been translated into a statistical modeling approach.
Early theoretical studies explored the geometric aspects of light-tailed multivariate sample clouds, particularly their convergence onto the limit set and the associated gauge functions \citep{davis1988almost,kinoshita1991convergence}.
More recently, \cite{nolde2022linking} established connections between the limit set and various representations of multivariate extremes.
Statistical exploration of this so-called geometric framework remains relatively recent. \citet{simpson2024estimating},
\cite{wadsworth2024statistical} and \citet{papastathopoulos2023statistical} proposed statistical methodology for estimating the limit set $G$ from data.
This method enables modeling of joint extremes across different sub-variable combinations and facilitates extrapolation into joint tails where only a subset of variables exhibit extreme behavior.
Compared to classical methods for multivariate extremes, the geometric approach provides greater flexibility in capturing complex dependence structures.

In this work, we focus on the bivariate case and explore statistical models in the geometric framework that facilitate a smooth transition between AD and AI.
We establish a connection between the extremal dependence class and the shape of the limit set, identifying conditions under which a given limit set and its associated gauge functions correspond to AD or AI.
While it is known that a ``blunt'' limit set implies asymptotic independence \citep{balkema2010asymptotic}, the absence of bluntness can be associated with both AD and AI.
Under a truncated gamma modeling assumption, we define the geometric characteristics of the limit set for each dependence class and show that a ``pointy'' limit set implies asymptotic dependence when the corresponding angular density is bounded, providing practical geometric criteria for identifying extremal dependence classes.

Various parametric copula models have also been developed in recent years to identify extremal dependence classes
\citep[e.g.,][]{wadsworth2017modelling,huser2019modeling,engelke2019extremal}.
These copula models are constructed using random scale representations.
In the bivariate case, the copula of random vector $\bX$ is generally expressed as some variation of the product of a positive random variable, $S$, and a random vector $(V_1,V_2)^\top,$ independent of $S$.
Implementation of copula-based models defined by random scale mixtures often relies on computation of numerical integrals, and numerical inversion techniques leading to relatively high computational costs.

In contrast, the geometric approach provides a simple alternative to the task of interpolating between dependence classes by employing models for the gauge function that can exhibit both pointy and blunt shapes.
We investigate the construction, properties, and efficient implementation of such gauge functions via two types of additive mixing.
Furthermore, we derive explicit formulations, explore their relationship with the proposed geometric criteria, and establish links to other extremal dependence measures.

The outline of this paper is as follows.
In Section~\ref{sec:Background}, we provide a brief review of the background on the geometric approach and the key model assumptions.
In Section~\ref{sec:Theory}, we develop criteria for classifying extremal dependence under a truncated gamma model assumption, while in Section~\ref{sec:additive} we derive closed-form expressions for various additively mixed gauge functions.
In Section~\ref{sec:InferenceSimulation}, we demonstrate the performance of the geometric approach with the additively mixed gauge functions and compare it to existing methods that accommodate both AD and AI.
In Section~\ref{sec:Application}, we apply the methodology to river flow data.
We conclude in Section~\ref{sec:Conclusion}.

\section{Background}
\label{sec:Background}

\subsection{Limit sets and gauge functions}
\label{sec:Gaugeft}

The framework for convergence onto limit sets requires light-tailed margins.
Standardizing the margins to a particular form ensures that the shape of the limit set can be interpreted in terms of extremal dependence.
Common choices are standard exponential and Laplace margins.
Since we focus on asymptotic dependence properties in the positive quadrant, we use the former.
Let $\bX_i$, $i=1,\ldots,n,$ be independent copies of a random vector $\bX=(X_1,\ldots,X_d)^\top$ with standard exponential margins, i.e., $\prob(X_j>x)=\exp(-x),$ $x>0$, and define the scaled $n$-point sample cloud as
\begin{equation*}
\label{eq:sampleCloud}
    N_n=\cbr{\bX_1/\log{n},\ldots,\bX_n/\log{n}}.
\end{equation*}
We assume that $N_n$ converges in probability onto a limit set $G$, which can be characterized by a gauge function $g_{\bX}(\bx)$ via $G=\cbr{\bx\in \reals_+^d: g_{\bX}(\bx)\le 1}$ \citep{balkema2010asymptotic}.
The limit set $G$ is star-shaped, meaning that $\bx\in G \Rightarrow t\bx\in G$ for all $t\in[0,1]$ \citep{kinoshita1991convergence}.
The convergence of the scaled sample cloud onto the limit set $G$ is illustrated in Figure~\ref{fig:SampleCloud} for two dependence structures: the logistic, which exhibits asymptotic dependence, and Gaussian, which exhibits asymptotic independence.

\begin{figure}[ht]
\centering
\includegraphics[width=3.5cm]{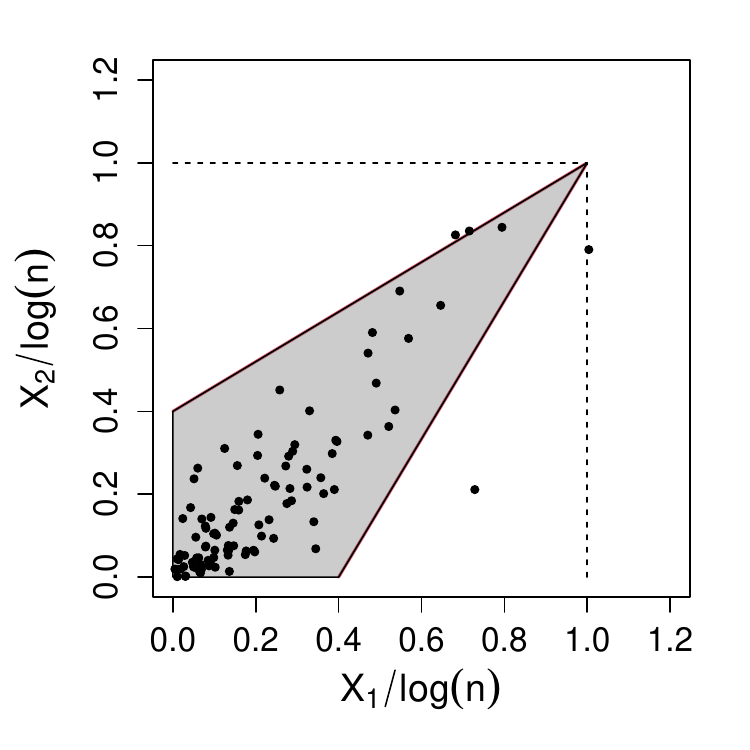}
\includegraphics[width=3.5cm]{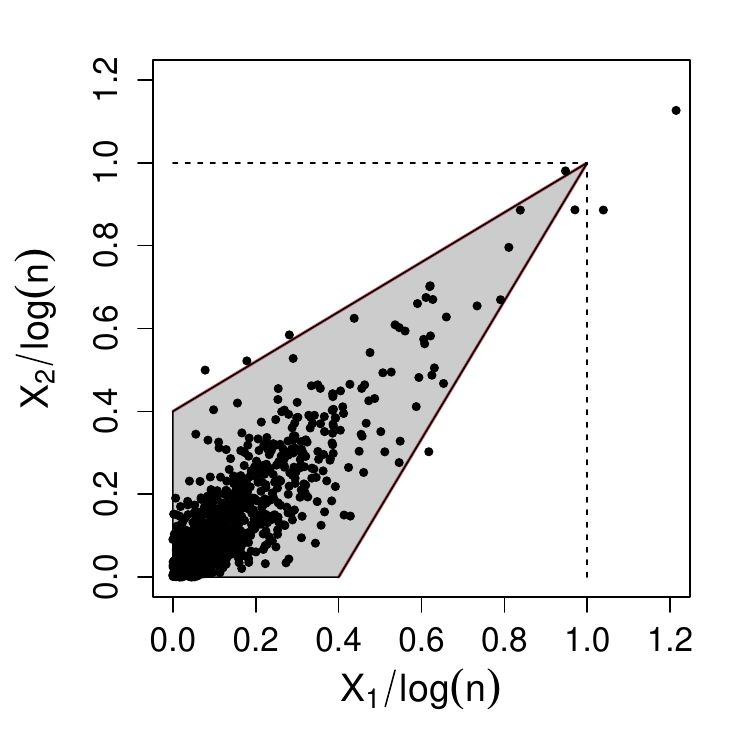}
\includegraphics[width=3.5cm]{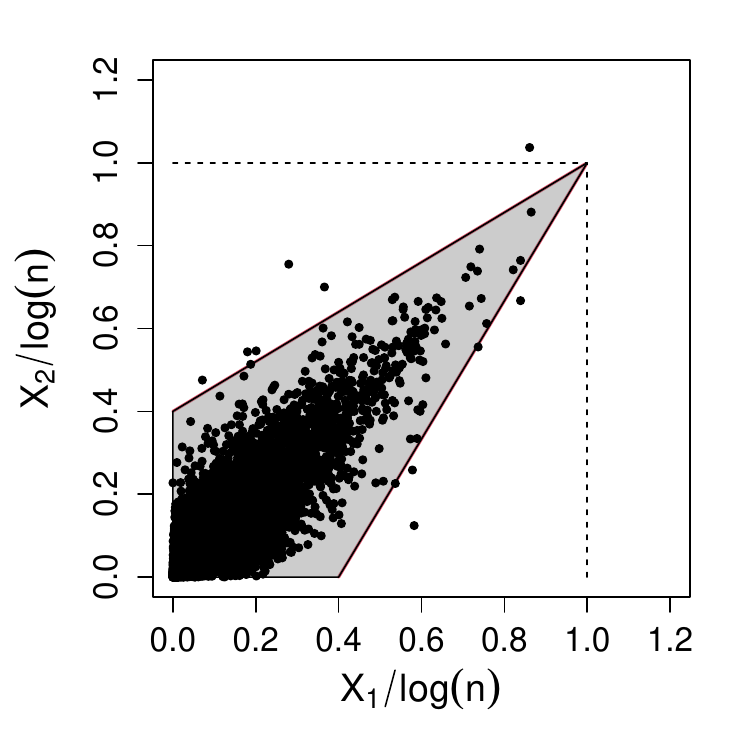}\\
\includegraphics[width=3.5cm]{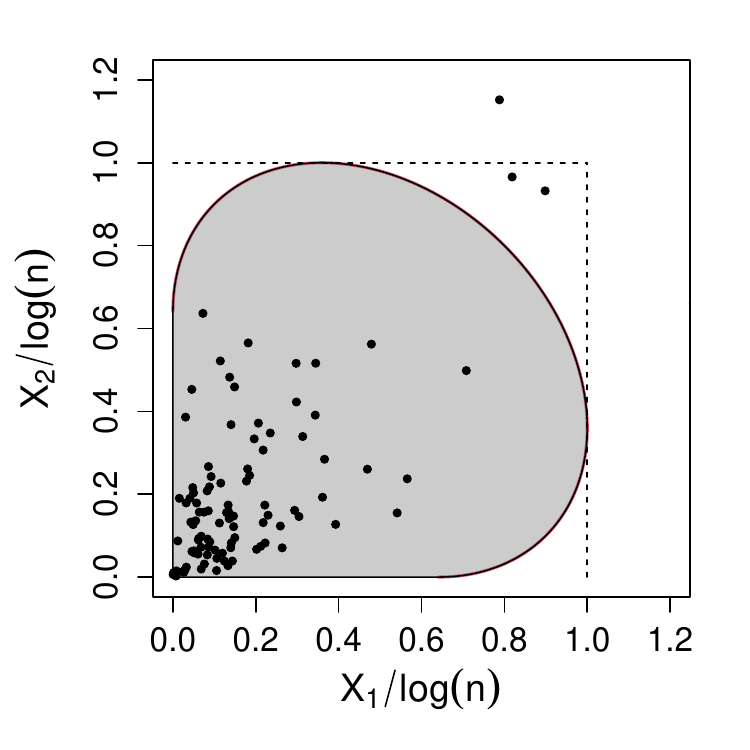}
\includegraphics[width=3.5cm]{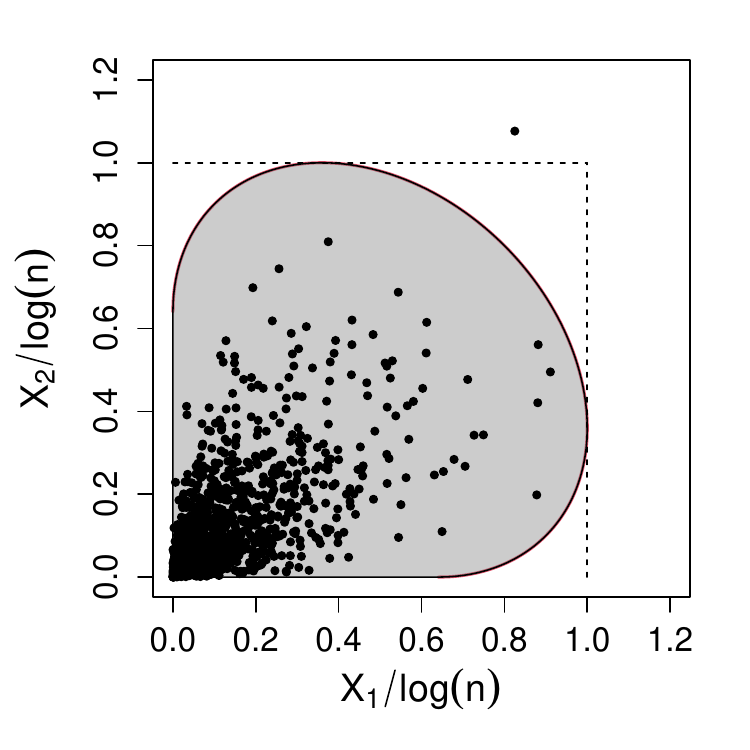}
\includegraphics[width=3.5cm]{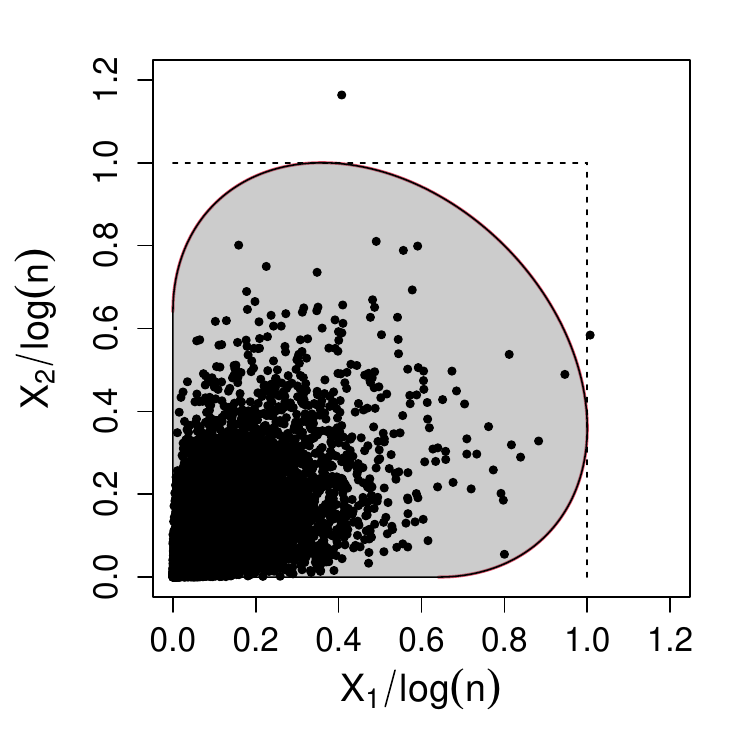}
\caption{\label{fig:SampleCloud} From left to right, the plots depict the scaled sample cloud with the limit set represented by a gray polygon, for different sample sizes, $n\in\cbr{100,1000,10000}$, respectively. The top row shows the convergence of the scaled sample cloud onto the limit set under the logistic dependence structure, while the bottom row represents the convergence under the Gaussian dependence structure.
Margins are exponential.
}
\end{figure}

\cite{balkema2010meta} derive necessary and sufficient conditions for the convergence of $N_n$ onto the limit set $G$.
However, as these conditions do not allow for simple determination of the limit set for a given distribution. \cite{nolde2014geometric} and \cite{nolde2022linking} consider a sufficient condition for the convergence in terms of the Lebesgue density of $\bX$, denoted by $f_{\bX}$.
Specifically, if
\begin{equation}
\label{eq:sufficient}
    \lim_{t\rightarrow\infty}\frac{-\log f_{\bX}(t\bx)}{t}=g_{\bX}(\bx),\quad \bx\in[0,\infty)^d,
\end{equation}
where $g_{\bX}(\bx)$ is a continuous gauge function, then $N_n$ converges onto $G$ defined by $g_{\bX}(\bx).$
An important property of $g_{\bX}(\bx)$ is homogeneity of order 1, i.e., $g_{\bX}(c\bx)=cg_{\bX}(\bx)$ for all $c>0$.
For exponential margins, the coordinatewise supremum of $G$ is $(1,\ldots,1)$, as   $\max_{i=1,\ldots,n}X_{ij}/\log{n}\xrightarrow{p}1$ for $j=1,\ldots,d.$
We will focus on statistical models derived from the sufficient condition~\eqref{eq:sufficient}, specifically in the bivariate case $d=2$.

\subsection{Model assumptions}
\label{sec:Assumption}

We assume that $(X,Y)$ has a joint Lebesgue density $f_{X,Y}$ with standard exponential margins.
The density assumption is standard for likelihood-based inference, while in practice the margins must be estimated and transformed to exponential via the probability integral transform.
This is a standard approach for extreme value dependence modeling.

We follow the statistical framework of \cite{wadsworth2024statistical}.
Using the radial-angular transformation $(X,Y)\mapsto(R,W)$ with $R=X+Y$ and $W=X/R,$ we assume that the conditional variable of $R\mid W=w$ follows a gamma distribution over a high threshold $r_\tau(w)$, i.e.,
\begin{equation}
\label{eq:tGam}
    R\mid \sbr{W=w, R>r_\tau(w)} \sim \text{truncGamma}\rbr{\lambda,g_{\bX}(w,1-w)},
\end{equation}
where $\lambda$ and $g_{\bX}(w,1-w)$ are the gamma shape and rate parameters, respectively, and $r_{\tau}(w)$ is a high quantile of the conditional distribution of $R\mid W$.
We refer to \cite{wadsworth2024statistical} for the theoretical justification of this model.
The gamma shape parameter corresponds to the dimension $\lambda=2$ for most underlying parametric models.
However, we allow for the estimation of this parameter to provide additional flexibility as in \cite{wadsworth2024statistical}.
The gauge function $g_{\bX}(w,1-w)$ can be estimated either parametrically or semi-parametrically \citep{majumder2025semiparametric,campbell2024piecewise} by fitting $g_{\bX}(w,1-w;\bm{\theta})$ under the truncated gamma approximation.
Our analysis focuses on particular classes of parametric models obtained through two types of additive mixing.
These provide flexibility in capturing diverse dependence structures, including both AD and AI forms, while retaining a relatively parsimonious form.

\section{Tail dependence coefficient under the truncated gamma model}
\label{sec:Theory}

The shape of the limit set $G$ describes various extremal dependence properties, including the potential presence or absence of asymptotic dependence.
In the bivariate case, \cite{balkema2010asymptotic} show that a ``blunt'' limit set, where $(1,1)\notin G,$ implies asymptotic independence between $X$ and $Y$.
However, the converse is not true in general: that is, $(1,1)\in G$ can be associated with either AD or AI.
The condition $(1,1)\in G$ alone is not sufficient to distinguish between the two extremal classes, implying further conditions are required.
\cite{BalkemaNoldeAD} addressed this problem by focusing on light-tailed {\em homothetic} densities and providing geometric criteria for AD through local and global conditions on the translated set $G-\bm{e}$ with $\bm{e}=\sup(G)$.

In this work, we focus on geometric criteria for AD under the truncated gamma model~\eqref{eq:tGam}. While there are similarities with the homothetic case considered in \citet{BalkemaNoldeAD}, our assumption is different: a homothetic density for $(X,Y)$ gives a prescribed form for the radial and angular variables, while we only prescribe a form for the distribution of $R|W=w$, allowing for a variety of angular densities $f_W(w)$. We find that it is possible for this density to influence the tail dependence coefficient.
To distinguish between the two extremal dependence classes, in the following section we define a ``pointy'' limit set.
Under the specific assumption of the truncated gamma model, we derive bounds on the tail dependence coefficient for a pointy limit set and show that, as long as the density $f_W(w)$ does not explode at key points away from the center, it implies AD.
In addition, we develop geometric criteria for AI, with these results formulated  in Section~\ref{sec:AI}.

We note a subtle but important point concerning our model. We assume the truncated gamma model to hold exactly on the region $\{(x,y) \in \mathbb{R}_+^2: (x+y)>r_\tau(x/(x+y))\}$; this region is illustrated in Figure~\ref{fig:rtau}, with an example $r_\tau(\cdot)$, the shape of which will depend on the dataset. Combining the truncated gamma radial model with an angular model $f_W$ on this region yields variables with assumed joint density $f_{X,Y}^\star(x,y)$, given in equation~\eqref{eq:tg_cartesian} below, which represents a model for the underlying true density $f_{X,Y}(x,y)$ on the region $\{(x,y) \in \mathbb{R}_+^2: (x+y)>r_\tau(x/(x+y))\}$. The joint density $f_{X,Y}^\star$ of our model may possess slightly different marginal behavior to the underlying margins, although differences will be very small at moderate levels. The root cause of this discrepancy is the assumption of the asymptotic form at finite levels, and there remains an open problem as to whether the modeling approach could be adapted to ensure equivalence of margins to the underlying ones. We will calculate dependence features of the model using the assumed density $f_{X,Y}^\star(x,y)$.

\begin{figure}[ht]
\centering
\includegraphics[width=0.35\textwidth]{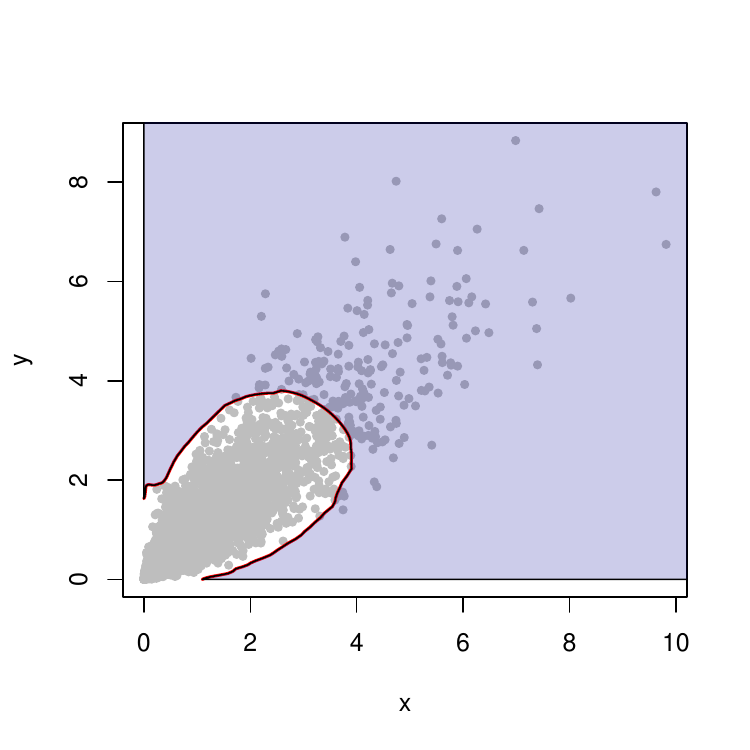}
\caption{Illustration of the model support. The threshold $r_\tau(w)$ is represented by the red line; the shaded region outside of this represents the support of the model.}
\label{fig:rtau}
\end{figure}

\subsection{Asymptotic dependence for a pointy limit set}
We focus on the bivariate case, assuming that the gauge function $g_{\bX}(x,y)$ is continuous for $(x,y)\in[0,\infty)^2$.
The key quantity in determining the extremal dependence class is the shape of the limit set, specifically the behavior of the gauge function along the boundary $\cbr{(x,y):\max(x,y)=1}.$
Since $\sup(G)=(1,1)$, so that $G\subseteq [0,1]^2$, we always have $g_{\bX}(x,y)\ge \max(x,y)$.
Define the gauge function values along this boundary through
\begin{equation}
\label{eq:k(q)}
k(q):=g_{\bX}(1,q)\quad\text{and}\quad \tilde{k}(q):=g_{\bX}(q,1),\quad q\ge 0.
\end{equation}
Note that
$k(q)\ge 1$ and $\tilde{k}(q)\ge 1$ since $g_{\bX}(x,y)\ge \max(x,y)$.

We define a ``pointy'' limit set as one for which $(1,1)\in G$ and $g_{\bX}(1,a)=g_{\bX}(b,1)=1$ hold at only finitely many boundary points $a, b\in [0,1)$, with no vertical and horizontal tangents at these points. Examples~\ref{ex:log} and~\ref{ex:indLog_min} represent such pointy limit sets, while Example~\ref{ex:invloglog} in Section~\ref{sec:AI} is not pointy because it fails the tangent condition.
\begin{definition}
\label{def:pointyset}
Define the set $D=\cbr{q\in[0,1]:k(q)=1}$ and $\Tilde{D}=\cbr{q\in[0,1]:\tilde{k}(q)=1}$.
The bivariate limit set $G=\cbr{(x,y)\in\reals_+^2:g_{\bX}(x,y)\le 1}$ on $[0,1]^2$ is termed a \emph{pointy limit set} if $k(1)=\tilde{k}(1)=1$ and the sets $D$ and $\Tilde{D}$ are finite and thus countable, with $k'(q_i^+),k'(q_i^-)\neq 0$, for all $q_i \in D$, and $\tilde{k}'(\tilde{q}_i^+),\tilde{k}'(\tilde{q}_i^-) \neq 0$, for all $\tilde{q}_i \in \tilde{D}$.

\end{definition}

\begin{example}
\label{ex:log}
Consider a bivariate logistic gauge function with a parameter $\gamma \in (0,1)$, given by
\begin{equation}
\label{eq:Logistic}
    g_{\bX}(x,y;\gamma)=\frac{1}{\gamma}\left(x+y\right)+\rbr{1-\frac{2}{\gamma}}\min\cbr{x,y},\quad (x,y)\in [0,\infty)^2.
\end{equation}
We visualise its limit set $G$ and the function of $k(q),$ $q\ge 0,$ for the logistic model in Figure~\ref{fig:Logistic}.
In this case, $D=\Tilde{D}=\cbr{1}$.
\end{example}
\begin{figure}[ht]
\centering
\includegraphics[width=4.5cm]{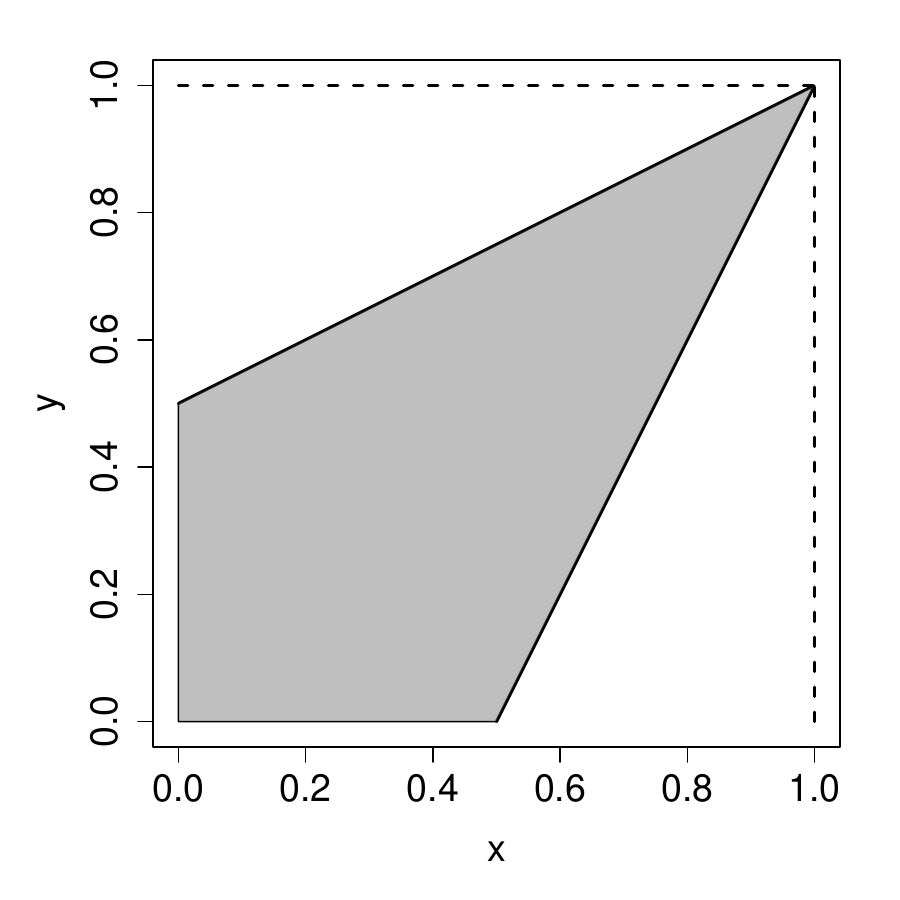}
\hspace{0.4cm}
\includegraphics[width=4.5cm]{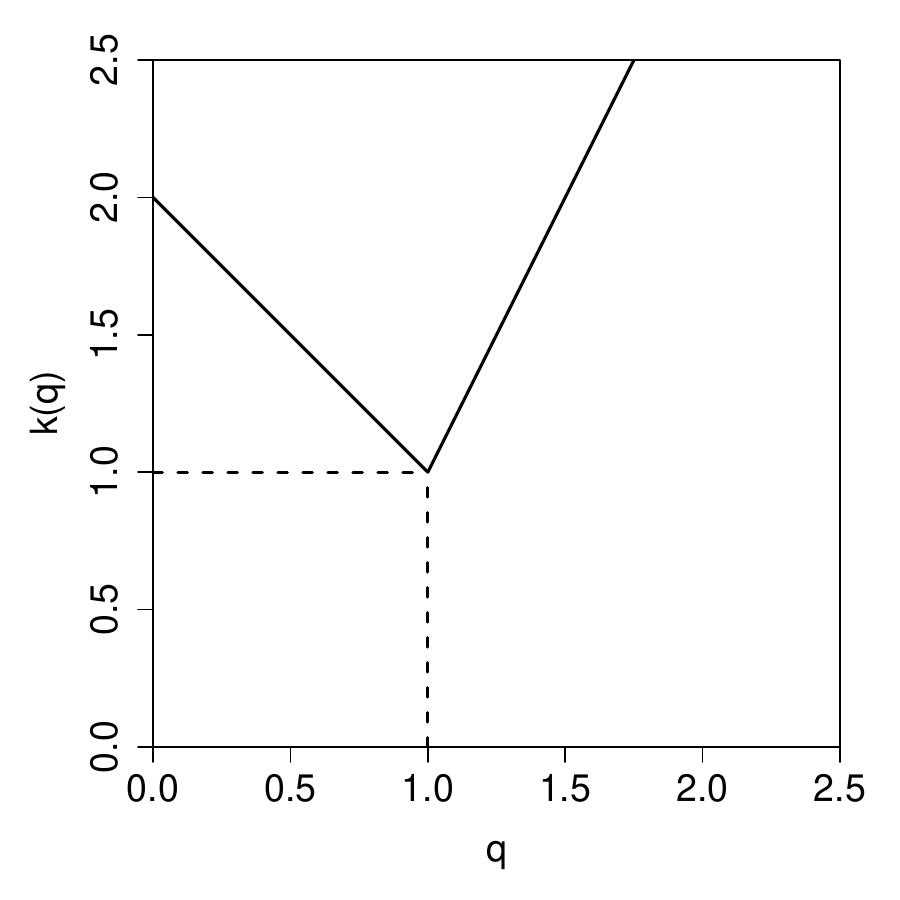}
\caption{\label{fig:Logistic} The shape of the limit set for the logistic model with $\gamma=0.5$ (Left) and the corresponding function $k(q),$ $q\ge 0,$ (Right).
}
\end{figure}

\begin{example}
\label{ex:indLog_min}
    Consider a simple mixing of independent and logistic gauge functions via minimization 
    \begin{equation*}
    \label{eq:IndLogistic}
        g_{\bX}(x,y;\gamma)=\min\cbr{x+y,\,\,\frac{1}{\gamma}\rbr{x+y}+\rbr{1-\frac{2}{\gamma}}\min\rbr{x,y}},\quad(x,y)\in[0,\infty)^2,\quad \gamma\in(0,1).
    \end{equation*}
    The shape of the limit set in the left panel of Figure~\ref{fig:mm_IndLog} retains the most protruding parts of the constituent limit sets. Its corresponding $k(q)$ function is shown in the right panel of Figure~\ref{fig:mm_IndLog}.
    In this case, we have $D=\Tilde{D}=\cbr{0,1}.$
\end{example}

\begin{figure}[ht]
\centering
\includegraphics[width=4.5cm]{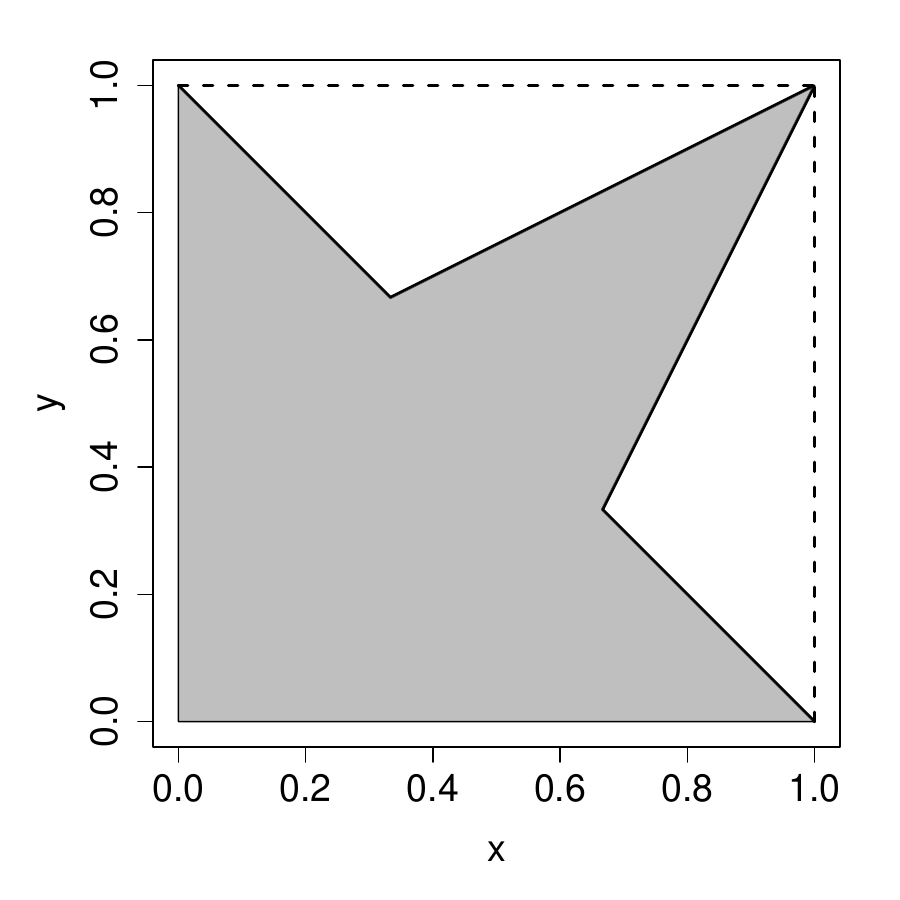}
\hspace{0.4cm}
\includegraphics[width=4.5cm]{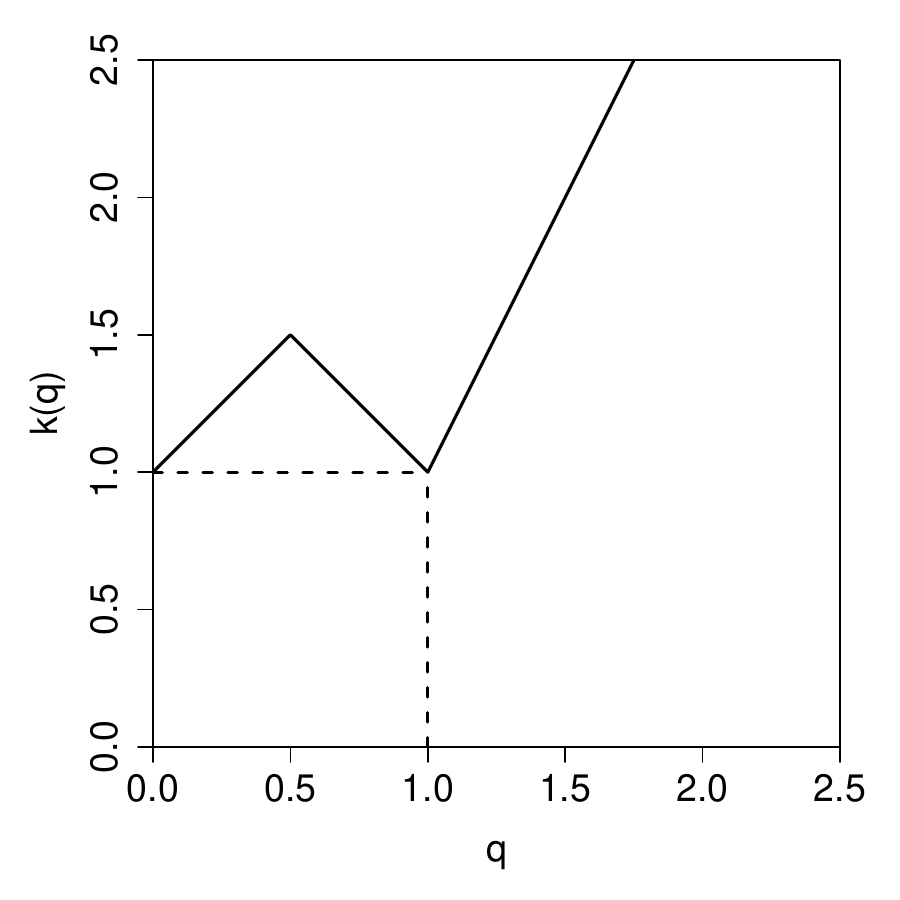}
\caption{\label{fig:mm_IndLog} (Left) the shape of the limit set for the mixture of independent and logistic gauge functions via minimization with $\gamma=0.5$. (Right) its corresponding function $k(q),$ $q\ge 0$.
}
\end{figure}

The pointy limit set may intersect the boundary set $\cbr{(x,y):\max(x,y)=1}$ at a finite number of points. 
Its essential feature is the presence of a sharp upper-right corner, with any other boundary intersections also being sharp.
The strength of asymptotic dependence for a pointy limit set is determined by both the number of boundary intersections and the magnitudes of the left- and right- derivative of $k(q)$ at the intersection points $q$, denoted as $k'(q^-)$ and $k'(q^+)$, respectively.
Similarly, the corresponding derivatives for $\tilde{k}(q)$ also contribute to the dependence strength.

In order to detail the result on tail dependence, we must consider the behavior of $k(q)$. Let $[0,1]=\bigcup_{i=1}^{n-1}[q_i,q_{i+1}]$ such that for any two adjacent points $q_i < q_{i+1}$, $k(q)$ on $[q_i,q_{i+1}]$ is  either bijective on $[q_i,q_{i+1}]$ or takes a constant value greater than 1, i.e., $k(q)=m > 1$ for some $m\in \reals.$
To classify these cases, we define three disjoint index sets, $A,B,C$:
\begin{align}
\label{eq:indexSets}
\begin{split}
    A&=\cbr{i:k(q_i)>1, k(q_{i+1})=1,\,\,\text{and}\,\,k(q)\searrow\,\,\text{on}\,\,[q_i,q_{i+1}]\,\,\text{for}\,\,i=1,\ldots,n-1}\\
    B&=\cbr{i:k(q_i)=1, k(q_{i+1})>1\,\, \text{and}\,\,k(q)\nearrow\,\,\text{on}\,\,[q_i,q_{i+1}]\,\,\text{for}\,\,i=1,\ldots,n-1}\\
    C&=\cbr{i:k(q)=m>1\,\,\text{on}\,\,[q_i,q_{i+1}]\,\,\text{for}\,\,i=1,\ldots,n-1}.
\end{split}
\end{align}


 The sets $\Tilde{A}$, $\Tilde{B}$, and $\Tilde{C}$, are defined as above in relation to $\tilde{k}(q)$. The sets $A$ and $B$, where $k(q)$ intersects the boundary, contribute directly to the tail dependence coefficient.
In contrast, the region in $C$, where $k(q)$ takes a constant value greater than 1, does not intersect the boundary and, therefore, do not contribute to the tail dependence coefficient. 
Assumption~\ref{ass:kq} provides a mild regularity condition on $k$, and its counterpart $\tilde{k}$. 
\begin{assumption}
\label{ass:kq}
\begin{enumerate}[label=(\roman*)]
    \item The functions $k(q)=g_{\bX}(1,q)$ and $\tilde{k}(q) = g_{\bX}(q,1)$ are continuous and three times (piecewise) continuously differentiable, i.e., possess left- and right-hand derivatives of up to order three for $q \in [0,\infty)$.
    \item If, for $i \in A$, $k(q_{i+1})=1$ and $k'(q_{i+1}^-)=0$, then $k(q)-1=a_k(q_{i+1}-q)^{\rho_k}+b_k(q_{i+1}-q)^{\rho_k+\epsilon}+\smallo\rbr{(q_{i+1}-q)^{\rho_k+\epsilon}}$ as $q\rightarrow q_{i+1}^-$ for some $a_k>0$, $b_k\in\mathbb{R}$, $\rho_k>1$, and $\epsilon>0$. If, for $i \in B$, $k(q_{i})=1$ and $k'(q_{i}^+)=0$, then $k(q)-1=a_k(q-q_i)^{\rho_k}+b_k(q-q_i)^{\rho_k+\epsilon}+\smallo\rbr{(q-q_i)^{\rho_k+\epsilon}}$ as $q\rightarrow q_{i}^+$ for some $a_k>0$, $b_k\in\mathbb{R}$, $\rho_k>1$, and $\epsilon>0$. 
    The constants $a_k,b_k,\rho_k,\epsilon$ may be different at each intersection point.
    The same conditions hold for $\tilde{k}$ with $i \in \tilde{A}, i\in\tilde{B}$.
\end{enumerate}
\end{assumption}

\begin{remark}
    By Definition~\ref{def:pointyset}, a pointy limit set does not possess points such that $k(q_i)=1$ and $k'(q_i^{\pm})=0$. However, Assumption~\ref{ass:kq}~(ii) will be useful in Proposition~\ref{prop:AI}.
\end{remark}

Consider a bivariate random vector $(X,Y)$ with exponential margins and possessing a density function.
Under the truncated gamma assumption for $f_{R|W}(r|w)$ in~\eqref{eq:tGam} over a high threshold $r_\tau\rbr{w}$,
we define the bivariate joint density 
\begin{equation}
\label{eq:tg_cartesian}
    f^\star_{X,Y}(x,y)=f_{R|W}\rbr{x+y\,\Bigg\rvert\,\frac{x}{x+y}}f_W\rbr{\frac{x}{x+y}}\frac{1}{x+y},
\end{equation}
with support $\cbr{(x,y):x+y > r_\tau\rbr{\frac{x}{x+y}}}$. 
Note that $f_W$ is unknown in general, though could be estimated.
Letting $y=xq,$ $q>0,$ the density \eqref{eq:tg_cartesian} can be further expressed in terms of $k(q)$ and $b(q)$
\begin{equation}
\label{eq:tg_kqbq}
    f^\star_{X,Y}(x,xq)=x^{\lambda-1} \exp\cbr{-xk(q)}b(q),
\end{equation}
with 
\[
    b(q):=\frac{\frac{1}{\Gamma(\lambda)}(1+q)^{-2}k(q)^\lambda f_W\rbr{\frac{1}{1+q}}}{\widebar{F}_{R|W}\rbr{r_\tau\rbr{\frac{1}{1+q}};\lambda,g_{\bX}\rbr{\frac{1}{1+q},\frac{q}{1+q}}}},
\]
where $\widebar{F}_{R|W}\rbr{\,\cdot\,;\lambda,g_{\bX}}$ is the gamma survival function.
For the detailed derivation of \eqref{eq:tg_kqbq}, refer to Section~\ref{appendix:proofThm3.2.} of the Appendix. We require a further regularity assumption on $b$, and its counterpart $\tilde{b}$, which comes from considering $f^\star_{X,Y}(yq,y) = y^{\lambda-1}\exp\{-y\tilde{k}(q)\}\tilde{b}(q)$ in place of~\eqref{eq:tg_kqbq}.

\begin{assumption}
    \label{ass:b}
    The functions $b(q), \tilde{b}(q)$, are continuous, finite and twice (piecewise) continuously differentiable for all $q \in [0,\infty)$.
\end{assumption}

\begin{proposition}
\label{prop:pointyAD}
Suppose that Assumptions~\ref{ass:kq} and~\ref{ass:b} hold. Under the truncated gamma formulation \eqref{eq:tg_cartesian}, a pointy limit set $G$, as specified in Definition~\ref{def:pointyset}, implies that the tail dependence lies between $\chi_{\text{lower}}$ and $\chi_{\text{upper}}$, with
\begin{align*}
    \chi_{\text{lower}} & = \frac{b(1)/k'(1^+) + \tilde{b}(1)/\tilde{k}'(1^+)}{ \max\left[- \sum_{i \in A} \frac{b(q_{i+1})}{k'(q_{i+1}^-)} +\sum_{i\in B} \frac{b(q_{i})}{k'(q_{i}^+)} + \frac{b(1)}{k'(1^+)} , - \sum_{i \in \tilde{A}} \frac{\tilde{b}(\tilde{q}_{i+1})}{\tilde{k}'(\tilde{q}_{i+1}^-)} +\sum_{i\in\tilde{B}} \frac{\tilde{b}(\tilde{q}_{i})}{\tilde{k}'(\tilde{q}_{i}^+)} + \frac{\tilde{b}(1)}{\tilde{k}'(1^+)} \right]},\\
    \chi_{\text{upper}} & = \frac{b(1)/k'(1^+) + \tilde{b}(1)/\tilde{k}'(1^+)}{ \min\left[- \sum_{i \in A} \frac{b(q_{i+1})}{k'(q_{i+1}^-)} +\sum_{i\in B} \frac{b(q_{i})}{k'(q_{i}^+)} + \frac{b(1)}{k'(1^+)} , - \sum_{i \in \tilde{A}} \frac{\tilde{b}(\tilde{q}_{i+1})}{\tilde{k}'(\tilde{q}_{i+1}^-)} +\sum_{i\in\tilde{B}} \frac{\tilde{b}(\tilde{q}_{i})}{\tilde{k}'(\tilde{q}_{i}^+)} + \frac{\tilde{b}(1)}{\tilde{k}'(1^+)} \right]},
\end{align*}
where the sets $A$ and $B$ are defined in~\eqref{eq:indexSets} with their counterparts $\tilde{A}$ and $\tilde{B}$ defined similarly.
\end{proposition}


The proof is detailed in Section~\ref{appendix:proofThm3.2.} of the Appendix.

\begin{remark}
$\chi_{\text{lower}}>0$, implying AD. In the event that $k=\tilde{k}$ and $b=\tilde{b}$, then the lower and upper bounds are equal, which also holds under the slightly less restrictive condition of equality at the boundary intersection points. When the upper and lower bounds are equal then
\begin{align*}
    \chi & = \frac{2b(1)/k'(1^+)}{-\sum_{i \in A} \frac{b(q_{i+1})}{k'(q_{i+1}^-)} +\sum_{i\in B} \frac{b(q_{i})}{k'(q_{i}^+)} + \frac{b(1)}{k'(1^+)}}.
\end{align*}
\end{remark}


\begin{example}
    Recall the bivariate logistic gauge function in~\eqref{eq:Logistic}.
    The tail dependence coefficient corresponds to the case where $|A|=|\Tilde{A}|=1$ with $q_1=0$ and $q_2=1$, $|B|=|\Tilde{B}|=0$, and $k(q)=\tilde{k}(q)$ due to symmetry. Since we always have $b(1)=\tilde{b}(1)$, we have the precise value of $\chi$:
\begin{align*}
    \chi
    =
    \frac{2k'(1^-)}{k'(1^-)-k'(1^+)} 
    =
    \frac{2-2\gamma}{2-\gamma} \in (0,1),
\end{align*}
where $k(1)=1$, $k(0)=1/\gamma >1 $, $k'(1^+)=1/\gamma$, and $k'(1^-)=1-1/\gamma$.
\end{example}

We compare the tail dependence coefficient based on the geometric approach to the theoretical one for the logistic distribution, $\chi=2-2^\gamma$, in the right panel of Figure~\ref{fig:Chi_Log}.
The small differences stem from the fact that the modeling frameworks are distinct.
However, in practice, we estimate $\chi$ using simulation-based techniques.
The results displayed in Section~\ref{sec:Simulation} show that we can estimate $\chi$ well for a range of $\gamma$ values under the geometric framework.

\begin{figure}[ht]
\centering
\includegraphics[width=4.5cm]{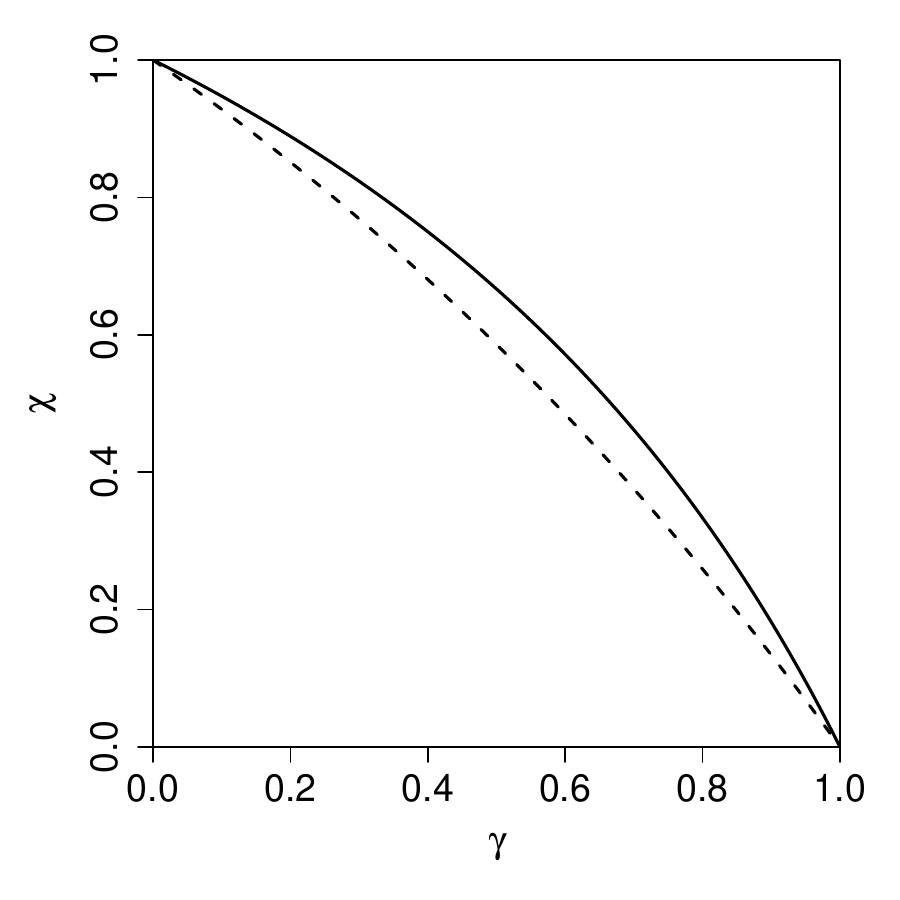}
\caption{\label{fig:Chi_Log} The comparison of the tail dependence coefficient derived from the geometric approach (solid line) with the theoretical value (dashed line).
}
\end{figure}

\begin{example}
\label{ex:indLog_min_chi}
    For the mixture of independent and logistic gauge functions in Example~\ref{eq:IndLogistic}, we have $|A|=|\Tilde{A}|=1$ with $q_1=1-\gamma$ and $q_2=1$, $|B|=|\Tilde{B}|=1$ with $q_1=0$ and $q_2=1-\gamma$, and $k(q)=\tilde{k}(q)$ due to symmetry.
    The tail dependence coefficient under the truncated gamma approximation has bounds
    \begin{align*}
       \chi_{\text{lower}} &= \frac{2b(1)/k'(1^+)}{\frac{b(1)}{k'(1^-)} + \frac{b(1)}{k'(1^+)}+ \frac{\max(b(0),\tilde{b}(0))}{k'(0^+)}}\\ 
              \chi_{\text{upper}} &= \frac{2b(1)/k'(1^+)}{\frac{b(1)}{k'(1^-)} + \frac{b(1)}{k'(1^+)}+ \frac{\min(b(0),\tilde{b}(0))}{k'(0^+)}} ,
    \end{align*}
    where $k(0)=k(1)=1$, 
    $k'(0^+)=1$, $k'(1^+)=1/\gamma$, and $k'(1^-)=1-1/\gamma.$
    In Proposition~\ref{prop:pointyAD}, we assume $b(0)<\infty$, so $\chi>0$. In Proposition~\ref{prop:AI} below, we will allow $b(0)$ and $\tilde{b}(0)$ to diverge, which happens when $f_W(1)$ and $f_W(0)$ diverge. This causes the tail dependence coefficient to be zero.
\end{example}


\subsection{Asymptotic independence}
\label{sec:AI}

We continue to focus on the boundary case, where $k(1)=\tilde{k}(1)=1$, or $(1,1)\in G$, and define criteria for asymptotic independence.
One key feature of a limit set indicating asymptotic independence is the presence of a line segment intersecting this boundary; equivalently, $g_{\bX}(1,a)=g_{\bX}(b,1)=1$ at infinitely many points $a$ and $b$ in $[0,1)$. Another feature of $G$ associated with AI is the presence of vertical and horizontal tangents along the boundary $\cbr{(x,y):\max(x,y)=1}$.
In this case, the contribution of the pointy part at $(1,1)$ becomes negligible compared to that of the smooth tangent at the boundary point shown in Figure~\ref{fig:invLogLog_min}. See also Figure 1 (b) and (c) in~\cite{BalkemaNoldeAD}.
Finally, if $b(q)$ and $\tilde{b}(q)$ diverge at an intersection point of the boundary, this also yields AI. Divergence of $b, \tilde{b}$ is driven by divergence of $f_W$, for which a sensible assumption is to allow explosion only at the endpoints $\{0,1\}$. We formalize the relevant conditions in Proposition~\ref{prop:AI}, proven in Section~\ref{appendix:proofAI} of the Appendix.



\begin{proposition}
\label{prop:AI}
Suppose that Assumption~\ref{ass:kq} holds, and that $(1,1)\in G$. Under the truncated gamma formulation~\eqref{eq:tg_cartesian}, we have asymptotic independence of $(X,Y)$ if
    any of the following conditions apply:
    \begin{enumerate}
        \item Assumption~\ref{ass:b} holds and there exist sets $D,\Tilde{D}\subseteq[0,1]$ (not necessarily identical) with $|D|$, $|\tilde{D}|>0$ such that $k(q)=1$ for all $q \in D$, and $\tilde{k}(q)=1$ for all $q \in \Tilde{D}$;
        \item Assumption~\ref{ass:b} holds and there exist boundary intersection points $q_i,\tilde{q}_i\in[0,1]$ (not necessarily identical), with $b(q_i), \tilde{b}(\tilde{q}_i) \neq 0$, such that 
        \begin{align*}
        \begin{cases}
            k'(q_{i+1}^-)=0,\quad \tilde{k}'(\tilde{q}_{i+1}^-)=0,\quad i \in A,\\
            k'(q_{i}^+)=0,\quad \tilde{k}'(\tilde{q}_i^+)=0,\quad i\in B,
        \end{cases}
        \end{align*}
        implying vertical and horizontal tangents along the boundary $\cbr{(x,y):\max(x,y)=1}$. 
        \item  If $k(0)=\tilde{k}(0)=1$, Assumption~\ref{ass:b} holds, except $b$, $\tilde{b}$ are infinite at $0$, with $b(q)\sim a_b q^{\rho_b}$
        as $q \to 0$, with $a_b>0$ and $\rho_b \in (-1,0)$ (similarly for $\tilde{b}$).
    \end{enumerate}
\end{proposition}

\begin{example}
\label{ex:invloglog}
Consider a mixing of inverted-logistic and logistic gauge functions via minimization 
    \begin{equation*}
    \label{eq:InvLogLog_min}
        g_{\bX}(x,y;\gamma)=\min\cbr{\rbr{x^{1/\theta}+y^{1/\theta}}^\theta,\,\,\frac{1}{\gamma}\rbr{x+y}+\rbr{1-\frac{2}{\gamma}}\min\rbr{x,y}},
    \end{equation*}
    for $(x,y)\in[0,\infty)^2$ and $(\theta,\gamma)\in(0,1]\times(0,1).$
    The shape of the limit set and its corresponding $k(q)$ are illustrated in Figure~\ref{fig:invLogLog_min}.
    Note that when $\theta=1$, this mixture reduces to the mixing of independent and logistic gauge functions obtained via minimization in Example~\ref{ex:indLog_min}.
    For $\theta\in(0,1)$, a vertical tangent arises at the boundary point $q=0$, with $k'(0^+)=0$.  For the inverted-logistic gauge function, we have $k(q)= \tilde{k}(q)=(1+q^{1/\theta})^\theta$, with $\theta\in(0,1)$. Since $k(q)-1=\theta q^{1/\theta}+(\theta(\theta-1)/2)q^{2/\theta} + \bigO(q^{3/\theta})$ as $q\rightarrow  0^+$, $k(q)-1 \sim \theta q^{1/\theta}$, we have $a_k=\theta>0$, $b_k=\theta(\theta-1)/2$,  $\rho_k = 1/\theta>1$ and $\epsilon=\rho_k$ satisfying Assumption~\ref{ass:kq}~(ii).
    \end{example}
\begin{remark}
Informally, we can view the formula for the tail dependence coefficient to be the same as in Example~\ref{ex:indLog_min_chi}, with $k'(0^+)=0$ yielding $\chi=0$ for $b(0),\tilde{b}(0)>0$.
\end{remark}

\begin{figure}[ht]
\centering
\includegraphics[width=4.5cm]{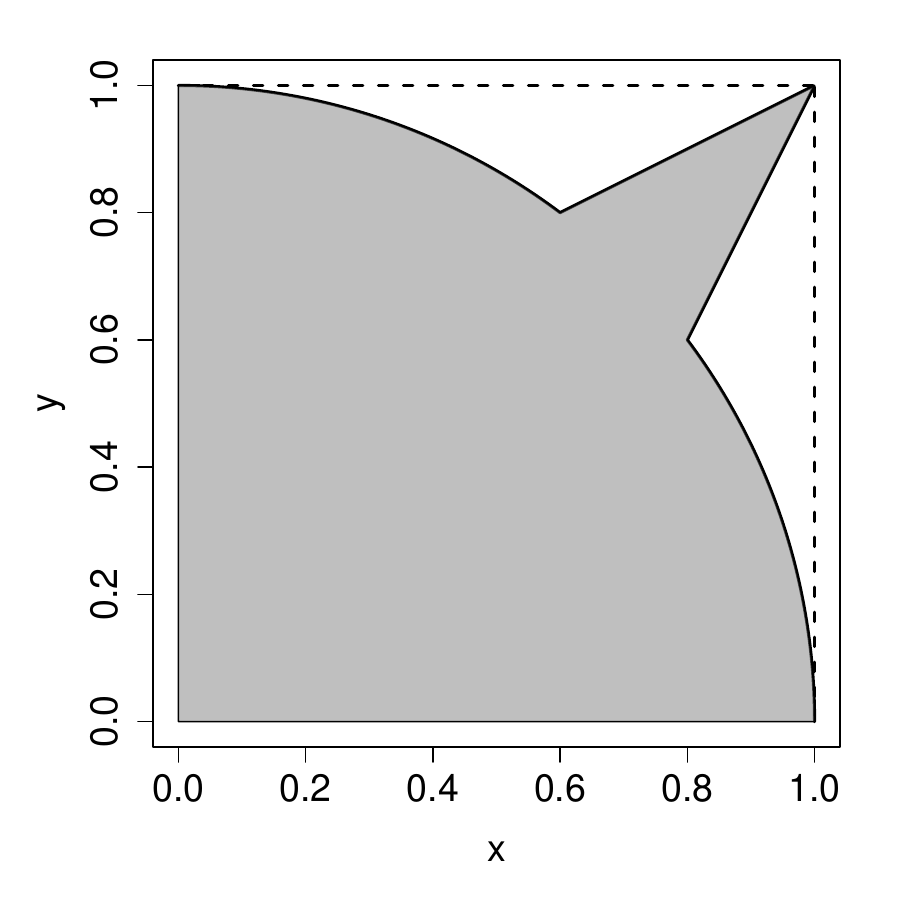}
\hspace{0.4cm}
\includegraphics[width=4.5cm]{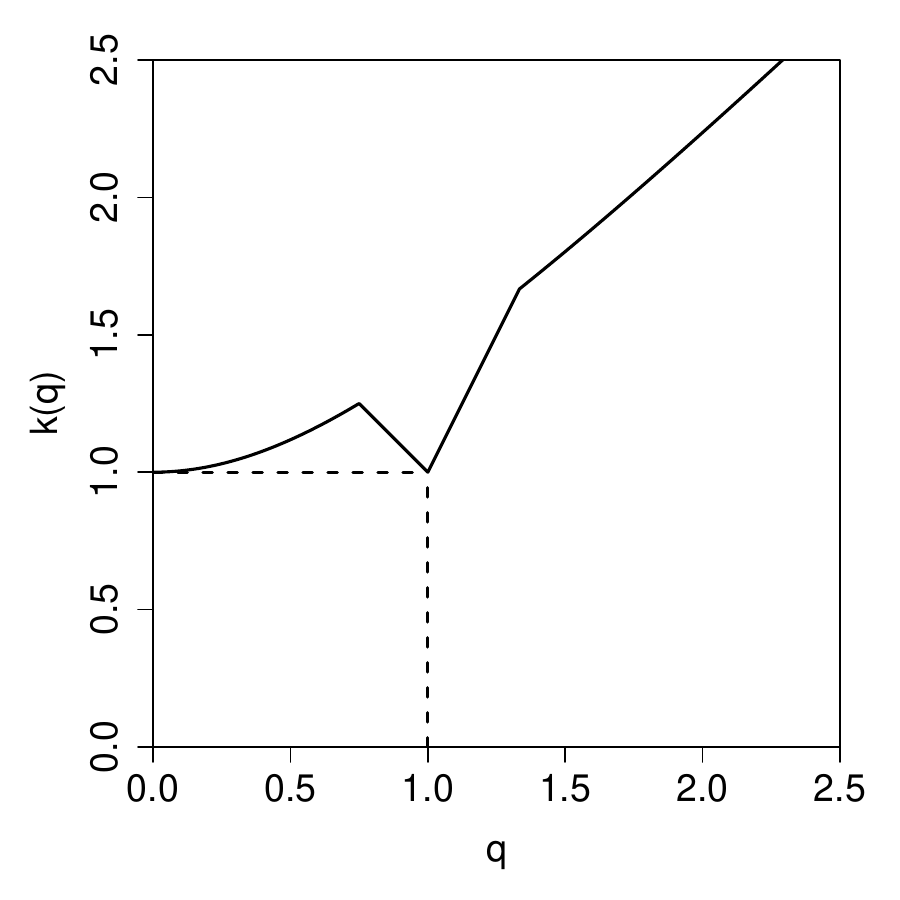}
\caption{\label{fig:invLogLog_min} (Left) the shape of the limit set for the mixture of inverted-logistic and logistic gauge functions via minimization with $(\theta,\gamma)=(0.5,0.5)$.
(Right) its corresponding function $k(q),$ $q\ge 0$.
}
\end{figure}

\section{Additive constructions}
\label{sec:additive}
In Section~\ref{sec:Amix} and \ref{sec:HWmodel}, we study two constructions for gauge functions that allow both AD and AI.
\subsection{Additively mixed gauge functions}
\label{sec:Amix}

A limit set $G$, characterized by a specific parametric gauge function $g_{\bX}(x,y)$, reflects its unique shape and corresponding dependence structure.
\cite{wadsworth2024statistical} suggested additively mixing parametric gauge functions to produce more flexible forms.
They noted the ability of such constructions to exhibit both pointy and blunt forms, so we explore this further, focusing on two component mixtures, defined as follows: 
\begin{equation*}
\label{eq:mixft}
    g_{\bX}^\star(x,y;\theta_1,\theta_2,p)
    =
    pg_{\bX}^{[1]}(x,y;\theta_1)+(1-p)g_{\bX}^{[2]}(x,y;\theta_2),\quad p \in (0,1), \,\, (x,y)\in [0,\infty)^2.
\end{equation*}
In general, the set $G^\star$ defined by $g_{\bX}^\star$ does not satisfy the required scaling condition $\sup(G^\star)=(1,1).$
Therefore, we need to determine the coordinatewise supremum $\sup(G^\star)=(c_1^\star,c_2^\star)$ and define the rescaled additively mixed model as $g_{\bX}(x,y):=g_{\bX}^\star(c_1^\star x,c_2^\star y).$
Expressing the gauge function in terms of angular components simplifies the task of finding the coordinatewise supremum due to the homogeneity property.
Specifically, plotting the points $(v_1,v_2):=(w,1-w)/g_{\bX}^\star(w,1-w)$ for $w \in [0,1]$ results in unit level sets of $g_{\bX}^\star$.
Using this observation, the coordinatewise supremum of the limit set $G^\star$ can be expressed as 
\[
c_1^\star=\max_{w\in[1/2,1]}\frac{w}{g_{\bX}^\star(w,1-w)},\quad c_2^\star=\max_{w\in[0,1/2]}\frac{1-w}{g_{\bX}^\star(w,1-w)}.
\]

While any gauge function in principle can be considered, we facilitate a transition between AD and AI by combining the logistic gauge function in~\eqref{eq:Logistic} with three other gauge functions that exhibit asymptotic independence.
The pointy limit shape of the logistic model is the fundamental feature shown in most parametric models exhibiting asymptotic dependence.
We denote $g_{\bX}^{[2]}(x,y;\gamma)$ as a bivariate logistic gauge function throughout.
Our goals are to derive closed-form expressions for rescaled additive mixture gauge functions and to establish connections between these functions and different extremal dependence measures, linking their parameters to the geometric criteria outlined in Section~\ref{sec:Theory}.

\subsubsection{Finding coordinatewise suprema analytically for specified classes}
\label{sec:findingcoordinatewise}
We aim to identify the coordinatewise supremum analytically for specified classes of additive mixtures.
We assume that both $g_{\bX}^{[1]}$ and $g_{\bX}^{[2]}$ are continuous and symmetric; it follows that $c_1^\star=c_2^\star=:c^\star$.
By expressing $g_{\bX}^\star$ in angular coordinates, we maximize the following objective function
\begin{equation}
\label{eq:s(w)}
    s(w):=\frac{1-w}{g_{\bX}^\star(w,1-w)}=\frac{1}{g_{\bX}^\star\rbr{\frac{w}{1-w},1}}
    =\frac{1}{\tilde{k}^\star\rbr{\frac{w}{1-w}}},\quad w\in[0,1/2],
\end{equation}
where $\tilde{k}^\star\rbr{\cdot}$ is defined as in~\eqref{eq:k(q)}, but for the function $g_{\bm{X}}^\star$. The domain restriction to $w\in[0,1/2]$ follows from symmetry. We observe from equation~\eqref{eq:s(w)} that maximizing $s$ over $[0,1/2]$ is equivalent to minimizing $\tilde{k}^\star$ over $[0,1]$. We will consider examples where both $g_{\bX}^{[1]}$ and $g_{\bX}^{[2]}$ are convex, from which it follows that $g_{\bX}^\star$ and $\tilde{k}^\star$ are convex. We denote the minimizer of $\tilde{k}^\star$ over the interval $[0,1]$ by $\kappa$; we thus have
\begin{equation*}
    c^\star = \frac{1}{\tilde{k}^\star(\kappa)} = \frac{1}{g_{\bX}^\star(\kappa,1)}, 
\end{equation*}
and hence
\begin{equation}
    g_{\bX}(x,y) = \frac{g_{\bX}^\star(x,y)}{g_{\bX}^\star(\kappa,1)}. \label{eq:govergstar}
\end{equation}
From equation~\eqref{eq:govergstar}, we observe $g_{\bX}(\kappa,1)=1$ and hence if there is a unique value of $\kappa$, then $\kappa$ also represents the slope in the conditional extremes model \citep{heffernan2004conditional,nolde2022linking}.

We now explore specific instances of additively mixed gauge functions, recalling that in each case $g_{\bX}^{[2]}$ is taken as a logistic gauge function. For detailed calculations of the derivative and the determination of the coordinatewise supremum, refer to Section 1 of the supplementary material.

\paragraph{Additive mixture of Gaussian and logistic gauge functions}
\label{sec:Ga&log}

We consider a bivariate Gaussian gauge function with a correlation parameter $\rho\in [0,1)$ for $g_{\bX}^{[1]}(x,y;\rho)$, defined as
\begin{equation*}
    g_{\bX}^{[1]}(x,y;\rho)=\left(x+y-2\rho(x y)^{1/2}\right)/(1-\rho^2),\quad (x,y)\in [0,\infty)^2.
\end{equation*}
The unit level sets of the rescaled mixture gauge function, $g_{\bX}(x,y)=1$, are shown in Figure~\ref{fig:Amix}, displaying two representative classes of extremal dependence: AI in the first panel of the top row and AD in the second panel of the top row.
In this case we find that
\begin{align*}
    \kappa=
    \begin{cases}
   \rho^2\left(1-\frac{1-\rho^2}{p}(1-p)\left(\frac{1}{\gamma}-1\right)\right)^{-2} \quad\quad\text{if}\,\,\rho\in(0,1)\,\, \text{and}\,\,(1+\rho)\rbr{\frac{1-p}{p}}\rbr{\frac{1-\gamma}{\gamma}}<1\\
    0 ,\quad\quad\text{if}\,\,\rho=0\,\,\text{and}\,\,\gamma\ge 1-p\\
    1 , \quad\quad\text{if}\,\,\rho=0\,\,\text{and}\,\,\gamma<1-p\,\,\text{or}\,\,\rho\in(0,1)\,\,\text{and}\,\,(1-\rho)\le(1-\rho^2)\rbr{\frac{1-p}{p}}\rbr{\frac{1-\gamma}{\gamma}}<1.
    \end{cases}
\end{align*}

\begin{figure}[ht!]
\centering
\includegraphics[width=4cm]{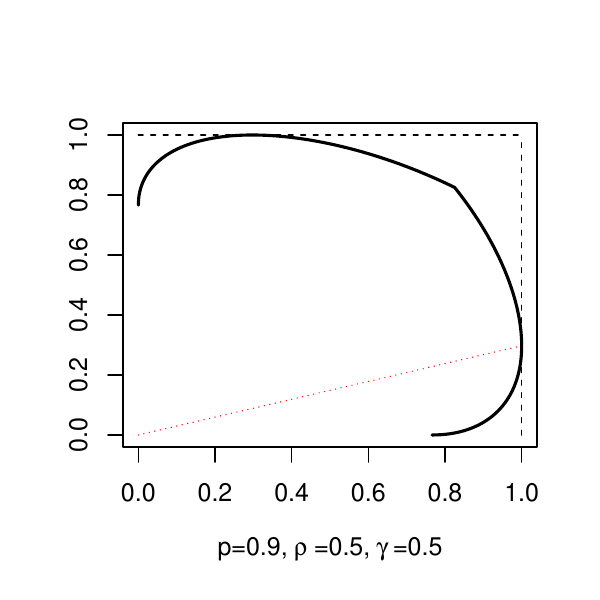}
\includegraphics[width=4cm]{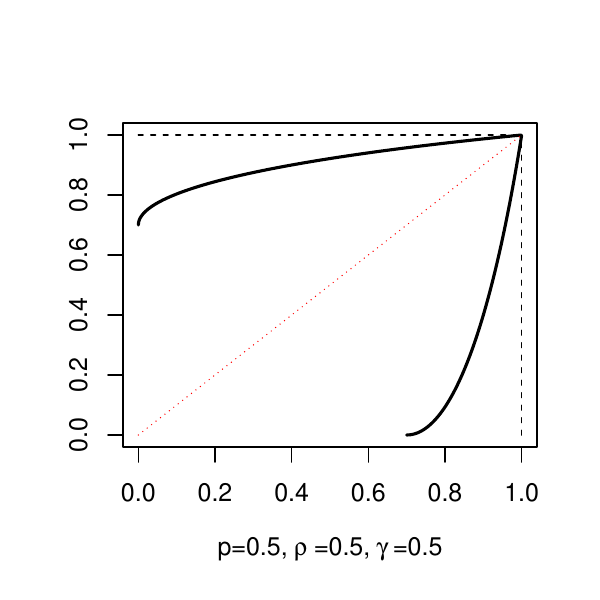}
\includegraphics[width=4cm]{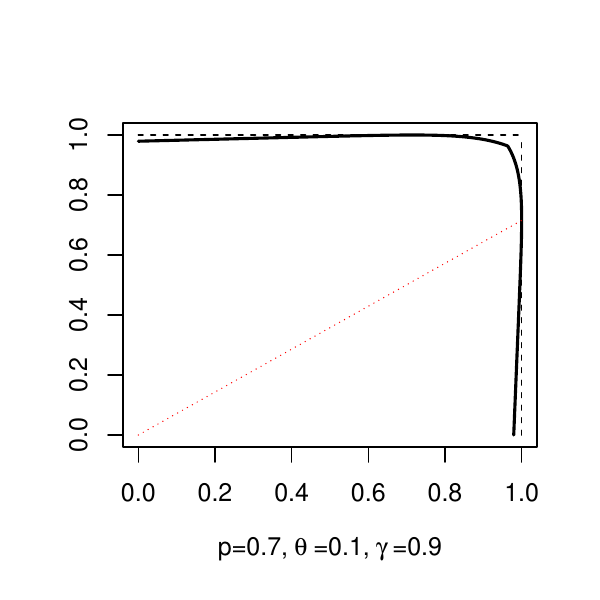}
\\
\vspace{-1cm}
\includegraphics[width=4cm]{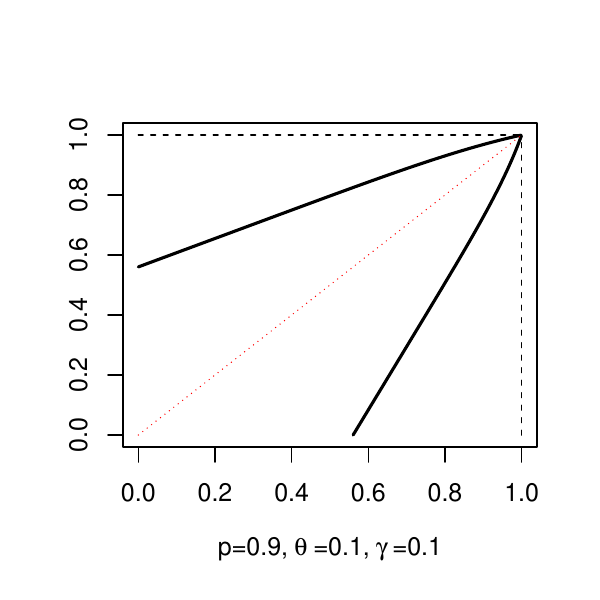}
\includegraphics[width=4cm]{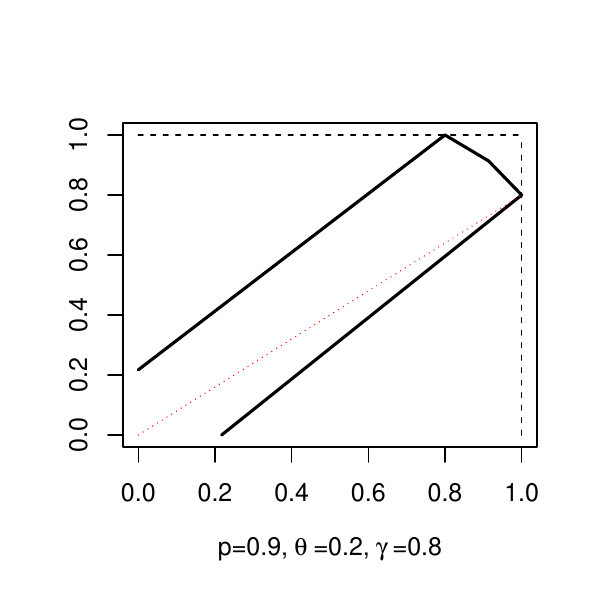}
\includegraphics[width=4cm]{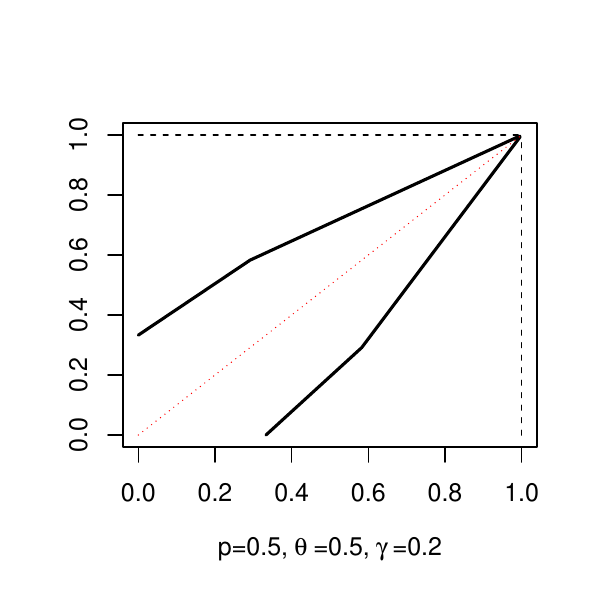}
\caption{\label{fig:Amix} Unit level sets of additively mixed gauge functions for selected parameter values.
Top row: Gaussian and logistic mixture with $(p,\rho,\gamma)=(0.9,0.5,0.5)$ (left) and $(0.5,0.5,0.5)$ (middle); inverted-logistic and logistic mixture with $(p,\theta,\gamma)=(0.7,0.1,0.9)$ (right).
Bottom row: inverted-logistic and logistic mixture with $(p,\theta,\gamma)=(0.9,0.1,0.1)$ (left); rectangle and logistic mixture with $(p,\theta,\gamma)=(0.9,0.2,0.8)$ (middle) and $(0.5,0.5,0.2)$ (right).
The red dotted line indicates the minimizer $\kappa$ or the slope in the conditional extremes model.
}
\end{figure}

\paragraph{Additive mixture of Inverted-logistic and logistic gauge functions}
\label{sec:Invlog&log}

Consider the bivariate inverted logistic gauge function $g_{\bX}^{[1]}(x,y)$ with $\theta\in (0,1]$, defined as
\begin{equation*}
    g_{\bX}^{[1]}(x,y)=\left(x^{1/\theta}+y^{1/\theta}\right)^{\theta},\quad (x,y)\in [0,\infty)^2.
\end{equation*}
The unit level sets of the rescaled mixture gauge function are shown in Figure~\ref{fig:Amix}, illustrating two representative classes of extremal dependence: AI in the third panel of the top row and AD in the first panel of the bottom row. In this case we have
\begin{align*}
    \kappa=
    \begin{cases}
   \left[\left\{\left(\frac{1-p}{p}\right)\left(\frac{1}{\gamma}-1\right)\right\}^{\frac{1}{\theta-1}}-1\right]^{-\theta}&\quad\text{if}\,\,\theta\in(0,1)\,\,\text{and}\,\,\rbr{\frac{1-\gamma}{\gamma}}\rbr{\frac{1-p}{p}}<2^{\theta-1}\\
    0 ,&\quad\text{if}\,\,\theta=1\,\,\text{and}\,\,\gamma\ge 1-p\\
    1 , &\quad\text{if}\,\,\theta=1\,\,\text{and}\,\,\gamma<1-p\,\,\text{or}\,\,\theta\in(0,1)\,\,\text{and}\,\,\rbr{\frac{1-\gamma}{\gamma}}\rbr{\frac{1-p}{p}}\ge 2^{\theta-1}.
    \end{cases}
\end{align*}

\paragraph{Additive mixture of rectangular and logistic gauge functions}
\label{sec:Rect&log}

A bivariate rectangle gauge function with a parameter $\theta \in (0,1]$ is given by
\begin{equation*}
    g_{\bX}^{[1]}(x,y)=\max\left(\frac{1}{\theta}(x-y),\frac{1}{\theta}(y-x),\frac{1}{2-\theta}(x+y)\right),\quad (x,y)\in [0,\infty)^2.
\end{equation*}
Figure~\ref{fig:Amix} shows the unit level sets of the rescaled additively mixed gauge function.
The bottom row displays two distinct classes of extremal dependence: AI in the second panel and AD in the third panel. We have for this case
\begin{align*}
    \kappa=
    \begin{cases}
    1-\theta, &\quad\text{if}\,\,\theta\in(0,1)\,\,\text{and}\,\,\frac{p}{2-\theta}+(1-p)(1-1/\gamma)>0\\
    0, &\quad\text{if}\,\,\theta=1\,\,\text{and}\,\,\gamma\ge 1-p\\
    1,
    &\quad\text{if}\,\,\theta=1\,\,\text{and}\,\,\gamma<1-p\,\,\text{or}\,\,\theta\in(0,1)\,\,\text{and}\,\,\frac{p}{2-\theta}+(1-p)(1-1/\gamma)\le 0.
    \end{cases} 
\end{align*}


\subsection{Additively mixing independent exponential random vectors}
\label{sec:HWmodel}

We now consider an alternative way to develop parametric gauge functions via an additive structure, but this time there is an additive stochastic representation, rather than additive gauge functions.
\cite{nolde2022linking} consider the spatial model proposed by \cite{huser2019modeling} as a case study, exploring its properties in terms of the limit set and the associated gauge function.
Following a similar approach, we consider the bivariate model constructed as:
\begin{align}
\label{eq:HW_amix}
    \bY=
    \begin{cases}
        \gamma S\bone + \bm{V},\quad &\gamma \in (0,1]\\
        S\bone + \gamma^{-1}\bm{V},\quad &\gamma >1,
    \end{cases}
\end{align}
where $S\sim \text{Exp}(1)$ is a standard exponential random variable, independent of the random vector $\bm{V}=(V_1,V_2)^\top$, which also has standard exponential margins and exhibits asymptotic independence.
This additive mixture can be viewed as a mixture of perfect dependence and asymptotic independence.
The quantity $\gamma$ controls the relative contributions of $S$ and $\bm{V}$, determining the strength of the extremal dependence.
\cite{huser2019modeling} showed that the random vector $\bY$ is asymptotically independent for $\gamma \in (0,1]$ and asymptotically dependent for $\gamma >1$.
For $\gamma >1$, we consider the rescaled version of the initial random vector $\bY$, ensuring that the coordinatewise supremum of the limit set for $\bY$ remains $\sup(G)=(1,1)$.
Building on the work of \cite{nolde2022linking}, who derived the gauge function of $\bY$ for $\gamma\le 1$, we extend this construction for any $\gamma >0$ through rescaling
\begin{align}
\label{eq:gauge_HW}
    g_{\bY}(x,y)=
    \begin{cases} \min_{s\in[0,\min(x,y)/\gamma]}\cbr{s+g_{\bm V}(x-\gamma s,y-\gamma s)},\quad &\gamma\in(0,1]\\ \min_{s\in[0,\min(x,y)]}\cbr{s+g_{\bm V}(\gamma x-\gamma s, \gamma y-\gamma s)},\quad &\gamma > 1,
    \end{cases}
\end{align}
where $g_{\bm V}$ is the gauge function for $\bm V$.
\cite{nolde2022linking} derived the residual tail dependence coefficient, $\eta\in(0,1]$ \citep{ledford1997modelling}, in terms of $g_{\bY}$. 
Denoting this by $\eta^{Y}$, they showed that
\begin{align}
\label{eq:eta}
    \eta^{Y}
    =
\left[\min_{(x,y):\min(x,y)=1}g_{\bY}(x,y)\right]^{-1}
    =
\sbr{\min_{s\in[0,1/\gamma]}\cbr{s+g_{\bV}(x^*,y^*)(1-\gamma s)}}^{-1}
    =
    \begin{cases}
        \eta^{V},\quad \gamma \le \frac{1}{g_{\bm V}(x^*,y^*)},\\
        \gamma,\quad \frac{1}{g_{\bm V}(x^*,y^*)} < \gamma \le 1,
    \end{cases}     
\end{align}
where $(x^*,y^*)=\arg\min_{(x,y):\min(x,y)=1}g_{\bm V}(x,y)$ and $\frac{1}{g_{\bm V}(x^*,y^*)}=\eta^{V}$.
When $\gamma > \eta^V$, the gauge function $g_{\bV}$ begins to transform to produce $g_{\bY}$ around the diagonal line.
Although $\bY$ in \eqref{eq:HW_amix} only has asymptotically exponential margins, its limit set is valid for exponential margins.
We henceforth use the notation $g_{\bX}(x,y)$ for the gauge defined in~\eqref{eq:gauge_HW}.


\subsubsection{Analytic form of the gauge function}

We derive an analytic form of the gauge function $g_{\bX}(x,y)$ in~\eqref{eq:gauge_HW}, which is formulated as a minimization problem.
The objective function to be minimized is defined as
\begin{equation}
\label{eq:objectiveft}
 h(s):=s+g_{\bV}(x-\gamma s, y-\gamma s),\quad s\in[0,\min(x,y)/\gamma].
\end{equation}
We focus on the case where $\gamma\in(0,1]$ in equation~\eqref{eq:gauge_HW}, since the second line is identical to the first with $(x,y)$ replaced by $(\gamma x,\gamma y)$. 
We will consider examples where $g_{\bV}$ is convex, symmetric, and (piecewise) differentiable.
We denote the minimizer of $h(s)$ over $s\in[0,\min(x,y)/\gamma]$ by $\hat{s}$.
Note that the location of the minimizer $\hat{s}$ depends on both the point $(x,y)$ and $\gamma.$ 
In the Appendix, we show that $\gamma \leq 1/g_{\bV}(1,1)$ implies that equation~\eqref{eq:objectiveft} is minimized at $\hat{s}=0$. This, together with the behavior for $\gamma>1/g_{\bV}(1,1)$, is formalized in Proposition~\ref{prop:tangentline}, proven in Section~\ref{Appendix:tangentplane} of the Appendix. Lemma~\ref{lemma:point(x_0,y_0)} introduces a concept used in this Proposition.

\begin{lemma}
\label{lemma:point(x_0,y_0)}
    Assume that $g_{\bV}$ is convex, symmetric, and that $\frac{1}{g_{\bV}(1,1)}< \gamma \le 1$. 
    If $g_{\bV}$ is (piecewise) differentiable, there exists a unique segment of the level curve with rightmost endpoint $(x_0,y_0)$, $y>x$, such that the tangent line to the unit level curve $g_{\bV}(x,y)=1$ along the segment, and hence at the point $(x_0,y_0)$, passes through the point $(\gamma,\gamma).$ The segment may be a single point, in which case it is $(x_0,y_0)$.
    
    For $y>x$, the point $(x_0,y_0)$ is a solution to the equation
    \begin{align}
        (y-\gamma)=-\frac{g_{V,1}(x,y)}{g_{V,2}(x,y)}(x-\gamma), \label{eq:tangentgamma}
    \end{align}
    where $y$ is defined as a function of $x$ via the unit level curve $g_{\bV}(x,y)=1$. If there are multiple solutions, then $(x_0,y_0)$ is the solution with the largest value of $x$.
\end{lemma}

\begin{proposition}
\label{prop:tangentline}
Assume that $g_{\bV}$ is convex, symmetric, and twice (piecewise) differentiable.
For $\gamma \leq \frac{1}{g_{\bV}(1,1)}$, $g_{\bX}(x,y)=g_{\bV}(x,y)$.
For $\frac{1}{g_{\bV}(1,1)}< \gamma \le 1$, the explicit minimizer of equation~\eqref{eq:objectiveft} is given by
\begin{equation*}
\hat{s}=
\begin{cases}
    0,\quad&\text{if}\,\,\,\, \frac{x}{x+y}\le \frac{x_0}{x_0+y_0}\,\,\,\,\text{for}\,\,\,\,y\ge x\,\,\,\, \text{and}\,\,\,\,\frac{y}{x+y}\ge \frac{y_0}{x_0+y_0}\,\,\,\,\text{for}\,\,\,\,y<x\\
    \frac{x_0 y-y_0 x}{\gamma(x_0-y_0)},\quad&\text{if}\,\,\,\, \frac{x}{x+y}> \frac{x_0}{x_0+y_0}\,\,\,\,\text{for}\,\,\,\,y \ge x\,\,\,\,\text{and}\,\,\,\,\frac{y}{x+y}< \frac{y_0}{x_0+y_0}\,\,\,\,\text{for}\,\,\,\,y<x,
\end{cases}
\end{equation*}
and the analytic form of the gauge function $g_{\bX}$ in~\eqref{eq:gauge_HW} is
\begin{align*}
    g_{\bX}(x,y)=
    \begin{cases}
        g_{\bV}(x,y),\quad&\text{if}\,\,\,\, \frac{x}{x+y}\le \frac{x_0}{x_0+y_0}\,\,\,\,\text{for}\,\,\,\,y\ge x\,\,\,\, \text{and}\,\,\,\,\frac{y}{x+y}\ge \frac{y_0}{x_0+y_0}\,\,\,\,\text{for}\,\,\,\,y<x\\
        \frac{y-\rbr{\frac{\gamma-y_0}{\gamma-x_0}}x}{\gamma\rbr{1-\rbr{\frac{\gamma-y_0}{\gamma-x_0}}}},\quad&\text{if}\,\,\,\, \frac{x}{x+y}> \frac{x_0}{x_0+y_0}\,\,\,\,\text{for}\,\,\,\,y \ge x\,\,\,\,\text{and}\,\,\,\,\frac{y}{x+y}< \frac{y_0}{x_0+y_0}\,\,\,\,\text{for}\,\,\,\,y<x.
    \end{cases}
\end{align*}
The second form of $g_{\bX}$ corresponds to the gauge function derived from the tangent line to the unit level curve $g_{\bV}(x,y)=1$ at the point $(x_0,y_0)$ for $y>x$, passing through the point $(\gamma,\gamma)$.
The point $(x_0,y_0)$ is uniquely determined by the tangent line equation, shown in Lemma~\ref{lemma:point(x_0,y_0)}.
\end{proposition}

Similarly to the additively mixed gauge functions considered in Section~\ref{sec:Amix}, we analyze specific models for the gauge function $g_{\bV}$ and derive their closed-form expressions for $\frac{1}{g_{\bX}(1,1)}< \gamma \le 1$.
For detailed calculations of finding the minimizer $\hat{s}$ and the unique point $(x_0,y_0)$, refer to Section 2 of the supplementary material.


\paragraph{Gaussian gauge function for $g_{\bV}$}
The gauge function $g_{\bX}$ is 
\begin{equation*}
\label{eq:HWGauss}
    g_{\bX}(x,y)=\min_{s\in[0,\min\{x,y\}/\gamma]}s+\rbr{x+y-2\gamma s -2\rho\sqrt{(x-\gamma s)(y-\gamma s)}}/\rbr{1-\rho^2},\quad \gamma \le 1,\,\, \rho\in[0,1).
\end{equation*}
For $y> x$, we obtain the unique point $(x_0,y_0)$ from the tangent equation~\eqref{eq:tangentgamma}, where $y$ is obtained from $g_{\bV}(x,y)=1$, that is,
\begin{equation}
\label{eq:y_Ga}
    y = 2\rho\sqrt{(1-\rho^2)(x-x^2)}-x(1-2\rho^2)+(1-\rho^2).
\end{equation}
After simplification, the tangent equation reduces to a quadratic form
\begin{equation*}
    (2\gamma-1)^2x^2+(2\gamma-1)(1-\rho^2-2\gamma)x+\rho^2\gamma^2=0.
\end{equation*}
The unique solution is given by
\begin{equation*}
    x_0 = \min\cbr{\frac{-b - \sqrt{b^2-4ac}}{2a},\frac{-b + \sqrt{b^2-4ac}}{2a}},
\end{equation*}
where $a:=(2\gamma-1)^2$, $b:=(2\gamma-1)(1-\rho^2-2\gamma)$, and $c:=\rho^2\gamma^2$, with $y_0$ obtained by evaluating \eqref{eq:y_Ga} at $x_0.$
In this example, differentiating $h(s)$ yields a closed-form minimizer $\hat{s}$ solving the quadratic equation.
The non-zero solution to $h'(s)=0$ has the following explicit linear form:
\begin{equation}
\label{eq:shat_Gauss}
    \hat{s}=\frac{K(x+y)+\sqrt{-K\hatc}|x-y|}{2\gamma K},
\end{equation}
where $\hat{c}:=(1-\rho^2-2\gamma)^2 \ge0$ and $K=4\rho^2\gamma^2-\hatc\le 0$ for $\frac{1+\rho}{2}<\gamma\le1$.
Given the unique point $(x_0,y_0)$ obtained from the tangent line, the minimizer $\hat{s}$ in~\eqref{eq:shat_Gauss} reduces to $\frac{x_0 y-y_0 x}{\gamma(x_0-y_0)}$ as in Proposition~\ref{prop:tangentline}.
For the details, refer to Section 2 of the supplement.

When $\gamma=1$, the gauge function $g_{\bX}$ intersects the $\max\{x,y\}=1$ boundary, as shown in Figure~\ref{fig:HWGauss}, implying AI by Proposition~\ref{prop:AI}.
Specifically, there exist sets $D=\Tilde{D}=[\rho^2,1]\subseteq[0,1]$ such that $k(q)=\tilde{k}(q)=\1_D(q)$.
For $\gamma=1$ and for $q\ge 1$, the minimizer is $\hat{s}=1$, yielding $k(q)=1+(q-1)/(1-\rho^2)$ for $\rho\in(0,1)$.
Equivalently, $k(q)=mq+\rbr{1-m}$ for $q\ge 1$, where $m=\frac{1}{1-\rho^2}>1$ for $\rho\in(0,1)$, as illustrated in the middle panel of Figure~\ref{fig:HWGauss}.

\begin{figure}[ht!]
\centering
\includegraphics[width=4.5cm]{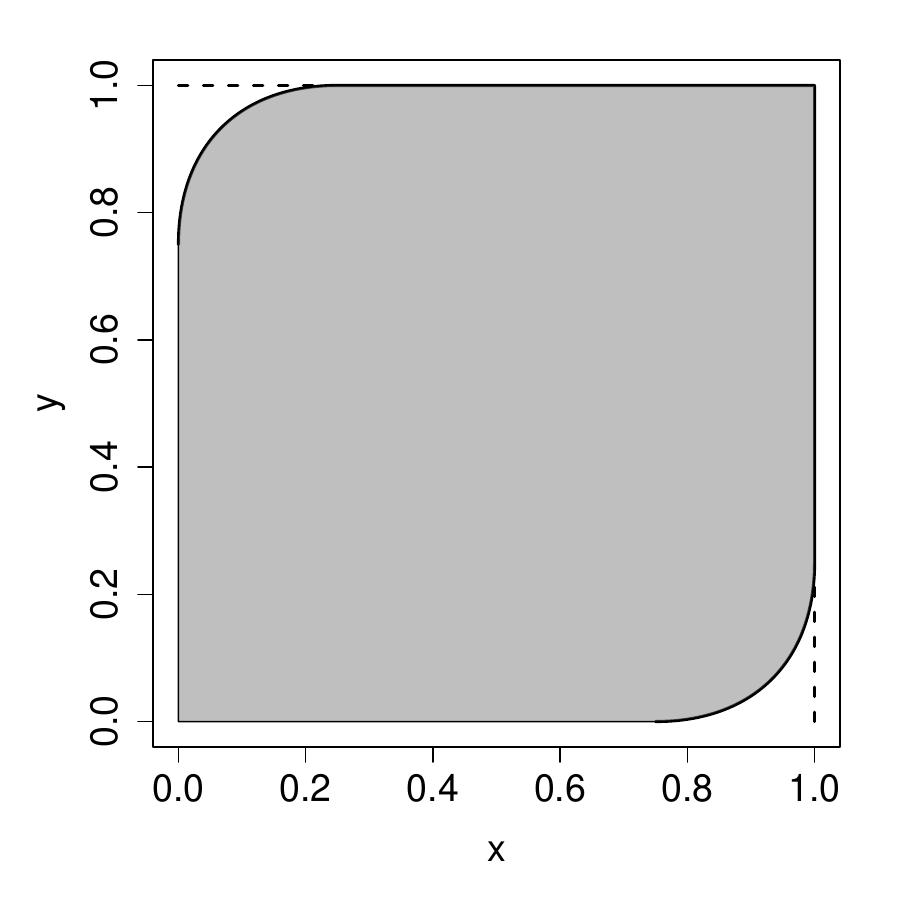}
\hspace{0.4cm}
\includegraphics[width=4.5cm]{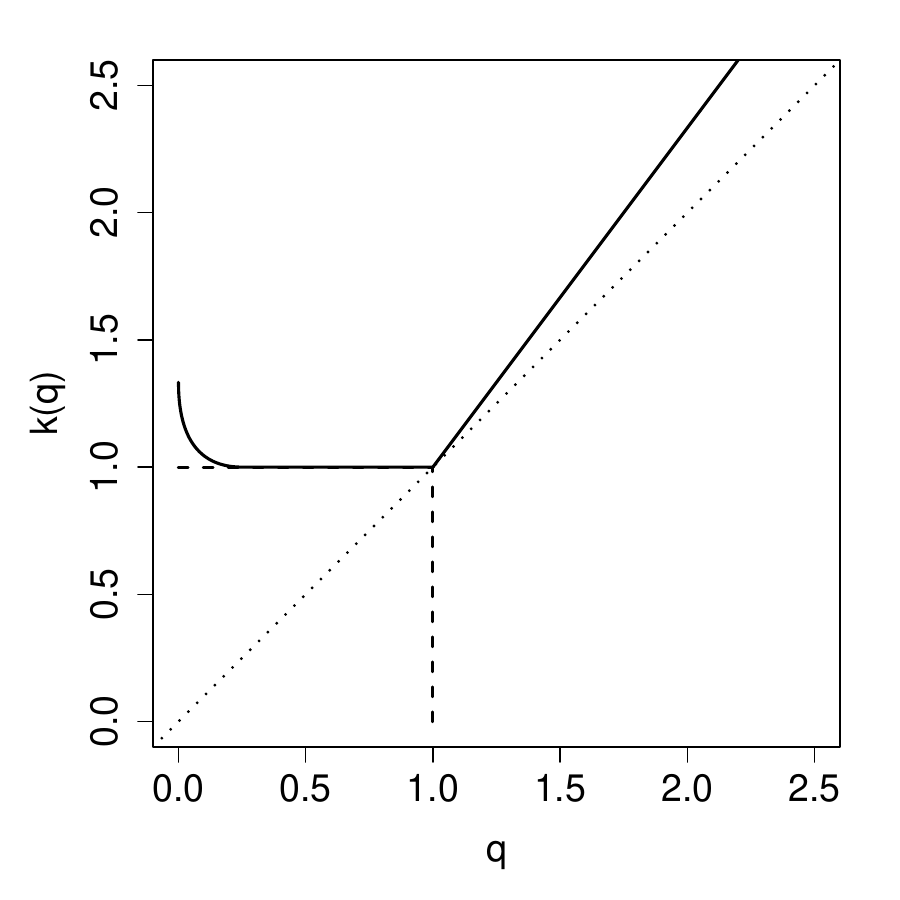}
\hspace{0.4cm}
\includegraphics[width=4.5cm]{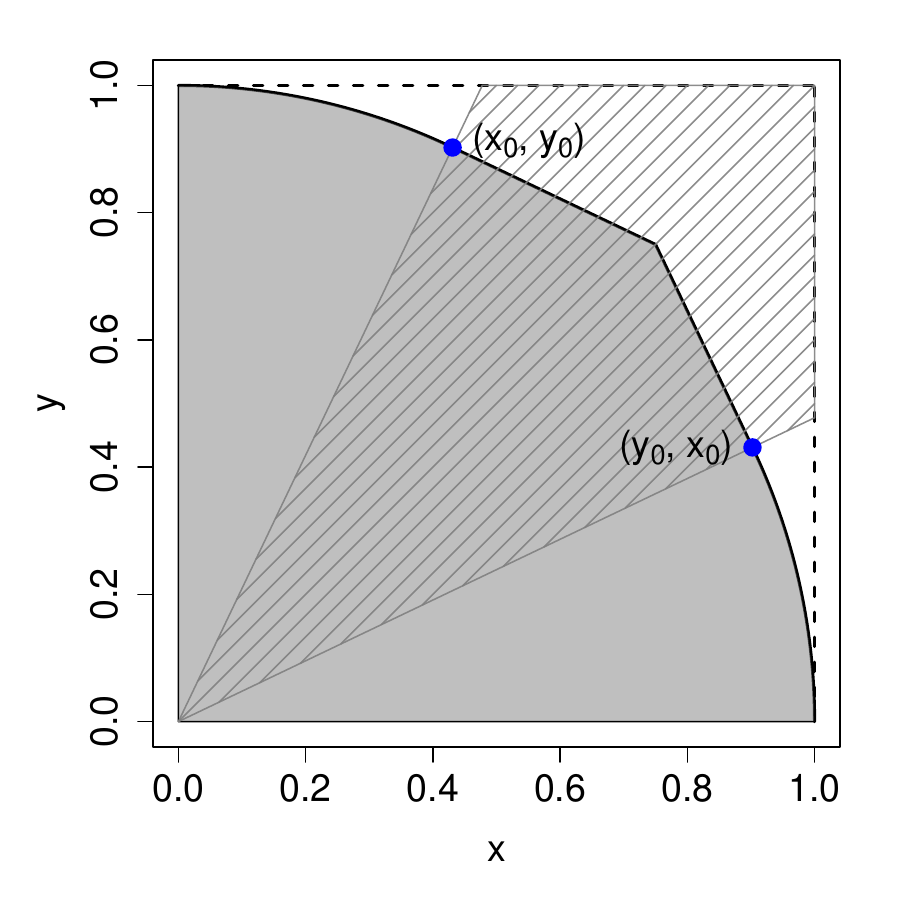}
\caption{\label{fig:HWGauss} (Left) unit level set of the additive mixture gauge function~\eqref{eq:gauge_HW} with Gaussian gauge $g_{\bV}$ for $\rho=0.5$ and $\gamma=1$. (Middle) corresponding function $k(q),$ $q\ge 0$. 
(Right) unit level set of the additive mixture gauge function with the inverted-logistic gauge $g_{\bV}$ for $\theta=0.5$ and $\gamma=0.75$. 
The blue solid circle marks the transition point $(x_0,y_0)=(0.43,0.90)$ for $y>x$ where the gauge switches form.
The hatched area indicates the region in which $g_{\bX}$ is determined by the tangent line equation. 
}
\end{figure}

\paragraph{Inverted-logistic gauge function for $g_{\bV}$}
The gauge function $g_{\bX}$ is
\begin{equation*}
    g_{\bX}(x,y)= \min_{s\in[0,\min\cbr{x,y}/\gamma]} s + \rbr{(x-\gamma s)^{1/\theta}+(y-\gamma s)^{1/\theta}}^\theta,\quad \gamma \le 1,\,\, \theta \in(0,1].
\end{equation*}
Unlike the Gaussian case, a closed-form expression for the minimizer $\hat{s}\in[0,\min\{x,y\}/\gamma]$ is not obtainable from the derivative due to the non-linear power term.
For $y> x$, the unique point $(x_0,y_0)$ is determined by the tangent line equation~\eqref{eq:tangentgamma}, where $y$ is obtained from the unit level curve $g_{\bV}(x,y)=1$, given by $y = \rbr{1-x^{1/\theta}}^\theta.$
Letting $z:=x^{1/\theta}$, it leads to the implicit equation $\frac{1}{\gamma}=(1-z)^{1-\theta}+z^{1-\theta},$ with $\gamma\in(2^{-\theta},1]$, $\theta\in(0,1]$, and $z\in[0,1]$.
The unique solution $x_0$ and corresponding $y_0$ are obtained numerically; see Section 2 of the supplement for more details, which then gives the explicit $\hat{s}$ via Proposition~\ref{prop:tangentline}.



\paragraph{Rectangular gauge function for $g_{\bV}$}
The gauge function $g_{\bX}$ is
\begin{equation*}
g_{\bX}(x,y)=\min_{s\in[0,\min\cbr{x,y}/\gamma]} s+ \max\cbr{\frac{1}{\theta}(x-y),\frac{1}{\theta}(y-x),\frac{1}{2-\theta}(x+y-2\gamma s)},\quad \gamma\le 1,\,\,\theta\in(0,1].
\end{equation*}
Note that the only term in $g_{\bV}$ depending on $s$ is $(x+y-2\gamma s)/(2-\theta)$. When $\gamma > 1/g_{\bV}(1,1) = 1-\theta/2$, the minimizer $\hat{s}\in[0,\min(x,y)/\gamma]$ is obtained by equating the two sides:
\begin{equation*}
    \frac{1}{\theta}(y-x)=\frac{1}{2-\theta}(x+y-2\gamma s),
\end{equation*}
yielding 
$
    \hat{s}=\frac{1}{\gamma\theta}\rbr{(\theta-1)y+x}.
$
The unique point $(x_0,y_0)=(1-\theta,1)$ for $y>x$ is a point of discontinuity in the derivative of $g_{\bV}$, but may be determined by setting $\hat{s}=\frac{1}{\gamma \theta}\rbr{(\theta-1)y_0+x_0}$ to zero and $y_0$ to one.



\section{Simulation study}
\label{sec:InferenceSimulation}

\subsection{Statistical inference procedure}
\label{sec:Inference}
We now explore the ability of the geometric framework, with gauge functions outlined in Sections~\ref{sec:Amix} and~\ref{sec:HWmodel}, to distinguish between the regimes of asymptotic dependence and asymptotic independence.
We follow the statistical inference and prediction procedure outlined in \cite{wadsworth2024statistical} and briefly summarize it here for completeness.

To fit the truncated gamma distribution \eqref{eq:tGam} over a high threshold, we first compute a high threshold $r_{\tau}(w)$ of the distribution $R\mid W=w$ across all $w\in [0,1]$.
Let $\widebar{F}(r\mid w)$ be the gamma survival function of $R\mid W=w.$
A candidate for the radial threshold is the quantile $r_{\tau}(w)$, defined such that $\widebar{F}_{R|W}(r_{\tau}(w)\mid w)=1-\tau$ for $\tau$ near 1.
This threshold can be estimated using either a rolling-windows quantile approach or additive quantile regression \citep{fasiolo2021fast}.
The rolling-windows approach partitions the simplex into overlapping blocks and calculates empirical $\tau$-quantiles of $R\mid W=w$ within each block.
For new angle values, the radial threshold is determined using local means of the threshold values from the overlapping blocks.
Using the threshold exceedances $(r_i,w_i)$, $i=1,\ldots,n'$, such that $r_i>r_\tau(w_i)$, the parameters are estimated maximizing the likelihood
\begin{equation*}
\label{eq:tglikeli}
    L\rbr{\bm{\Omega};(r,w)}=\prod_{i=1}^{n'}\frac{g_{\bX}(w_i,1-w_i;\bm{\theta})^\lambda}{\Gamma(\lambda)}\frac{r_i^{\lambda-1}\exp\rbr{-r_i g_{\bX}(w_i,1-w_i;\bm{\theta})}}{\widebar{F}\rbr{r_\tau(w_i);\lambda,g_{\bX}(w_i,1-w_i,\bm{\theta})}},
\end{equation*}
with $\bm{\Omega}=\rbr{\bm{\theta}^\top,\lambda}^\top$, and $\widebar{F}(\,\cdot\,;\lambda,g_{\bX}(w;\bm{\theta}))$ is the gamma survival function.
For model selection, standard criteria such as AIC or BIC can be used, while model assessment can be performed using PP-plots.

To estimate the probability of extreme sets of interest, letting $R'=R/r_\tau(W)$, we consider the equation
\begin{equation*}
\label{eq:tailprob}
    \prob(\bX \in C)=\prob(\bX\in C \mid R' > 1)\prob(R' > 1),
\end{equation*}
for any set $C$ contained within the region $\cbr{(x,y)\in\reals_+^{2}:x+y > r_\tau\rbr{x/(x+y)}}.$
We can simulate samples from the distribution $\bX\mid R'>1$ by multiplying samples $w^*$ from the empirical distribution of $W\mid R'>1$, and samples $r^*$ from the fitted truncated gamma distribution $R\mid \cbr{W=w, R>r_\tau(\bw)}$, conditional on $w^*$, to give $\bx^*=(r^*w^*,r^*(1-w^*))$.
The quantity $\prob(R'>1)$ can be empirically estimated as the proportion of threshold exceedances.
For extrapolation, these procedures can be extended to simulation from $\bX \mid R'>k$ with a suitably chosen $k > 1$.
The adapted formula for extrapolation is
\begin{equation}
\label{eq:tailprob_extrapo}
    \prob(\bX \in C)=\prob(\bX\in C\mid R'>k)\prob(R'>k\mid R'>1)\prob(R'>1),
\end{equation}
where we choose the largest value of $k$ such that $C \subset \cbr{(x,y)\in\reals_+^2:x+y>k r_\tau\rbr{x/(x+y)}}$.
For further details on estimation of $\prob(R'>k\mid R'>1)$ and the choice of $k$, see \cite{wadsworth2024statistical}.

\subsection{Simulation study}
\label{sec:Simulation}

We compare the performance of the geometric approach with flexible additively mixed gauge functions to existing random-scale copula models that are able to interpolate between AD and AI, focusing on their ability to identify the class of extremal dependence.
The geometric criteria developed in Section~\ref{sec:Theory} are characterized by the behavior of $k(q)$ and $\tilde{k}(q)$.
In particular, for a limit set to indicate AD, it must exhibit a pointy shape.
We implement the geometric methodology using the R package $\texttt{geometricMVE}.$
The implementation is computationally fast, leveraging the explicit forms of the additively mixed functions, obtained in Section~\ref{sec:Amix} and \ref{sec:HWmodel}.
We also consider the following very simple, single parameter gauge function, referred to as the max-min gauge, defined as
\begin{equation}
\label{eq:gauge_mm}
    g_{\bX}(x,y)=
    \begin{cases}
        \frac{1}{\theta}\max\rbr{x,y} + \rbr{1-\frac{1}{\theta}}\min\rbr{x,y},\quad \theta\in(0,1),\\
        \max\rbr{x,y} + \rbr{\frac{1}{\theta}-1}\min\rbr{x,y},\quad \theta\in[1,\infty).
    \end{cases}
\end{equation}
For $\theta < 1 $, the function reduces to the logistic gauge function, which yields a pointy limit set, indicating asymptotic dependence.
Parameter values $\theta \ge 1$ correspond to asymptotic independence, where we rescale the gauge function to ensure the coordinatewise supremum equals 1.

For copula-based models, we specifically consider the bivariate models from \cite{wadsworth2017modelling} (EV) and \cite{huser2019modeling} (HW), which are implemented in the R packages $\texttt{EVcopula}$ and $\texttt{spatialADAI}$, respectively.
The EV model represents the copula of the random vector $\bY=(X,Y)^\top=S(V_1,V_2)^\top$, where $S\sim GP(1,\xi)$ follows a unit-scale generalized Pareto and is independent of the random vector $(V_1,V_2)$ such that $\max(V_1,V_2)=1$.
Specifically, $(V_1,V_2)^\top=\rbr{\frac{V}{\max(V,1-V)},\frac{1-V}{\max(V,1-V)}}^\top$,
with $V\sim \text{Beta}(\alpha,\alpha)$, whose simple and flexible shape leads to a straightforward model \citep{wadsworth2017modelling}.
Asymptotic dependence is implied when $\xi >0 $, whereas $\xi \le 0$ indicates asymptotic independence.
The HW model is defined through the random scale representation $\bY=(X,Y)^\top=R^{\delta}(W_1^{1-\delta},W_2^{1-\delta})^\top$ for $\delta\in[0,1]$, where $R$ is distributed to a standard Pareto and is independent of the random vector $(W_1,W_2)$ with standard Pareto margins.
We consider a bivariate Gaussian dependence structure for $(W_1,W_2)$. 
\cite{huser2019modeling} showed that $\delta > 1/2$ implies AD, whereas $\delta \le 1/2$ indicates AI.
Note that this is the same dependence structure as the model given in~\eqref{eq:HW_amix}, with $\gamma=\delta/(1-\delta).$

We fit the copula models to data transformed to uniform margins.
Parameter estimation was carried out using censored likelihood, to ensure models are fitted only to extreme data.
For both models we consider {\em max-censoring}, in which the contributions of observations are fully censored if all variables are below the threshold, but are left completely uncensored otherwise.
For the HW model, we additionally consider {\em partial censoring}, where only observations above a high threshold contribute to the estimation, with non-extreme values being censored.

We consider five distinct scenarios, each characterized by two different dependence structures with nearly equivalent strengths of dependence, as measured by $\chi$ or $\eta$:

\begin{enumerate}[topsep=0.5pt,itemsep=-0.6ex]
    \item \textbf{Strongly dependent AD (st.d.AD)}: Logistic model with $\gamma=0.2$ and Dirichlet model \citep{coles1991modelling} with $\alpha=\beta=14$, resulting in $\chi=0.85$.
    \item \textbf{Moderately strongly dependent AD (mst.d.AD)}: Logistic model with $\gamma=0.4$ and Dirichlet mode with $\alpha=\beta=2.85$, yielding $\chi=0.68$.
    \item \textbf{Weakly dependent AD (w.d.AD)}: Logistic model with $\gamma=0.8$ and Dirichlet model with $\alpha=\beta=0.285$, producing $\chi=0.26$.
    \item \textbf{Strongly dependent AI (st.d.AI)}: Inverted-logistic model with $\theta=0.2$ and Gaussian model with $\rho=0.74$, yielding $\eta=0.87$.
    \item \textbf{Weakly dependent AI (w.d.AI)}: Inverted-logistic model with $\theta=0.8$ and Gaussian model with $\rho=0.14$, resulting in $\eta=0.57$. 
\end{enumerate}

We draw samples of size $n=5,000$ from each model, repeating the process across $1,000$ iterations.
For the geometric approach in the bivariate case, we employ a rolling-windows quantile method to calculate the high threshold $r_\tau(w)$.
We then fit the truncated gamma distribution~\eqref{eq:tGam} with different additively mixed gauge functions to the threshold exceedances $(r_i,w_i)$, $i=1,\ldots,n',$ such that $r_i\ge r_\tau(w_i)$, and estimate the associated parameters for each case.
We consider seven different additive mixtures: additive mixtures of an independent exponential random vector with Gaussian (ExpGa), Inverted-logistic (ExpInv), and Rectangle (ExpRect) for $g_{\bV}$, as described in Section~\ref{sec:HWmodel}; additive mixtures of logistic with Gaussian (GaLog), Inverted-logistic (InvLog), and rectangle (RectLog), respectively as detailed in Section~\ref{sec:Amix}; and the max-min gauge function (MM) in~\eqref{eq:gauge_mm}.

We fix $\tau=0.95$ as the threshold for identifying radial threshold exceedances in the geometric approach.
The corresponding threshold for copula-based models is then set to ensure they are fitted to the same number of exceedances.
We summarize the sensitivity (correctly identifying AD) and specificity (correctly identifying AI) as measures of performance for identifying the class of extremal dependence across the five scenarios, as shown in Table~\ref{tab:coverage}.
Particular attention is given to scenarios 3 and 4 due to the challenges in correctly identifying these cases.

Across the strongly dependent AD and weakly dependent AI scenarios, all three methods perform well in identifying the extremal dependence class under two different dependence structures.
In the challenging weakly dependent AD case, the copula-based model of \cite{wadsworth2017modelling} performs best, whereas the model of \cite{huser2019modeling} shows noticeably weaker performance.
By contrast, the geometric approach exhibits consistently robust behavior, with particularly strong coverage for additive mixtures such as ExpInv and MM.
In the strongly dependent AI scenario, the geometric approach continues to perform reliably, although the single-parameter MM gauge function shows lower specificity.
Interestingly, in this setting, the bivariate model of \cite{huser2019modeling} outperforms the other methods, while the model of \cite{wadsworth2017modelling} struggles to correctly identify asymptotic independence.
Overall, the EV model under the random-scale copula framework achieves the highest average coverage rates across all scenarios.
Nevertheless, the geometric additive mixtures attain coverage rates exceeding 0.90 for the `ExpGa', `ExpInv', `ExpRect', and `GaLog' models.
Notably, across weakly dependent AD and strongly dependent AI regimes, the geometric approach exhibits lower dispersion in coverage rates across scenarios than the copula-based methods.

\begin{table}[ht!]
    \centering
    \begin{tabular}{|c||c|c|c|c|c|c|} 
      \hline
        & 1. st.d.AD & 2. mst.d.AD & 3. w.d.AD & 4. st.d.AI & 5. w.d.AI & Overall \\
        & log / diri & log / diri & log / diri & Inv / Ga & Inv / Ga & average \\
      \hline\hline
      ExpGa & 0.982 / 0.938 & 0.971 / 0.942 & 0.859 / 0.890 & 0.828 / 0.768  & 0.993 / 0.977 & 0.917\\
      \hline
      ExpInv & 1.000 / 1.000 & 1.000 / 1.000 & 0.927 / 0.939 & 0.741 / 0.485 & 0.986 / 0.994 & 0.907\\
      \hline
      ExpRect & 0.991 / 0.941 & 0.977 / 0.944 & 0.774 / 0.781 & 0.832 / 0.803 & 0.996 / 0.999 & 0.904\\
      \hline
      GaLog & 0.980 / 0.833 & 0.958 / 0.969 & 0.776 / 0.789 & 0.901 / 0.874 & 0.985 / 0.994 & 0.906\\
      \hline
      InvLog & 1.000 / 1.000 & 0.999 / 1.000 & 0.616 / 0.668 & 0.848 / 0.664 & 0.986 / 0.991 & 0.877\\
      \hline
      RectLog & 0.991 / 0.946 & 0.905 / 0.840 & 0.727 / 0.724 & 0.795 / 0.769 & 0.988 / 0.992 & 0.868\\
      \hline
      MM & 1.000 / 1.000 & 1.000 / 1.000 & 0.926 / 0.938 & 0.691 / 0.356 & 0.987 / 0.994 & 0.889\\
      \hline\hline
      HW-Cmax & 1.000 / 1.000 & 0.921 / 0.260 & 0.323 / 0.316 & 0.999 / 1.000 & 1.000 / 1.000 & 0.777\\
      HW-Cpar & 1.000 / 1.000 & 0.921 / 0.921 & 0.191 / 0.173 & 0.963 / 0.978 & 1.000 / 1.000 & 0.811\\
      \hline\hline
      EV & 1.000 / 1.000 & 1.000 / 1.000 & 0.998 / 1.000 & 0.428 / 0.847 & 1.000 / 1.000 & 0.927\\
      \hline
    \end{tabular}
\caption{The sensitivity and specificity for classifying AD and AI for each model across the five scenarios.
log: logistic dependence structure, diri: Dirichlet dependence structure, HW: the model of \cite{huser2019modeling} with a Gaussian dependence structure for $(W_1,W_2)^\top$, where Cmax and Cpar represents the max- and partial-censoring schemes, respectively, and EV: the copula-based model of \cite{wadsworth2017modelling} with a beta distribution for $V$.
}
\label{tab:coverage}
\end{table}

We now consider summary dependence measures: the slope $\alpha$ in the conditional extremes model, which coincides with the minimizer $\kappa$ in Section~\ref{sec:findingcoordinatewise}, the residual tail dependence coefficient $\eta$ in~\eqref{eq:eta}, and the tail dependence coefficient $\chi$ for the seven additively mixed gauge functions considered.
To calculate the model-based $\chi_{m}(u)$, we evaluate the probability of extreme sets
$C_1=\rbr{-\log(1-u),\infty}\times\rbr{-\log(1-u),\infty}$ and $C_2=\rbr{-\log(1-u),\infty}\times\rbr{0,\infty}$, where $-\log(1-u)$ for $u$ near 1 corresponds to a high quantile of a standard exponential distribution.
The model-based tail dependence coefficient is then defined as
\begin{equation}
\label{eq:modelChi}
    \chi_{m}(u)=\prob\rbr{\rbr{X,Y}\in C_1\mid(X,Y)\in C_2}.
\end{equation}
To estimate $\chi_{m}(u)$, we employ equation \eqref{eq:tailprob_extrapo} for a suitably chosen $k>1$ and draw sufficiently large samples accordingly.
For each iteration, the largest $k$ is chosen such that \newline $C_{i}\subset \cbr{(x,y):(x+y)>k r_\tau\rbr{x/x+y}}$ for $i=1,2.$
Setting $u=0.9999$, we present box-plots of the estimated model-based $\hat{\chi}_m(u)$ across five scenarios for the seven additively mixed gauge functions in Figure~\ref{fig:chi_ks}.
In weakly dependent AI cases, the model-based estimates $\hat{\chi}_m(u)$ tend to be close to zero, so we take the logarithm of $\hat{\chi}_m(u)$ for better visualization.

\begin{figure}[ht!]
\centering
\includegraphics[width=4cm]{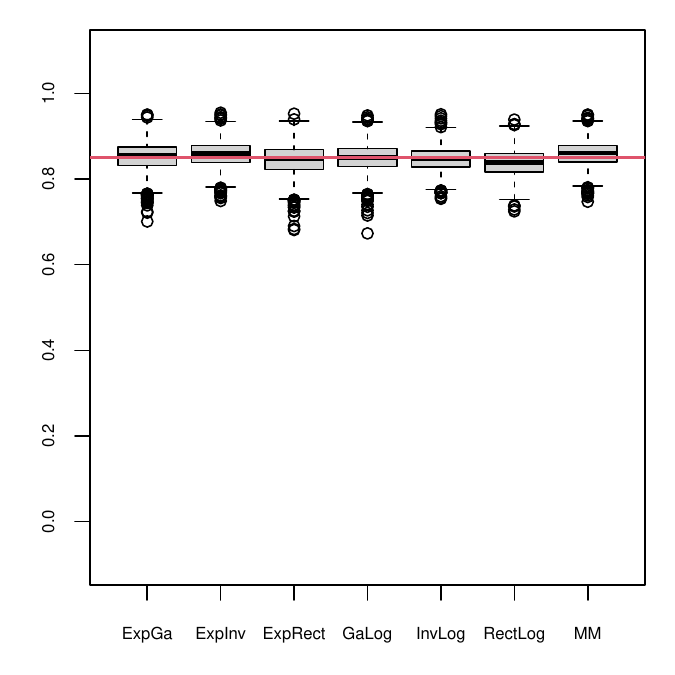}
\includegraphics[width=4cm]{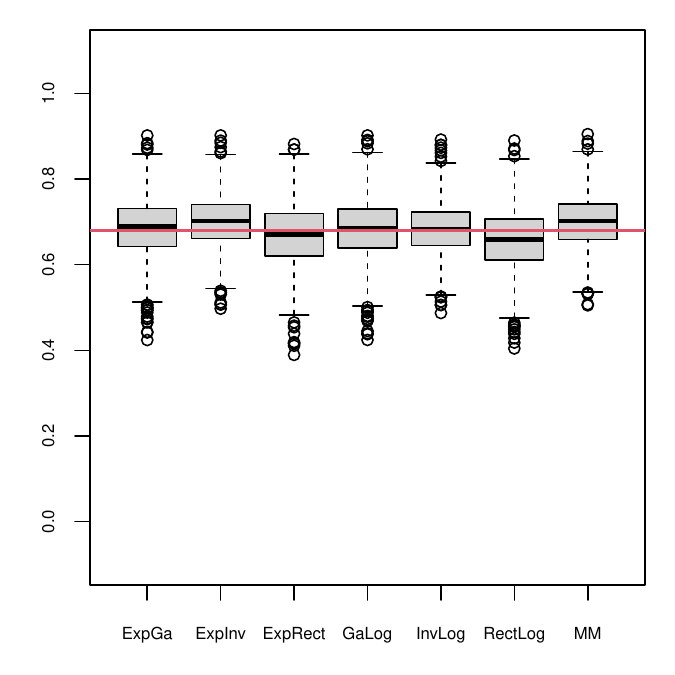}
\includegraphics[width=4cm]{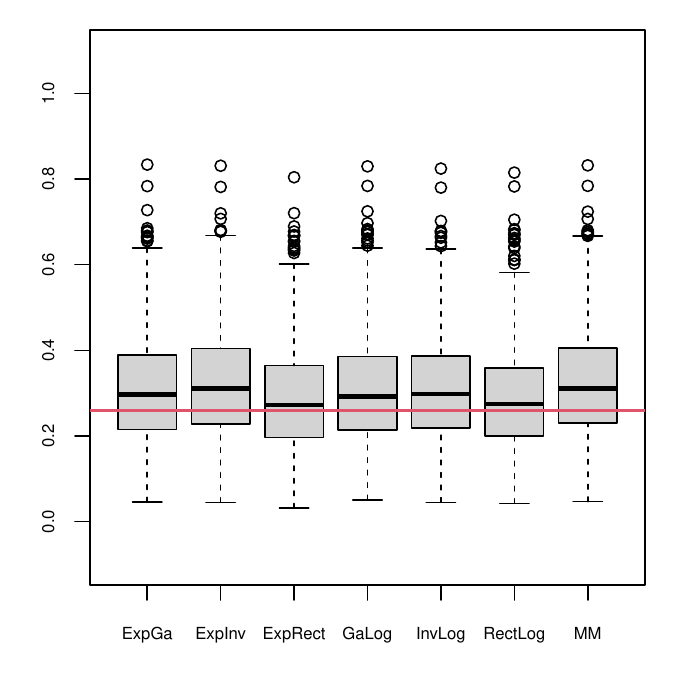}\\
\includegraphics[width=4cm]{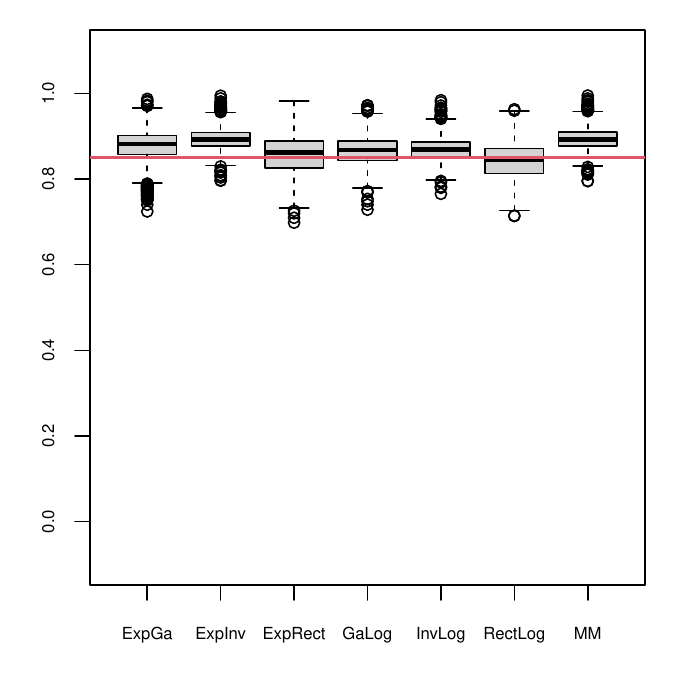}
\includegraphics[width=4cm]{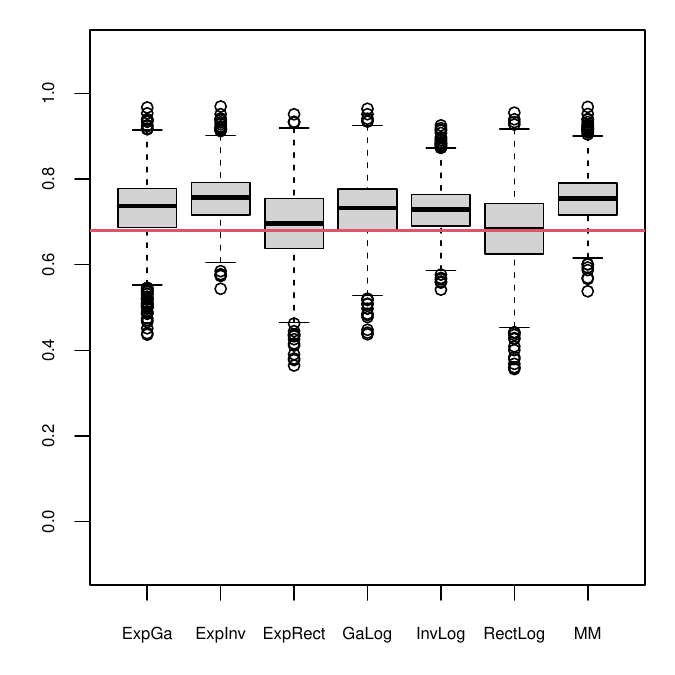}
\includegraphics[width=4cm]{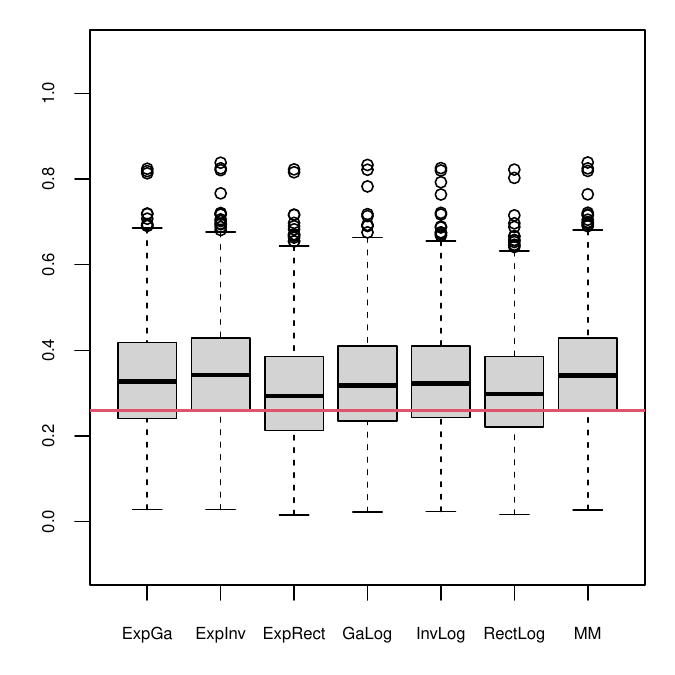}\\
\includegraphics[width=4cm]{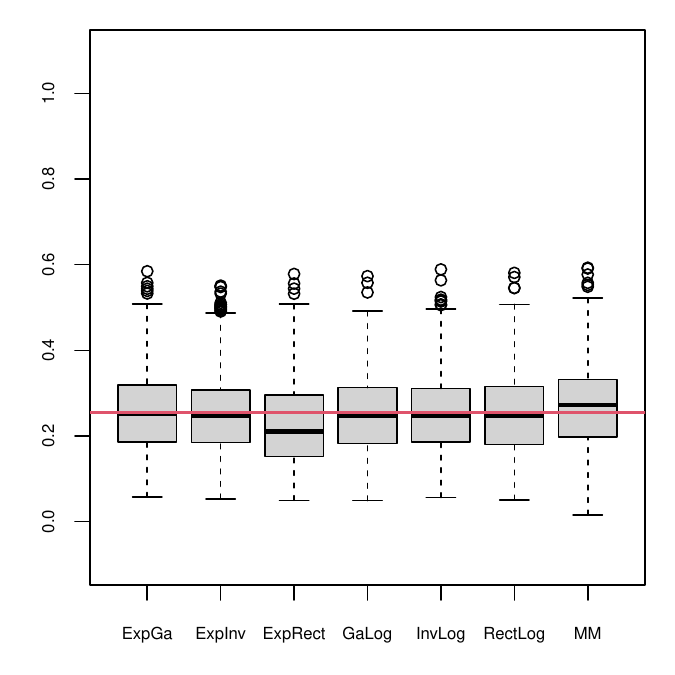}
\includegraphics[width=4cm]{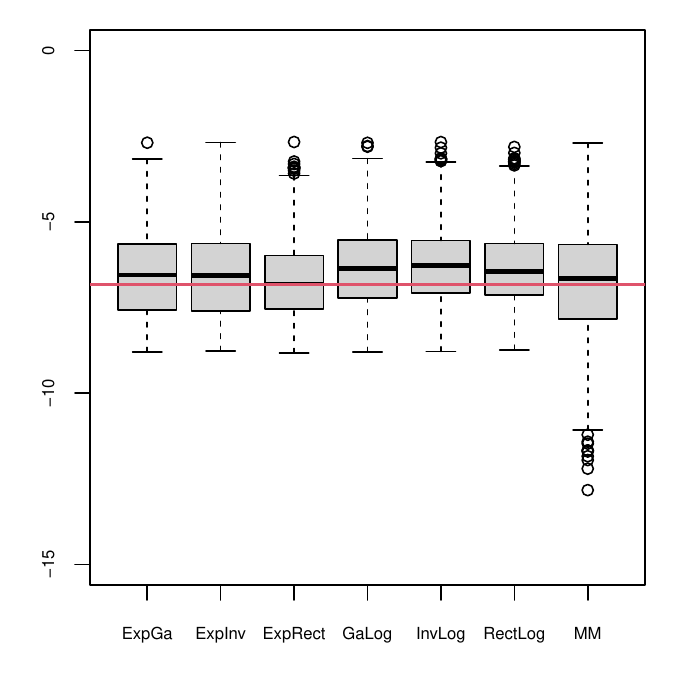}
\includegraphics[width=4cm]{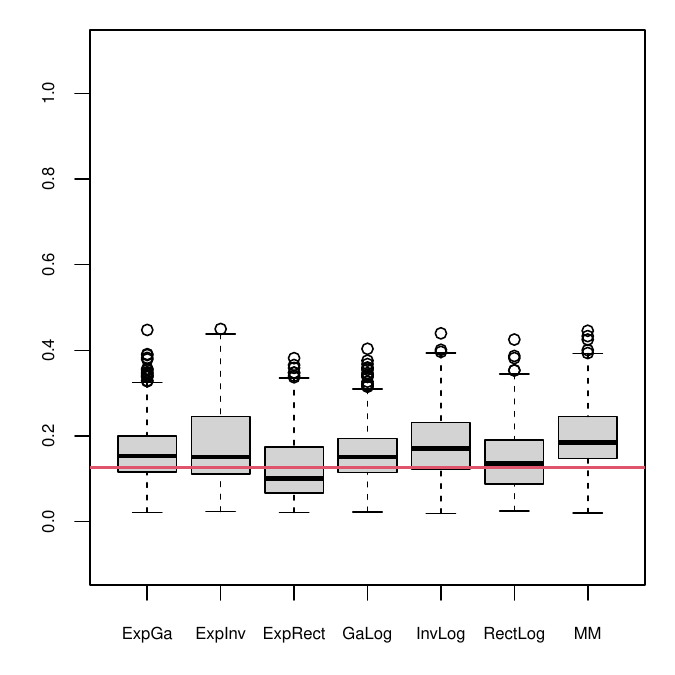}
\includegraphics[width=4cm]{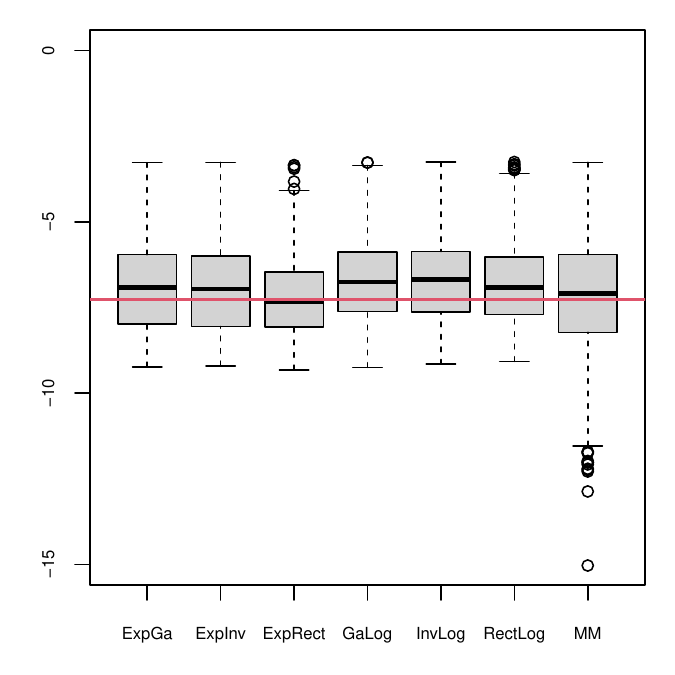}
\caption{\label{fig:chi_ks} Box-plots of the estimated model-based $\hat{\chi}_m(u)$ over 1,000 iterations. 
From left to right, the top row represents the scenarios for st.d.AD, mst.d.AD, w.d.AD under the logistic dependence structure.
The second row represents the same scenarios (st.d.AD, mst.d.AD, w.d.AD) under the Dirichlet dependence structure.
In the bottom row, the first two plots show s.d.AI and w.d.AI under the Inverted-logistic dependence structure, while the last two plots indicate st.d.AI and w.d.AI under the Gaussian dependence structure.
The box-plots are ordered as follows for each plot: ExpGa, ExpInv, ExpRect, GaLog, InvLog, RectLog, and MM.
The red lines indicate the true values of $\chi(0.9999)$ for each case.
For weakly dependent AI cases, the values are shown on the logarithm scale.}
\end{figure}

Additionally, we present box-plots of the estimated dependence measures $\hat{\kappa}$, $\hat{\eta}$, as well as box-plots of the estimated probabilities of the lower-right corner set $C_{lr}=\rbr{8,\infty}\times\rbr{0,7}$ for the seven additively mixed gauge functions across five scenarios, included in Section 3 of the supplementary material.
Considering sets such as $C_{lr}$, where not all variables are extreme, is of particular interest in the geometric approach, as it allows for extrapolation in any direction and provides better performance in comparison to other existing extreme value methodologies \citep{wadsworth2024statistical}.
All additive mixtures provide reasonably accurate estimates of the probability $\prob((x,y)\in C_{lr})$.

\section{Application to river flow data}
\label{sec:Application}

We apply the geometric approach, using multiple additively mixed gauge functions, to analyze a dataset of river flow measurements.
Our primary objective is to investigate the identification of the extremal dependence class between pairs of sites, using the geometric criteria developed in Section~\ref{sec:Theory}.
For comparison, we fit the random scale copula models outlined in Section~\ref{sec:Simulation}.

\subsection{River flow data}
\label{sec:riverflow}

We apply the geometric approach with additively mixed gauge functions to the daily average river flow $(m^3/s)$ collected from four gauging stations in the north-west of England, denoted as $(X_1,X_2,X_3,X_4)$.
This river flow dataset was initially analyzed by \cite{simpson2020determining} who investigated which groups of variables exhibited simultaneous extremes.
The analysis was followed by \cite{wadsworth2024statistical} who fitted the geometric model using three candidate gauge functions, investigating the joint tail dependence coefficients and assessing the goodness of fit.

Our primary focus is on identifying the class of extremal dependence for pairwise variables.
For the geometric approach, marginal transformations are applied using a semi-parametric approach, where a generalized Pareto distribution is fitted above a high threshold, and the empirical distribution is fitted below the threshold \citep{coles1991modelling}.
The threshold values for each variable are set at the 0.95 quantiles of their respective distributions.
Using a rolling-windows quantile method to calculate a high threshold $r_\tau(w)$ with $\tau=0.95$, we fit a truncated gamma distribution~\eqref{eq:tGam} with the seven additively mixed gauge functions for $g_{\bX}(x,y;\bm{\theta})$ and obtain the associated ML estimates as described in Section~\ref{sec:Simulation}.

Based on the geometric criteria in Section~\ref{sec:Theory}, asymptotic dependence is identified only for the pair $(X_2,X_3)$ across most of the additively mixed gauge functions (ExpGa, ExpInv, ExpRect, InvLog, RectLog, and MM), resulting in pointy limit sets.
Numerical summaries for the pair $(X_2,X_3)$ are provided in Section 4 of the supplementary material.
To further illustrate the results across different variable pairs, we visualize the shape of the limit sets characterized by the fitted additive gauge functions for each pair.
For each pair, the best and second-best fitted gauge functions, as determined by AIC, are overlaid in Figure~\ref{fig:fittedgauge_riverflow}: MM and ExpGa for $(X_1,X_2)$; MM and ExpRect for $(X_1,X_3)$; ExpRect and RectLog for $(X_1,X_4)$; MM and ExpGa for $(X_2,X_3)$; MM and ExpRect for $(X_2,X_4)$; and ExpRect and MM for $(X_3,X_4)$.

\begin{figure}[ht!]
\centering
\includegraphics[width=4.5cm]{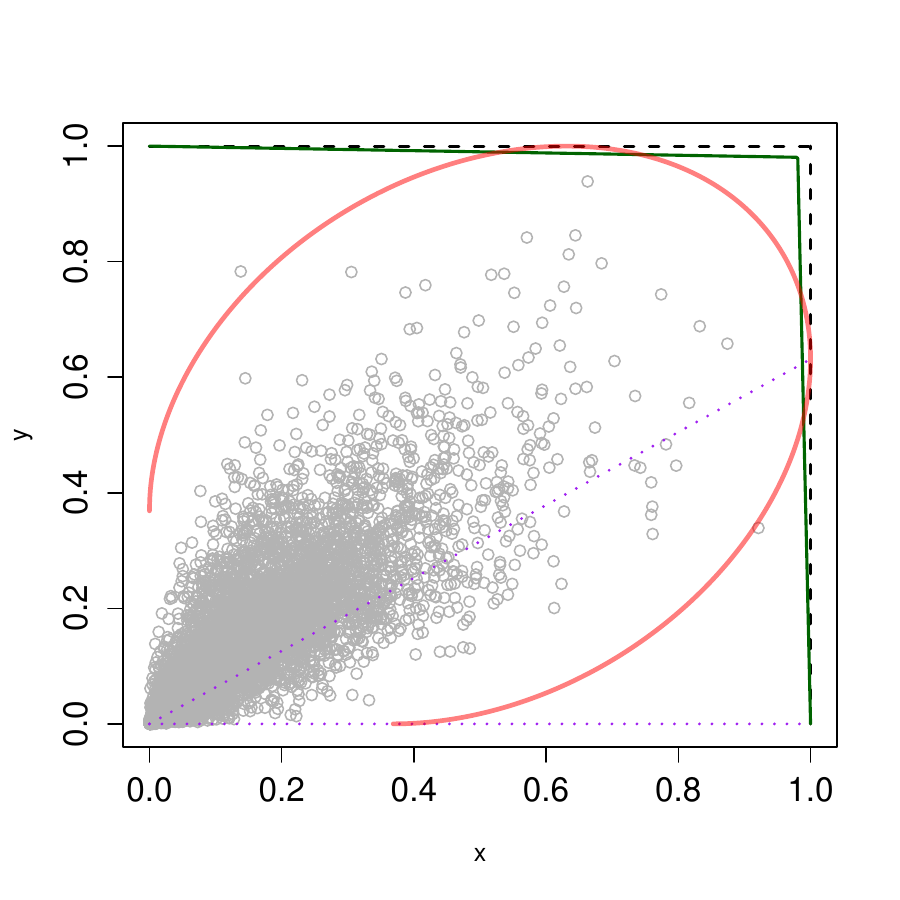}
\includegraphics[width=4.5cm]{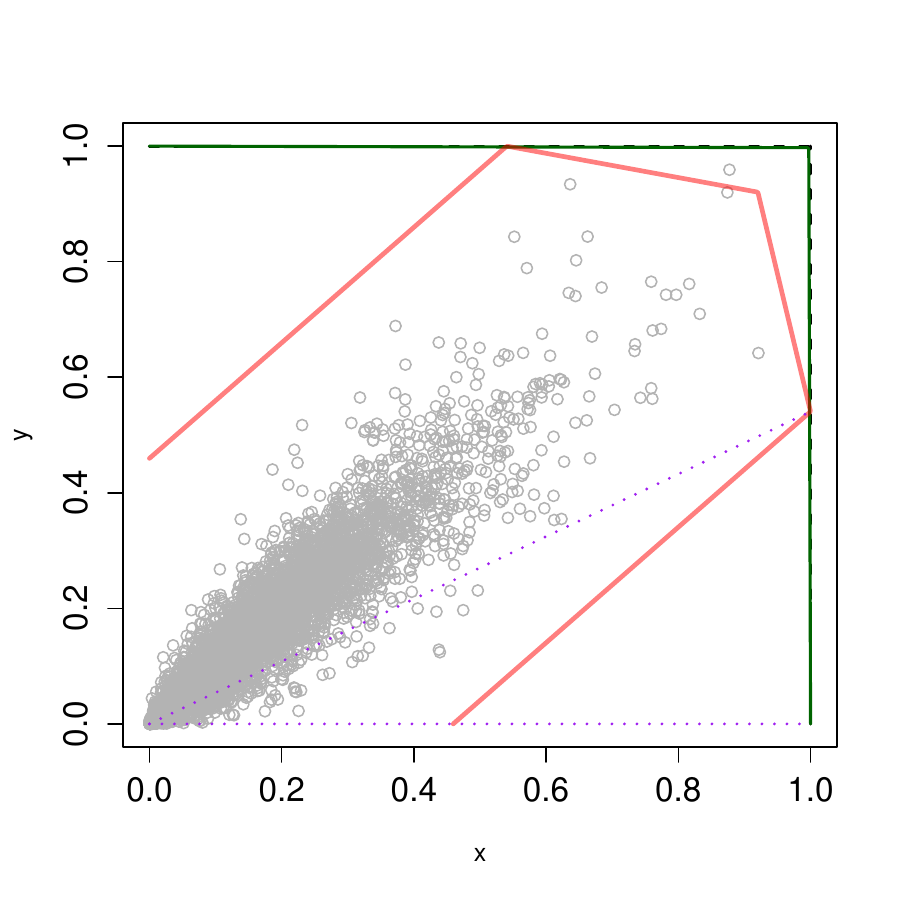}
\includegraphics[width=4.5cm]{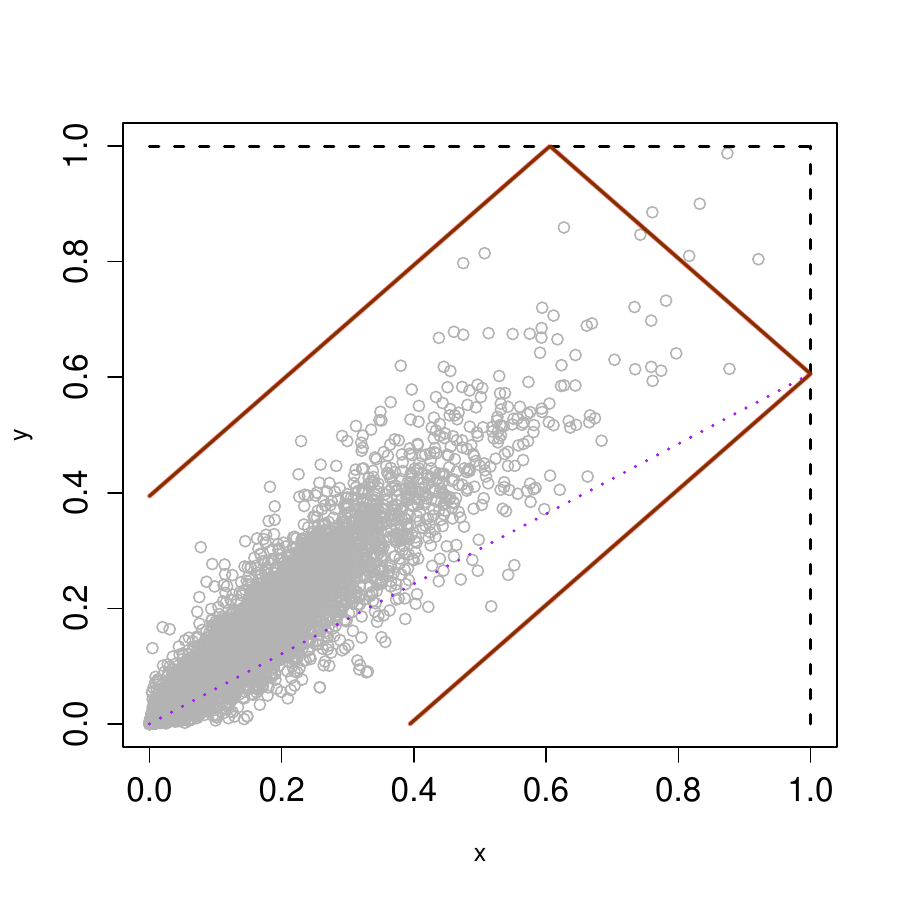}
\includegraphics[width=4.5cm]{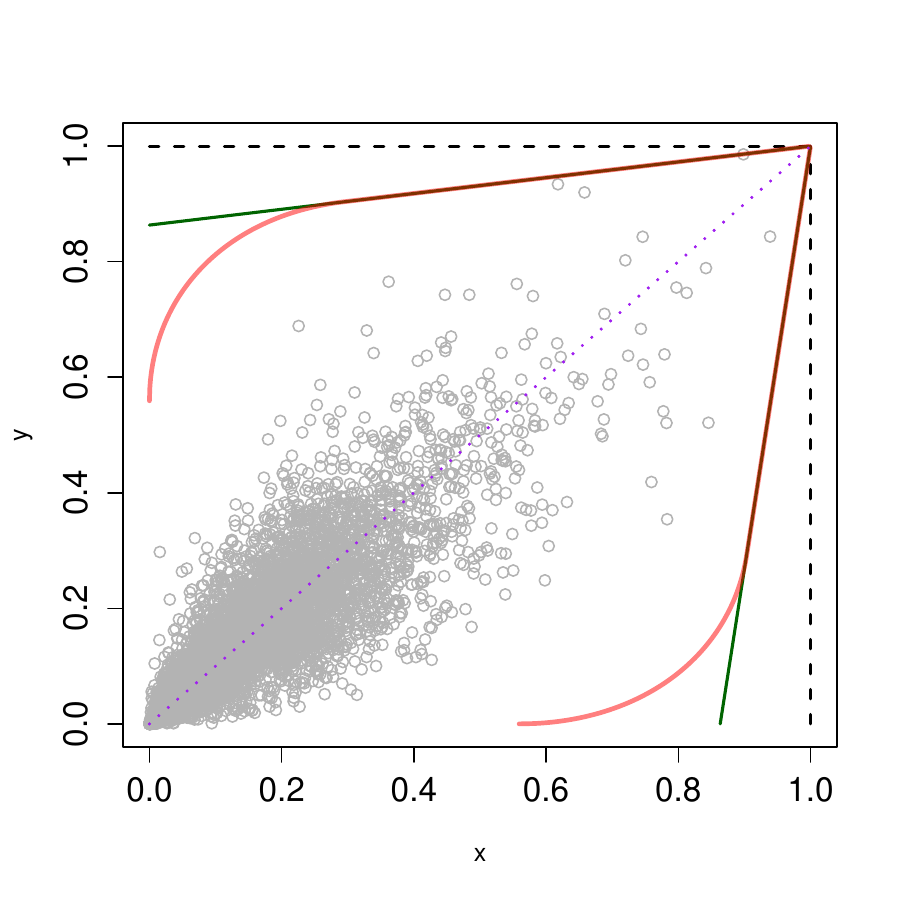}
\includegraphics[width=4.5cm]{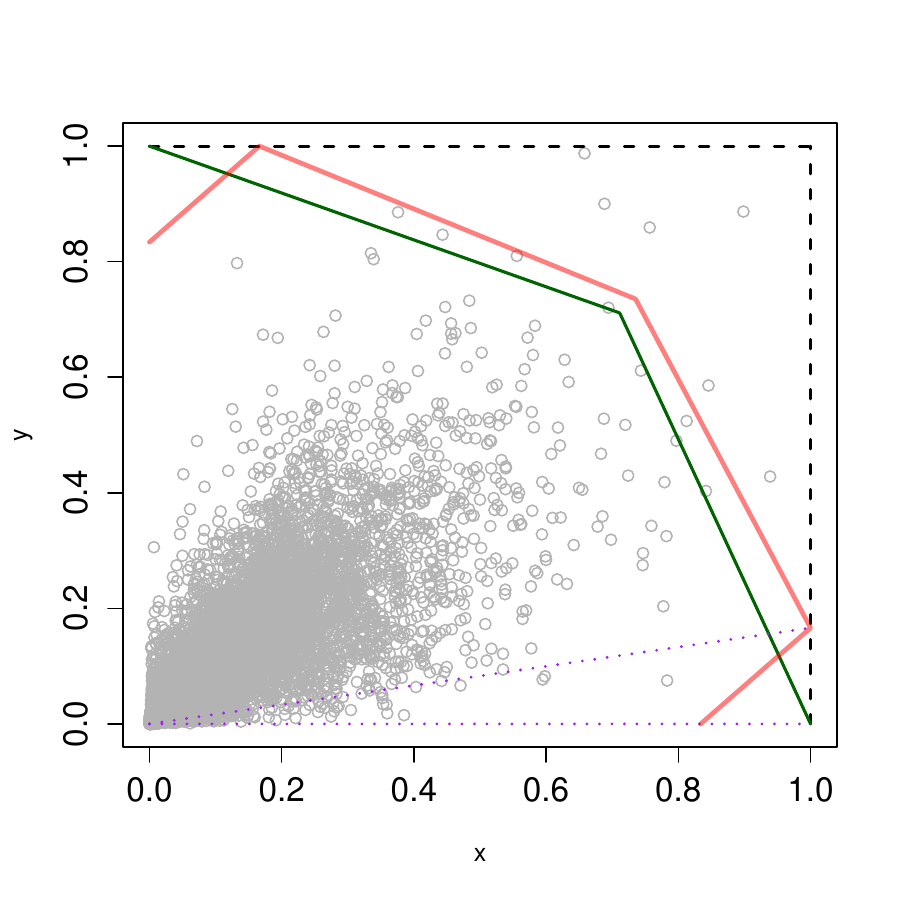}
\includegraphics[width=4.5cm]{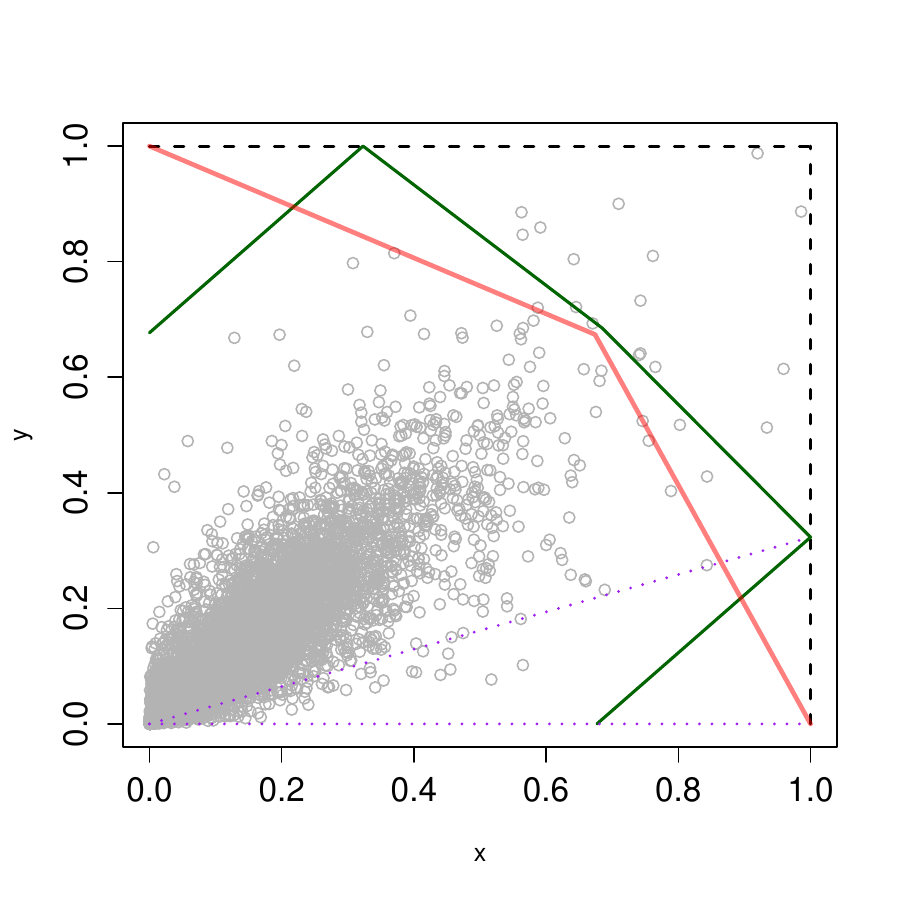}
\caption{\label{fig:fittedgauge_riverflow} The best (dark green) and second-best (red) fitted gauge functions are shown for each variable pair.
The plots are arranged in the order, from top to bottom and from left to right: $(X_1,X_2)$, $(X_1,X_3)$, $(X_{1},X_4)$, $(X_2,X_3)$, $(X_2,X_4)$, and $(X_3,X_4)$.
The purple dotted line corresponds to the slope estimate $\hat{\alpha}$ for the conditional extremes model.}
\end{figure}

In Figure~\ref{fig:chiplot_riverflow}, we compare the empirical $\hat{\chi}(u)$ with the model-based estimate $\hat{\chi}_m(u)$ as $u\rightarrow 1$.
\begin{figure}[ht!]
\centering
\includegraphics[width=4.5cm]{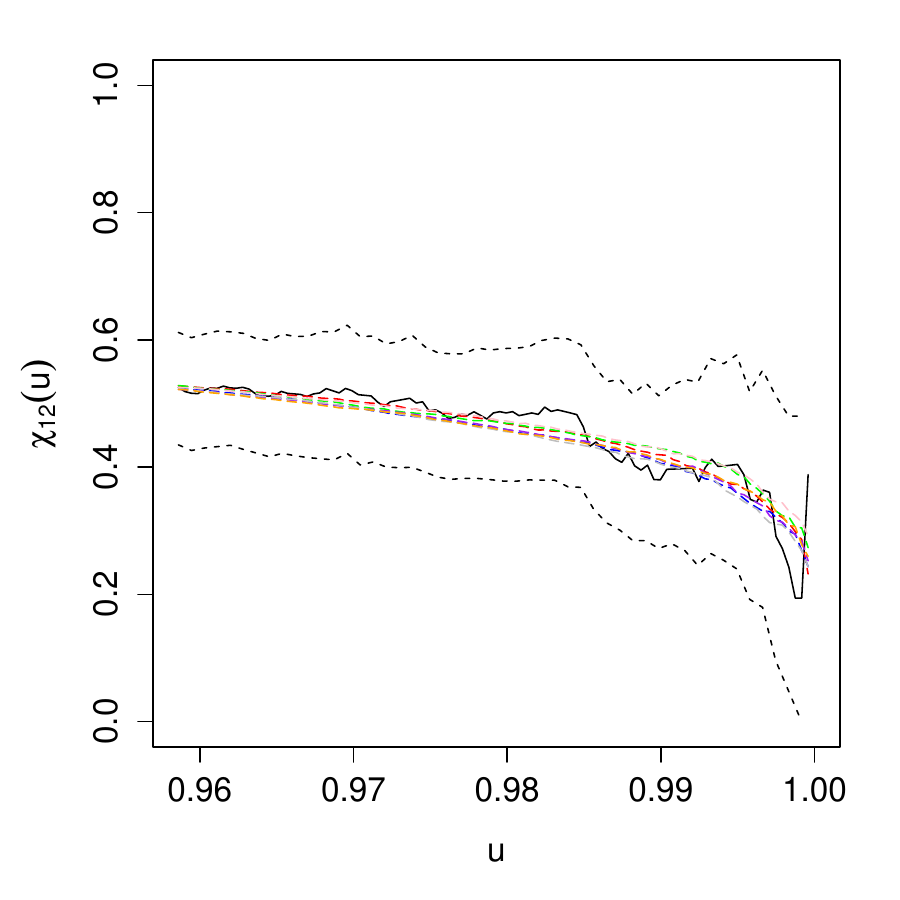}
\includegraphics[width=4.5cm]{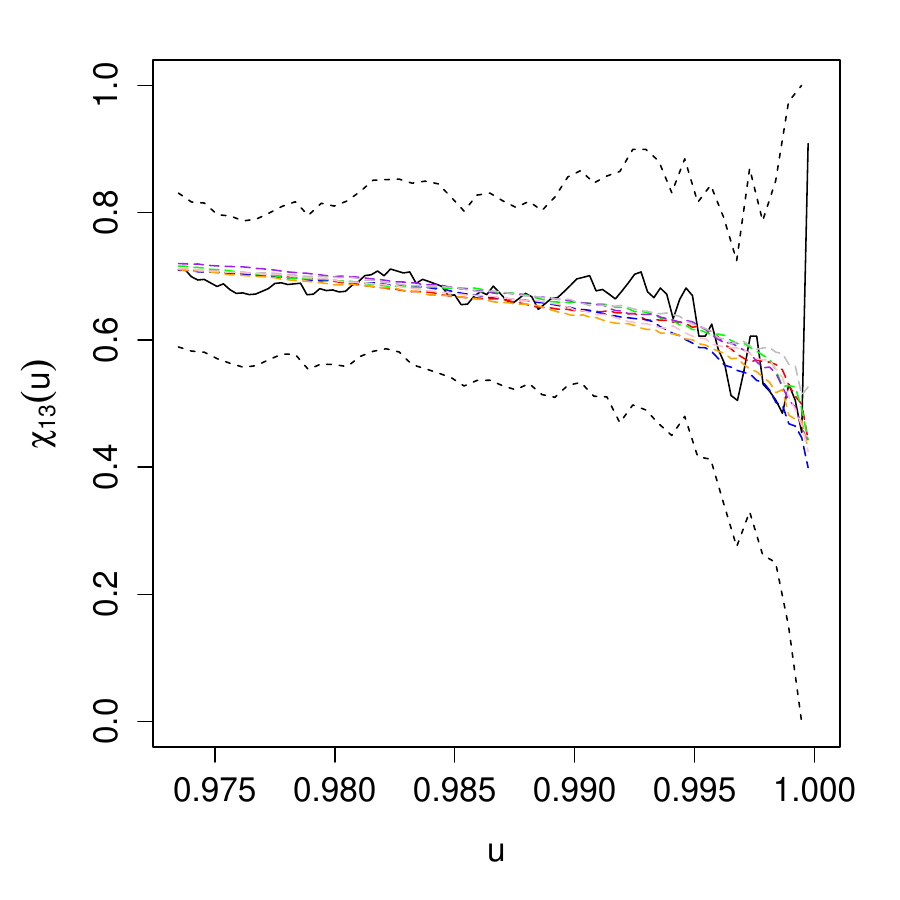}
\includegraphics[width=4.5cm]{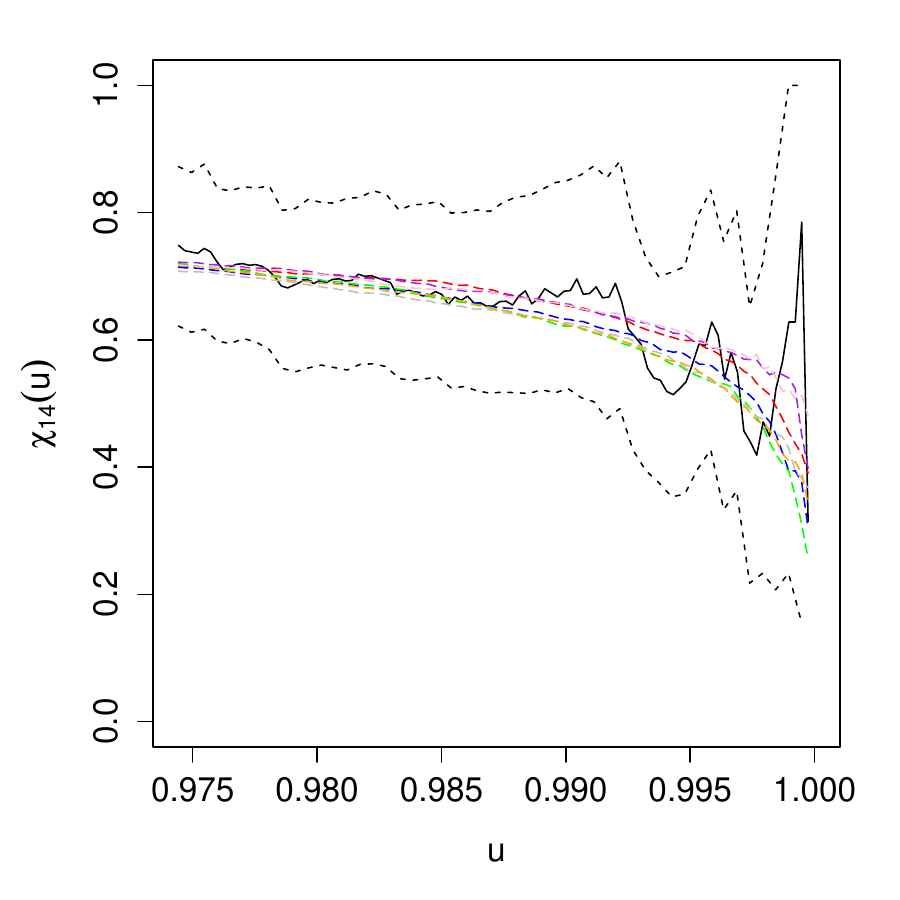}
\includegraphics[width=4.5cm]{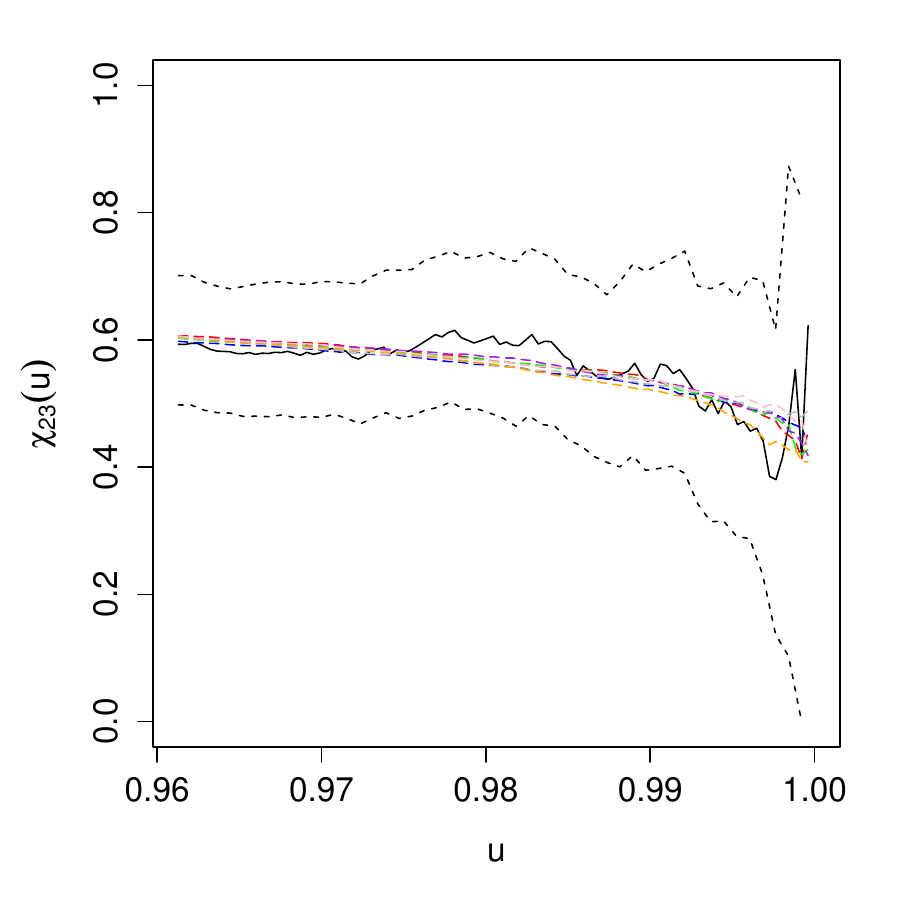}
\includegraphics[width=4.5cm]{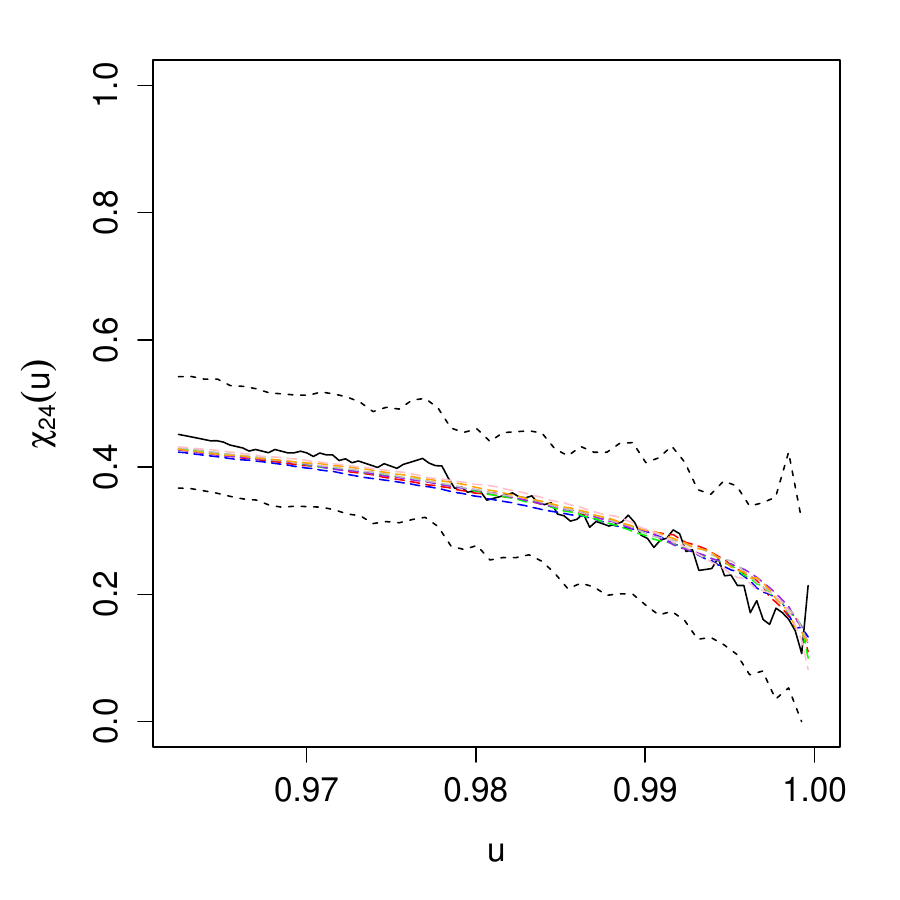}
\includegraphics[width=4.5cm]{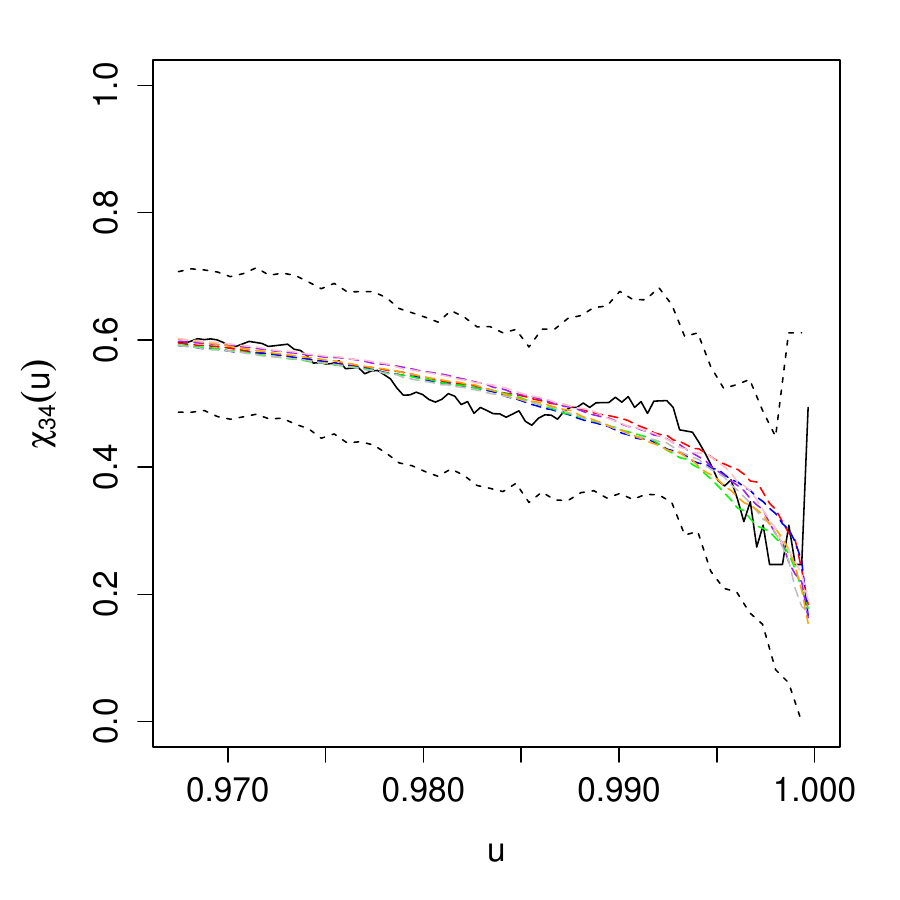}
\caption{\label{fig:chiplot_riverflow} Chi-plots for each pair where the solid black lines represent the empirical tail dependence coefficient $\hat{\chi}(u)$ with 95\% pointwise confidence intervals.
The plots are arranged in the order, from top to bottom and from left to right: $(X_1,X_2)$, $(X_1,X_3)$, $(X_{1},X_4)$, $(X_2,X_3)$, $(X_2,X_4)$, and $(X_3,X_4)$.
The colored dashed lines indicate the model-based estimates $\hat{\chi}_m(u)$ corresponding to different additive mixtures: ExpGa (red), ExpInv (blue), ExpRect (green), GaLog (purple), InvLogLog (grey), RectLog (orange), and MM (pink), respectively.
}
\end{figure}
To account for temporal dependence in the river flow data, we employ the block-bootstrap method to obtain 95\% confidence intervals for $\chi(u)$, using a block length of 20.
The choice is based on the typical persistence of river flows as well as the autocorrelation function.
All mixture models exhibit similar behavior in the joint tails for each variable pair.

For comparison, we fit the model of \cite{huser2019modeling} with a Gaussian dependence structure for $(W_1,W_2)^\top$ and the copula-based model of \cite{wadsworth2017modelling} with a beta distribution for $V$ to each pair of variables with exponential margins, using the same high threshold at the 0.95 quantile. 
The ML estimates under two censoring schemes for the model of \cite{huser2019modeling} and the max-censoring scheme for the model of \cite{wadsworth2017modelling} are summarized in Table~\ref{tab:X23_HWEV}.
Interestingly, the model of \cite{huser2019modeling} indicates asymptotic independence for all variable pairs, with $\hat{\delta}< 1/2$, while the copula model of \cite{wadsworth2017modelling} identifies asymptotic independence only for the pair $(X_1,X_2)$. In the simulation study we saw that in the boundary cases of weakly dependent AD and strongly dependent AI, the HW model tended to select AI, while the EV model tended to select AD; these tendencies are borne out in this example. While the geometric models struggled most in this region as well, they showed no overall bias either way.


\begin{table}[ht!]
    \centering
    \begin{tabular}[t]{|c||c|c|c|c|c|c|} 
      \hline
      & $(X_1,X_2)$ & $(X_1,X_3)$ & $(X_1,X_4)$ & $(X_2,X_3)$ & $(X_2,X_4)$ & $(X_3,X_4)$ \\
      \hline\hline
      $\hat{\delta}_{\text{HW-Cmax}}$ & 0.12 & 0.43 & 0.43 & 0.32 & 0.36 & 0.47 \\
      \hline
      $\hat{\delta}_{\text{HW-Cpar}}$ & 0.32 & 0.40 & 0.41 & 0.35 & 0.41 & 0.47 \\
      \hline
      $\hat{\xi}_{\text{\tiny EV}}$ & --0.03 & 0.16 & 0.18 & 0.13 & 0.01 & 0.24 \\
      \hline
      $\hat{\alpha}_{\text{\tiny EV}}$ & 7.72 & 18.30 & 16.84 & 8.09 & 3.92 & 5.01 \\
      \hline
      \end{tabular}
\caption{Summary of the ML estimates for all variable pairs under two copula-based models. 
`HW-Cmax' and `HW-Cpar' refer to the model of \cite{huser2019modeling} with a Gaussian dependence structure for $(W_1,W_2)$, using the max-censoring and partial-censoring schemes, respectively.
`EV' represents the copula-based model of \cite{wadsworth2017modelling} with a beta distribution for $V$.
The ML estimates of $\hat{\xi}_{EV}$ and $\hat{\alpha}_{EV}$ correspond to the generalized Pareto shape parameter and the beta shape parameter, respectively.
}
\label{tab:X23_HWEV}
\end{table}


\section{Conclusion}
\label{sec:Conclusion}
We have developed practical geometric criteria for identifying the class of extremal dependence under the truncated gamma approximation in Propositions~\ref{prop:pointyAD} and~\ref{prop:AI}.
We considered two types of flexible additively mixed gauge functions, which enable a smooth transition between asymptotic independent and dependent structures.
We derived explicit forms for the additively mixed gauge functions and investigated the properties of these additive mixture models.
The implementation of this geometric methodology with explicit forms is computationally efficient compared to other existing random scale-based methods, which involve integral calculations in the copula likelihood.
The performance of this geometric approach demonstrates reliable accuracy in classifying both asymptotic dependence and asymptotic independence, even in extreme scenarios such as weakly dependent AD and strongly dependent AI, as measured by sensitivity and specificity.
Therefore, the geometric approach with additively mixed gauge functions can be  a simple and efficient alternative for classifying extremal dependence classes.


\section*{Acknowledgments}
This work is supported by UK Engineering and Physical Sciences Research Council (EPSRC) grant EP/X010449/1.

\section*{Data and Code}
The code used in the analysis is available at \texttt{https://github.com/JeongjinLee88/geocrit}.
The river flow data analyzed in Section~\ref{sec:Application} is available in the supplementary material of \cite{wadsworth2024statistical}.

\setlength{\bibsep}{0pt}

\bibliographystyle{chicago}
\bibliography{biblio}

\newpage

\appendix
\section{Appendix}

The geometric criteria for identifying the class of extremal dependence between AD and AI are determined by the behavior of the gauge function along the boundary $\cbr{(x,y):\max(x,y)=1}$.
The gauge function values along this boundary are defined as $k(q)$ and $\tilde{k}(q)$ in \eqref{eq:k(q)}.
Watson's lemma~\citep{watson1918harmonic}, see e.g., \citet[Section 6.4]{bender2013advanced}, is useful for finding the asymptotic behavior of integrals that are in the Laplace form, which is used in proving Proposition~\ref{prop:pointyAD} and Proposition~\ref{prop:AI}.


\subsection{Lemmas and proofs for Section~\ref{sec:Theory}}


\subsubsection{Watson's lemma}
\begin{lemma}
\label{lemma:Watson}
    Let $0<u\le \infty$ be fixed.
    Assume $c(v)$ has the asymptotic series expansion $c(v)\sim v^\eta \sum_{n=0}^\infty a_n v^{\beta n}$ as $v\rightarrow 0^+$, where $\eta > -1$ and $\beta>0$.
    Suppose that $|c(v)| < K_1\exp(K_2v)$ for all $v>0$ and for some positive constants $K_1$ and $K_2$.
    Then, the following integral has the asymptotic expansion: for each integer $N\ge0,$
    \begin{equation*}
        \int_{0}^u c(v)\exp(-xv)\diff v = \sum_{n=0}^{N}\frac{a_n\Gamma(\eta+\beta n+1)}{x^{\eta+\beta n+1}}+\bigO(x^{-\eta-\beta (N+1)-1}),\quad x\rightarrow\infty,
    \end{equation*}
    where $\Gamma$ is the gamma function.
\end{lemma}

\begin{remark}
    We use Watson's lemma to obtain leading order terms of integrals.
    For this purpose, it is sufficient to have $c(v)= a_0 v^{\eta}+\bigO(v^{\eta+\beta})$ as $v\rightarrow 0^+,$ with $\eta>-1$ and $\beta>0$.
    Under the conditions of Lemma~\ref{lemma:Watson},
    we obtain the asymptotic expansion
    $\int_0^\infty c(v)\exp(-xv)\diff v= a_0\Gamma(\eta+1)x^{-\eta-1}+\bigO(x^{-\eta-\beta-1})$ as $x\rightarrow \infty.$
\end{remark}

\subsubsection{Proof for Proposition~\ref{prop:pointyAD}}
\label{appendix:proofThm3.2.}

\begin{proof}
We have
\begin{equation*}
\label{eq:f_tilde}
\begin{split}
    f^\star_{X,Y}(x,y)
    &=
    f_{R|W}\rbr{x+y\,\Bigg\rvert\, \frac{x}{x+y}}f_{W}\rbr{\frac{x}{x+y}}\rbr{\frac{1}{x+y}}\\
    &=
    \frac{\frac{1}{\Gamma(\lambda)}(x+y)^{-2}g_{\bX}\rbr{x,y}^\lambda \exp\cbr{-g_{\bX}(x,y)}}{\widebar{F}_{R|W}\rbr{r_\tau\rbr{\xw};\lambda,g_{\bX}\rbr{\xw,\yw}}}f_{W}\rbr{\frac{x}{x+y}},
\end{split}
\end{equation*}
defined on the support $\cbr{(x,y):x+y > r_\tau\rbr{\xw}}$, where $\widebar{F}_{R|W}\rbr{\,\cdot\,;\lambda,g_{\bX}\rbr{\frac{x}{x+y},\frac{y}{x+y}}}$ is the gamma survival function.

If $X \sim F_X^\star$, $Y \sim F_Y^\star$, and $(X,Y)$ have joint survival function $\widebar{F}^\star(x,y)=\prob(X>x,Y>y)$, then 
\[
\chi = \lim_{u \to 1} \frac{\widebar{F}^\star\rbr{F_X^{\star-1}(u),F_Y^{\star-1}(u)}}{1-u}.
\]
To shorten notation, let $x_u=F_{X}^{\star-1}(u)$, $y_u=F_{Y}^{\star-1}(u)$. Notice that we have the inequality
\[
\frac{\widebar{F}^\star\rbr{\max(x_u,y_u),\max(x_u,y_u)}}{1-u} \leq  \frac{\widebar{F}^\star(x_u,y_u)}{1-u} \leq \frac{\widebar{F}^\star\rbr{\min(x_u,y_u),\min(x_u,y_u)}}{1-u}.
\]
For the purposes of these calculations, and without loss of generality, suppose that $x_u \leq y_u$, at least for $u$ sufficiently close to 1, i.e., $x_u=\min(x_u,y_u)$ and $y_u=\max(x_u,y_u)$. This gives
\[
\frac{\widebar{F}^\star(y_u,y_u)}{1-F_{Y}^\star(y_u)} = \frac{\widebar{F}^\star(y_u,y_u)}{1-u} \leq  \frac{\widebar{F}^\star(x_u,y_u)}{1-u} \leq \frac{\widebar{F}^\star(x_u,x_u)}{1-u}=\frac{\widebar{F}^\star(x_u,x_u)}{1-F_{X}^\star(x_u)}.
\]
Consider the limits of the lower and upper bounds. Using the chain rule and L'H\^{o}pital's rule, these can be expressed as:
\begin{align*}
 \chi_{\text{lower}} =\lim_{u\to 1} \frac{\widebar{F}^\star(y_u,y_u)}{1-F_{Y}^\star(y_u)} &= \lim_{x \to \infty} \frac{\widebar{F}^\star(x,x)}{1-F_{Y}^\star(x)} \\&= \lim_{x \to \infty} -\widebar{F}^\star_1(x,x)/f_{Y}^\star(x) -\widebar{F}^\star_2(x,x)/f_{Y}^\star(x);\\
 \chi_{\text{upper}} =\lim_{u\to 1} \frac{\widebar{F}^\star(x_u,x_u)}{1-F_{X}^\star(x_u)} &= \lim_{x \to \infty} \frac{\widebar{F}^\star(x,x)}{1-F_{X}^\star(x)} \\&= \lim_{x \to \infty} -\widebar{F}^\star_1(x,x)/f_{X}^\star(x) -\widebar{F}^\star_2(x,x)/f_{X}^\star(x),
\end{align*}
where $\widebar{F}^\star_1(x,y) = \partial\widebar{F}^\star(x,y)/\partial x$, $\widebar{F}^\star_2(x,y) = \partial\widebar{F}^\star(x,y)/\partial y$. In terms of the joint density $f_{X,Y}^\star$, we have
\[
\widebar{F}^\star_1(x,x) = -\int_{x}^\infty f_{X,Y}^\star(x,y) \,\diff y, \qquad \widebar{F}^\star_2(x,x) = -\int_{x}^\infty f_{X,Y}^\star(y,x) \,\diff y,
\]
while the marginals are
\begin{align*}
 f_{X}^\star(x) &= \int_0^\infty f_{X,Y}^\star(x,y) \, \diff y = \int_0^x f_{X,Y}^\star(x,y) \, \diff y + \int_x^\infty f_{X,Y}^\star(x,y) \, \diff y\\
  f_{Y}^\star(x) &= \int_0^\infty f_{X,Y}^\star(y,x) \, \diff y = \int_0^x f_{X,Y}^\star(y,x) \, \diff y + \int_x^\infty f_{X,Y}^\star(y,x) \, \diff y.
\end{align*}
Due to the support of $f^\star_{X,Y}$, these limits of integration hold only for large enough $x$, but since we are interested in $x\to\infty$, this is sufficient.

Note that $x_u \leq y_u$ implies $F_X^\star(x) \ge F_Y^\star(x)$ and hence $\int_x^\infty f_{X}^\star(a) \,\diff a = 1-F_X^\star(x) \leq 1-F_Y^\star(x) = \int_x^\infty f_{Y}^\star(a) \,\diff a$, giving $f_{X}^\star(x) \leq f_{Y}^\star(x)$.
Returning to general notation around which of $x_u$ and $y_u$ is smaller, we have 
\begin{align*}
 \chi_{\text{lower}} &= \lim_{x \to \infty} -[\widebar{F}^\star_1(x,x)+\widebar{F}^\star_2(x,x)]/\max\{f_{X}^\star(x),f_{Y}^\star(x)\},\\
  \chi_{\text{upper}} &= \lim_{x \to \infty} -[\widebar{F}^\star_1(x,x)+\widebar{F}^\star_2(x,x)]/\min\{f_{X}^\star(x),f_{Y}^\star(x)\}.
 \end{align*}

We first focus on the integral 
$\int_{0}^{x}f^\star_{X,Y}(x,y)\diff y$.
With the substitution $y=xq$, where $q\in[0,1]$, the integral becomes
\begin{equation}
\label{eq:Integral_01}
\begin{split}
    \int_{0}^{x}f^\star_{X,Y}(x,y)\diff y
    &=
    \int_{0}^{1} x^{\lambda-1} \exp\cbr{-xg_{\bX}(1,q)}\frac{\frac{1}{\Gamma(\lambda)}(1+q)^{-2}g_{\bX}\rbr{1,q}^\lambda f_{W}\rbr{\frac{1}{1+q}}}{\widebar{F}_{R|W}\rbr{r_\tau\rbr{\frac{1}{1+q}};\lambda,g_{\bX}\rbr{\frac{1}{1+q},\frac{q}{1+q}}}} \diff q\\
    &=
    x^{\lambda-1} \int_{0}^{1} \exp\cbr{-xk(q)}b(q)\diff q
    =: \mathcal{I}_{[0,1]}(x),
\end{split}
\end{equation}
with 
\[
    k(q):=g_{\bX}(1,q)
\]
and
\[
    b(q):=\frac{\frac{1}{\Gamma(\lambda)}(1+q)^{-2}k(q)^\lambda f_W\rbr{\frac{1}{1+q}}}{\widebar{F}\rbr{r_\tau\rbr{\frac{1}{1+q}};\lambda,g_{\bX}\rbr{\frac{1}{1+q},\frac{q}{1+q}}}}.
\]
We decompose the integral $\mathcal{I}_{[0,1]}$ into the sum of disjoint integrals over sub-intervals $0=q_{1}<\cdots<q_{n}=1$ such that for any two adjacent points $q_i < q_{i+1}$, the function $k(q)$ satisfies the following conditions
\begin{enumerate}
    \item $k(q)$ is bijective on $[q_{i},q_{i+1}]$ with two cases:
    \begin{itemize}
    \item $k(q)$ is strictly decreasing on $[q_i,q_{i+1}]$ such that $k(q_i)>1$, $k(q_{i+1})=1$, and $k'(q_{i+1}^-)\ne 0$.
    \item $k(q)$ is strictly increasing on $[q_i,q_{i+1}]$ such that $k(q_i)=1$, $k(q_{i+1})>1$, and $k'(q_{i}^+)\ne 0$.
    \end{itemize}
    \item Non-injective but surjective case: $k(q)$ takes a constant value greater than $1$ over $[q_i,q_{i+1}]$, implying that $k(q)$ does not intersect the boundary set $\cbr{(x,y):\max(x,y)=1}.$
\end{enumerate} 
To classify these cases, recall the three disjoint index sets with $A\cup B \cup C=\{1,\ldots,n-1\}$ defined in equation~\eqref{eq:indexSets}.
The integral can be split as follows
\begin{align*}
    \mathcal{I}_{[0,1]}(x)
    &=
    \sum_{i\in A}\int_{q_i}^{q_{i+1}}x^{\lambda-1}\exp\cbr{-xk(q)}b(q)\diff q
    +
    \sum_{i\in B}\int_{q_i}^{q_{i+1}}x^{\lambda-1}\exp\cbr{-xk(q)}b(q)\diff q\\
    &+
    \sum_{i\in C}\int_{q_i}^{q_{i+1}}x^{\lambda-1}\exp\cbr{-xm}b(q)\diff q. 
\end{align*}
It suffices to consider a sub-integral of the form for each index set
\begin{equation*}
    \mathcal{I}_{i}(x):=x^{\lambda-1}\int_{q_i}^{q_{i+1}}\exp\cbr{-xk(q)}b(q)\diff q,\quad i\in\cbr{1,\ldots,n-1}.
\end{equation*}
For large $x$, the exponential term in the integral dominates, with $\exp\cbr{-xk(q)}$ attaining its maximum where $k(q)$ is minimized, i.e., at the boundary intersection point.
Without loss of generality, the integral can be expressed in a form suitable for application of Watson's Lemma (Lemma~\ref{lemma:Watson}).

For $i\in A,$ let $v=k(q)-1$, where $v$ has an inverse $q=k^{-1}(v+1)$ on $[q_i,q_{i+1}]$ as $k(q)$ is strictly decreasing. We can express the integral as
\begin{equation}
    \mathcal{I}_{i}(x)=-x^{\lambda-1}\exp(-x)\int_{0}^{k(q_i^+)-1}c(v)\exp\cbr{-xv}\diff v, \label{eq:Ii}
\end{equation}
with $c(v):=\frac{b\rbr{k^{-1}(v+1)}}{k'\rbr{k^{-1}(v+1)}}$, which is given in full as
\begin{equation}
\label{eq:c(v)}
    c(v)=\frac{1}{k'(k^{-1}(v+1))}\frac{(v+1)^\lambda\frac{1}{\Gamma(\lambda)}\rbr{1+k^{-1}(v+1)}^{-2}f_W\rbr{\frac{1}{1+k^{-1}(v+1)}}}{\widebar{F}\rbr{r_{\tau}\rbr{\frac{1}{1+k^{-1}(v+1)}};\lambda,g_{\bX}\rbr{\frac{1}{1+k^{-1}(v+1)},\frac{k^{-1}(v+1)}{1+k^{-1}(v+1)}}}}.
\end{equation}
By Assumptions~\ref{ass:kq} and~\ref{ass:b}, the function $c$ has one-sided derivatives of up to order two and is bounded. Since $\lim_{v\rightarrow k(q_{i+1}^-)-1}c(v)=\frac{b(q_{i+1})}{k'(q_{i+1}^-)}\ne 0$, we thus have
\[
c(v) = \frac{b(q_{i+1})}{k'(q_{i+1}^-)} + c'(0^+)v+ \bigO(v^2), \qquad v\to 0^+.
\]
Applying Watson's lemma to $\mathcal{I}_i(x)$ then gives the asymptotic behavior of the integral. Specifically,
\begin{equation*}
    \label{eq:Watson}
    \int_{0}^{k(q_i^+)-1}c(v)\exp\cbr{-xv}\diff v = \rbr{\frac{b(q_{i+1}^-)/k'(q_{i+1}^-)}{x}} + \bigO(x^{-2}),\quad x\rightarrow\infty.
\end{equation*}
Consequently, the leading-order term of the integral $\mathcal{I}_{i}(x)$ is
\begin{align*}
    \mathcal{I}_i(x)
    =-x^{\lambda-2}\exp(-x)\rbr{\frac{b(q_{i+1}^-)}{k'(q_{i+1}^-)} +\bigO(x^{-1})}.
\end{align*}
Similarly to the case $i\in A$, for $i\in B$, the first-order asymptotic term of the integral is
\begin{align*}
    \mathcal{I}_{i}(x)
    &=x^{\lambda-2}\exp(-x)\rbr{\frac{b(q_{i}^+)}{k'(q_{i}^+)} +\bigO(x^{-1})}.
\end{align*}
For $i\in C$ and for fixed $m>1$, the sub-integral simplifies to
\begin{equation*}
    \mathcal{I}_{i}(x)
    =
    x^{\lambda-1}\exp\cbr{-mx}\int_{q_i}^{q_{i+1}}b(q)\diff q,
\end{equation*}
with $\int_{q_i}^{q_{i+1}}b(q)\diff q < \infty.$
Putting this together, we get
\begin{align*}
\mathcal{I}_{[0,1]}(x) \sim x^{\lambda-2}\exp(-x) \left[- \sum_{i \in A} \frac{b(q_{i+1})}{k'(q_{i+1}^-)} +\sum_{i\in B} \frac{b(q_{i})}{k'(q_{i}^+)} \right].
\end{align*}

We now evaluate the integral $\int_{x}^{\infty} f^\star_{X,Y}(x,y)\diff y$.
By substituting $y=qx$, where $q\ge 1$, we have
\begin{equation*}
\begin{split}
    \int_{x}^{\infty} f^\star_{X,Y}(x,y)\diff y
    =
    x^{\lambda-1}\int_{1}^{\infty} \exp\cbr{-xk(q)}b(q) \diff q=:\mathcal{I}_{[1,\infty)}(x),
\end{split}
\end{equation*}
where $k(q)$ and $b(q)$ are the same functions as those defined inside the integral~\eqref{eq:Integral_01}. Let again $v=k(q)-1$ for $q\ge 1$ and hence $v\ge 0$.
The integral is then expressed as
\begin{equation*}
    \mathcal{I}_{[1,\infty)}(x)
    =
    x^{\lambda-1}\exp(-x)\int_{0}^{\infty}\frac{b\rbr{k^{-1}(v+1)}}{k'(k^{-1}(v+1))}\exp(-xv)\diff v,
\end{equation*}
Note that the dominant contribution to the full integral $\mathcal{I}_{[1,\infty)}$ arises at the minimum $v=0$, where the exponential decay is slowest.
As such, the integral evaluated at $v=0$ determines the leading-order asymptotic contribution. For some choices of $g_{\bX}$, equivalently $k$, it may only make sense to define the change of variables $v=k(q)-1$ on an interval $[1,q_{n+1}]$, with $q_{n+1}>1$, on which $k$ is monotonic. However, since $k(q)>q$, we do not have to worry about contributions to the asymptotic form from integrals over any other intervals. Without loss of generality, we therefore analyze the integral $\mathcal{I}_{[1,\infty)}$ over the entire interval $q\in[1,\infty)$ as a whole. Since $\lim_{v\rightarrow k(1^+)-1}c(v)=\frac{b(1)}{k'(1^+)}\ne 0$, Watson's lemma again yields 
\begin{equation*}
\label{eq:int_1toinfity}
    \mathcal{I}_{[1,\infty)}(x)
    =x^{\lambda-2}\exp(-x)\rbr{\frac{b(1^+)}{k'(1^+)} +\bigO(x^{-1})},
\end{equation*}
where $k'(1^+)=\lim_{\epsilon \to 0^+} (k(1+\epsilon)-k(1))/\epsilon \geq (1+\epsilon-1)/\epsilon=1$, since $k(1)=1$ and $k(q) \geq q$.
Applying the same arguments to $\tilde{\mathcal{I}}_{[0,1]}(x) :=\int_0^x f_{X,Y}^\star(y,x)\,\diff y$ and $\tilde{\mathcal{I}}_{[1,\infty)}(x) :=\int_x^\infty f_{X,Y}^\star(y,x)\,\diff y$, with $f_{X,Y}^\star(xq,x) = x^{\lambda-2}\exp\{-x\tilde{k}(q)\}\tilde{b}(q)$, we get
\begin{align}
\chi_{\text{lower}} &= \lim_{x\to\infty} \frac{\mathcal{I}_{[1,\infty)}(x)+\tilde{\mathcal{I}}_{[1,\infty)}(x)}{\max\left[\mathcal{I}_{[0,1]}(x)+\mathcal{I}_{[1,\infty)}(x),\tilde{\mathcal{I}}_{[0,1]}(x)+\tilde{\mathcal{I}}_{[1,\infty)}(x)\right]} \notag \\
&=\frac{b(1)/k'(1^+) + \tilde{b}(1)/\tilde{k}'(1^+)}{ \max\left[- \sum_{i \in A} \frac{b(q_{i+1})}{k'(q_{i+1}^-)} +\sum_{i\in B} \frac{b(q_{i})}{k'(q_{i}^+)} + \frac{b(1)}{k'(1^+)} , - \sum_{i \in \tilde{A}} \frac{\tilde{b}(\tilde{q}_{i+1})}{\tilde{k}'(\tilde{q}_{i+1}^-)} +\sum_{i\in\tilde{B}} \frac{\tilde{b}(\tilde{q}_{i})}{\tilde{k}'(\tilde{q}_{i}^+)} + \frac{\tilde{b}(1)}{\tilde{k}'(1^+)} \right]}\notag \\
\chi_{\text{upper}} &= \lim_{x\to\infty} \frac{\mathcal{I}_{[1,\infty)}(x)+\tilde{\mathcal{I}}_{[1,\infty)}(x)}{\min\left[\mathcal{I}_{[0,1]}(x)+\mathcal{I}_{[1,\infty)}(x),\tilde{\mathcal{I}}_{[0,1]}(x)+\tilde{\mathcal{I}}_{[1,\infty)}(x)\right]} \label{eq:chiupper}\\
&=\frac{b(1)/k'(1^+) + \tilde{b}(1)/\tilde{k}'(1^+)}{ \min\left[- \sum_{i \in A} \frac{b(q_{i+1})}{k'(q_{i+1}^-)} +\sum_{i\in B} \frac{b(q_{i})}{k'(q_{i}^+)} + \frac{b(1)}{k'(1^+)} , - \sum_{i \in \tilde{A}} \frac{\tilde{b}(\tilde{q}_{i+1})}{\tilde{k}'(\tilde{q}_{i+1}^-)} +\sum_{i\in\tilde{B}} \frac{\tilde{b}(\tilde{q}_{i})}{\tilde{k}'(\tilde{q}_{i}^+)} + \frac{\tilde{b}(1)}{\tilde{k}'(1^+)} \right]}.\notag
\end{align}

\end{proof}

\subsubsection{Proof for Proposition~\ref{prop:AI}.}
\label{appendix:proofAI}

\begin{proof}
    We follow a similar argument as in Proposition~\ref{prop:pointyAD}, now showing that $\chi_{\text{upper}} = 0$, implying $\chi=0$.
    The case of interest is $k(1)=\tilde{k}(1)=1$.
    In the first two cases the limit set implies AI; in the third case, AI is implied by the form of the density $f_W(w)$ in conjunction with the limit set shape.

    In all cases we consider $\mathcal{I}_{[0,1]}$ and $\mathcal{I}_{[1,\infty)}$, with the same arguments applying to $\tilde{\mathcal{I}}_{[0,1]}$ and $\tilde{\mathcal{I}}_{[1,\infty)}$. Furthermore, $\mathcal{I}_{[1,\infty)}$ has the same form identified in the proof of Proposition~\ref{prop:pointyAD}, so in each of the three cases we focus on the form of $\mathcal{I}_{[0,1]}$.
        
    \begin{enumerate}
    \item Assumption~\ref{ass:b} holds and there exist sets $D,\Tilde{D}\subseteq[0,1]$ (not necessarily identical) with $|D|$, $|\tilde{D}|>0$ such that $k(q)=1$ for all $q \in D$, and $\tilde{k}(q)=1$ for all $q \in \Tilde{D}$.
    \end{enumerate}

    Consider the integral $\mathcal{I}_{[0,1]}(x)=\int_{0}^{1}x^{\lambda-1}\exp\rbr{-xk(q)}b(q)\diff q$. 
    The sub-integral over $D\subseteq[0,1]$, where $k(q)=\1_{D}(q)$, is then
    \begin{align*}
        \mathcal{I}_{D}(x)
        :=
        \int_{D}x^{\lambda-1}\exp(-x)b(q)\diff q 
        =
        x^{\lambda-1}\exp(-x)\int_{D}b(q)\diff q.
    \end{align*}
    From Proposition~\ref{prop:pointyAD}, we also have 
    \begin{align*}
        \text{For}\,\,i\in A\,\,\text{and}\,\,[q_i,q_{i+1}]\subset[0,1]\setminus D,\quad \mathcal{I}_i(x)
        &=
        -x^{\lambda-2}\exp(-x)\rbr{b(q_{i+1}^-)/k'(q_{i+1}^-) +\bigO(x^{-1})}\\
        \text{For}\,\,i\in B\,\,\text{and}\,\,[q_i,q_{i+1}]\subset[0,1]\setminus D,\quad\mathcal{I}_i(x)
        &=
        x^{\lambda-2}\exp\rbr{-x}\rbr{b(q_{i}^+)/k'(q_i^+)+\bigO(x^{-1})}\\
        \text{For}\,\,i\in C\,\,\text{and}\,\,[q_i,q_{i+1}]\subset[0,1]\setminus D,\quad\mathcal{I}_i(x)
        &=
        x^{\lambda-1}\exp(-mx)\int_{q_i}^{q_{i+1}}b(q)\diff q,\quad m>1.
    \end{align*}
Putting these together, $\mathcal{I}_{[0,1]}(x) \sim x^{\lambda-1}\exp(-x)\int_{D}b(q)\diff q$, and similarly $\tilde{\mathcal{I}}_{[0,1]}(x) \sim  x^{\lambda-1}\exp(-x)\int_{D}\tilde{b}(q)\diff q$. Therefore taking the limit in equation~\eqref{eq:chiupper}, $\chi_{\text{upper}}=0$.

\begin{enumerate}[resume]
        \item Assumption~\ref{ass:b} holds and there exist boundary intersection points $q_i,\tilde{q}_i\in[0,1]$ (not necessarily identical), with $b(q_i), \tilde{b}(\tilde{q}_i) \neq 0$, such that 
        \begin{align*}
        \begin{cases}
            k'(q_{i+1}^-)=0,\quad \tilde{k}'(\tilde{q}_{i+1}^-)=0,\quad i \in A,\\
            k'(q_{i}^+)=0,\quad \tilde{k}'(\tilde{q}_i^+)=0,\quad i\in B,
        \end{cases}
        \end{align*}
        implying vertical and horizontal tangents along the boundary $\cbr{(x,y):\max(x,y)=1}$. 
\end{enumerate}
Again we focus on $\mathcal{I}_{[0,1]}$, and specifically $\mathcal{I}_i(x)$ as defined in equation~\eqref{eq:Ii}, and without loss of generality, consider $i\in A$. 
We have
\begin{equation*}
        \mathcal{I}_{i}(x)=-x^{\lambda-1}\exp(-x)\int_{0}^{k(q_i^+)-1}c(v)\exp\cbr{-xv}\diff v, \qquad c(v)=\frac{b(k^{-1}(v+1))}{k'(k^{-1}(v+1))},
\end{equation*}
where $c(v)$ is defined in~\eqref{eq:c(v)}.
When $v \to 0^+$, $b(k^{-1}(v+1)) \to b(q_{i+1})\in (0,\infty)$. Using the asymptotic relation given in Assumption~\ref{ass:kq}, we have \[
1/k'(k^{-1}(v+1))=-\rho_k^{-1} a_k^{-1/\rho_k} v^{1/\rho_k-1}\rbr{1-\frac{b_k(\epsilon+1)}{a_k\rho_k}(v/a_k)^{\epsilon/\rho_k}+\bigO(v^{2\epsilon/\rho_k})},\quad v\rightarrow0^+,\]
and thus
 \begin{equation*}
     c(v)=-b(q_{i+1})\rho_k^{-1}a_k^{-1/\rho_k}v^{1/\rho_k-1}\rbr{1-\frac{b_k(\epsilon+1)}{a_k \rho_k}(v/a_k)^{\epsilon/\rho_k}+\bigO(v^{2\epsilon/\rho_k})},\quad v\rightarrow0^+.
 \end{equation*}

We can apply Watson's lemma with $\eta=1/\rho_k-1 \in (-1,0)$, and $a_0 =-b(q_{i+1})a_k^{-1/\rho_k}\rho_k^{-1} \neq 0$. This yields
\begin{equation*}
    \mathcal{I}_i(x)=-x^{\lambda-1/\rho_{k}-1}\exp(-x)\rbr{a_0\Gamma(1/\rho_{k})+\bigO\rbr{x^{-\epsilon/\rho_k}}},
\end{equation*}
for some $\epsilon/\rho_k>0$.
If there are multiple intersection points $q_i$ with different indices $\rho_{k,i}>1$, then $\mathcal{I}_{[0,1]}(x) \sim -x^{\lambda-\eta-2}\exp(-x)[a_0\Gamma(\eta+1)]$, with $\eta=1/\max_i(\rho_{k,i})-1 \in (-1,0)$, and $a_0$ is the constant corresponding to the $q_i$ with largest $\rho_{k,i}$, or the sum of all such constants in the case of ties. 
We similarly have $\tilde{\mathcal{I}}_{[0,1]}(x) \sim -x^{\lambda-\tilde{\eta}-2}\exp(-x)\tilde{a}_0\Gamma(\tilde{\eta}+1)$ with $\tilde{\eta}=1/\max_i(\rho_{\tilde{k},i})-1 \in (-1,0)$, and $\tilde{a}_0$ defined analogously to $a_0$. Using limit~\eqref{eq:chiupper}, this again yields $\chi_{\text{upper}}=0$.

\begin{enumerate}[resume]
        \item  If $k(0)=\tilde{k}(0)=1$, Assumption~\ref{ass:b} holds, except $b$, $\tilde{b}$ are infinite at $0$, with $b(q)\sim a_b q^{\rho_b}$
        as $q \to 0$, with $a_b>0$ and $\rho_b \in (-1,0)$ (similarly for $\tilde{b}$).
\end{enumerate}
In this final case we consider $i=1\in B$ such that $q_1=0$. The integral of interest is
 \begin{equation*}
        \mathcal{I}_{i}(x)=x^{\lambda-1}\exp(-x)\int_{0}^{k(q_{i+1}^-)-1}c(v)\exp\cbr{-xv}\diff v, \qquad c(v)=\frac{b(k^{-1}(v+1))}{k'(k^{-1}(v+1))}.
\end{equation*}
By assumption, $b(q) \sim a_b q^{\rho_b}, q \to 0$, with $\rho_b \in (-1,0)$. 
Using Assumption~\ref{ass:kq}, we have $v=k(q)-1=a_{k}q^{\rho_k}+\bigO(q^{\rho_k+\epsilon})$ with index $\rho_k \geq 1$ for some $\epsilon>0$: if $\rho_k=1$ then $a_k=k'(0^+)>0$, i.e., there may not be a vanishing first derivative. 
This yields 
\[
    c(v)= a_b a_k^{-(\rho_b+1)/\rho_k}\rho_k^{-1} v^{(\rho_b+1)/\rho_k-1}\rbr{1-\frac{b_k(\rho_b+\epsilon+1)}{a_k\rho_k}(v/a_k)^{\epsilon/\rho_k}+\bigO\rbr{v^{2\epsilon/\rho_k}}},\quad v\rightarrow 0^+,
\]
where $(\rho_b+1)/\rho_k-1 \in (-1,0)$. Watson's lemma again yields 
\begin{equation*}
    \mathcal{I}_i(x)=-x^{\lambda-\eta-2}\exp(-x)\rbr{a_0\Gamma(\eta+1)+\bigO\rbr{x^{-\epsilon/\rho_k}}},
\end{equation*}
with $\eta=(\rho_b+1)/\rho_k-1\in(-1,0)$, $a_0=a_b a_k^{-(\rho_b+1)/\rho_k}\rho_k^{-1}$, and the remainder of the argument the same as in part 2.
\end{proof}

\subsection{Lemmas and Proofs for Section~\ref{sec:additive}}

In the following lemmas and proofs, we analyze gauge functions as defined in~\eqref{eq:gauge_HW}, constructed through an additive stochastic representation. In Section~\ref{sec:phi}, we give a key lemma on the behavior of partial derivatives of convex, symmetric $g_{\bV}$, which is useful throughout. In Section~\ref{sec:szero}, we outline values of $\gamma$ for which the objective function $h$ in equation~\eqref{eq:objectiveft} is always minimized at zero. In Section~\ref{sec:prooflemx0y0} we prove Lemma~\ref{lemma:point(x_0,y_0)}, while in Section~\ref{sec:proofproptangent}, we prove Proposition~\ref{prop:tangentline}.

\subsubsection{Behavior of the sum of partial derivatives}
\label{sec:phi}
Define the sum of the partial derivatives of $g_{\bV}$ as 
\begin{equation*}
\label{eq:partialsum}
    H(u,v)=g_{V,1}(u,v)+g_{V,2}(u,v),\quad (u,v)\in[0,\infty)^2,
\end{equation*}
which is homogeneous of order 0. Because of this, it depends only on the ratio $z=u/v$, and we therefore define
    \begin{equation}
    \label{eq:phi(z)}
     \phi(z):=H(z,1)=g_{V,1}(z,1)+g_{V,2}(z,1)\quad z\ge0.   
    \end{equation}
The following lemma describes the behavior of $\phi$ and is useful in the proof of Proposition~\ref{prop:tangentline} .
\begin{lemma}
    \label{lem:phi}
 For convex and twice (piecewise) differentiable $g_{\bV}$, the function $\phi(z)$ defined in equation~\eqref{eq:phi(z)} is maximized at $z=1$. It is non-decreasing on the range $[0,1]$, and non-increasing on the range $[1,\infty)$.   
\end{lemma}
\begin{proof}
The first derivative is $\phi'(z)=g_{V,11}(z,1)+g_{V,21}(z,1).$ By Euler's homogeneous function theorem, we also have that
    \begin{align*}
    zg_{V,11}(z,1)+g_{V,21}(z,1) &= 0 &\Rightarrow &&  g_{V,21}(z,1) = -zg_{V,11}(z,1),
    \end{align*}
so that 
    \begin{align*}
 \phi'(z) = g_{V,11}(z,1)(1-z),
    \end{align*}
and hence $\phi'(1)=0$. Convexity of $g_{V}$ implies $g_{V,11}(z,1) \geq 0$.  Therefore $\phi'(z) \geq 0$ for $z\le 1$ and $\phi'(z) \leq 0$ for $z\ge 1$, from which it follows that $\phi(z)\leq \phi(1)$ for all $z \geq 0$.
\end{proof}

\subsubsection{When the minimizer $\hat{s}=0$}
\label{sec:szero}
\begin{lemma}
\label{lem:mins0}
    For convex and twice (piecewise) differentiable $g_{\bV}$, when $\gamma \leq 1/g_{\bV}(1,1)$, the function $h(s) = s + g_{\bV}(x-\gamma s, y- \gamma s)$ is minimized over $s\in[0,\min(x,y)/\gamma]$, at $\hat{s}=0$.
\end{lemma}
\begin{proof}
The derivative of $h$ is
\begin{align*}
h'(s) = 1-\gamma\left(g_{V,1}(x-\gamma s,y-\gamma s)+g_{V,2}(x-\gamma s,y-\gamma s)\right) = 1-\gamma \phi\left(\frac{x-\gamma s}{y-\gamma s}\right). 
\end{align*}
Suppose $y \geq x$, so that $(x-\gamma s)/(y-\gamma s) \leq 1$. We have that $\phi(1) = g_{V,1}(1,1)+g_{V,2}(1,1) = g_{\bV}(1,1) \leq 1/\gamma$. By Lemma~\ref{lem:phi}, $\phi$ is non-decreasing over $[0,1]$, so $\phi\left(\frac{x-\gamma s}{y-\gamma s}\right) \leq 1/\gamma$ and hence $h'(s) \geq 0$. There are three cases to consider:
\begin{enumerate}
    \item[(i)] If $\phi(1)=g_{\bV}(1,1)<1/\gamma$, then $\phi\left(\frac{x-\gamma s}{y-\gamma s}\right) < 1/\gamma$, so $h'(s)>0$ on $[0,\min(x,y)/\gamma]$ and hence $\hat{s}=0$.
    \item[(ii)] If $\phi(1)=g_{\bV}(1,1)=1/\gamma$ and $\phi$ is strictly increasing, then $\phi\left(\frac{x-\gamma s}{y-\gamma s}\right) < 1/\gamma$ for $y>x$, which implies $h'(s)>0$ and hence $\hat{s}=0$.
    If $x=y$, $h'(s)=1-\gamma/\gamma=0$ does not depend on $s$, so we can take $\hat{s}=0$.
    \item[(iii)] If there is an interval $(z_\star,1]$ on which $\phi(z)=1/\gamma$, we have
  \[
  \phi\left(\frac{x-\gamma s}{y-\gamma s}\right) = 1/\gamma, \qquad \text{for}\,\,\,\,\frac{x-\gamma s}{y-\gamma s} \in (z_\star,1].
  \]
  Any value of $s$ such that $\frac{x-\gamma s}{y-\gamma s} \in (z_\star,1]$ is a solution, so in particular we can take $\hat{s}=0$.
\end{enumerate}
\end{proof}

\subsubsection{Proof for Lemma~\ref{lemma:point(x_0,y_0)}}
\label{sec:prooflemx0y0}
\begin{proof}
    Consider firstly the continuously differentiable case.
    To determine the point $(x_0,y_0)$, we first establish the conditions that define it.
    The tangent line to the unit level curve $g_{\bV}(x,y)=1$ at a generic point $(x^\star,y^\star)$ on the curve $g_{\bV}(x^\star,y^\star)=1$ is
    \begin{align}
     (x-x^\star)g_{V,1}(x^\star,y^\star) + (y-y^\star)g_{V,2}(x^\star,y^\star) &= 0 &\Rightarrow &&  xg_{V,1}(x^\star,y^\star)+yg_{V,2}(x^\star,y^\star)=1. \label{eq:tang}
    \end{align}
    The condition that this tangent line passes through the point $(\gamma,\gamma)$ implies
    \begin{equation}
        \gamma\rbr{g_{V,1}(x^\star,y^\star)+g_{V,2}(x^\star,y^\star)}=1. \label{eq:tanggamma}
    \end{equation}
    We write this equation as 
        \begin{equation*}
       \psi(x^\star/y^\star):= \gamma \phi(x^\star/y^\star)-1=0,
    \end{equation*}
with $x^\star/y^\star \leq 1$. We therefore seek to show that under the conditions of convexity and differentiability of $g_V$, the equation $\psi(z) = 0$ has solution(s) in the range $z \in [0,1]$, with $x_0/y_0$ the maximum solution.

By Lemma~\ref{lem:phi}, $\psi(z) \leq \psi(1) = \gamma \phi(1)-1>0$ since $\phi(1) = g_V(1,1)$ and $\gamma>1/g_{V}(1,1)$. Furthermore, $\phi$, and hence $\psi$, is non-decreasing on $[0,1]$.

We have $\phi(0) = g_{V,1}(0,1)+g_{V,2}(0,1)$. Note also by Euler, $g_{V,2}(0,1) = g_V(0,1) \geq 1$. We need $\phi(0)<1/\gamma$. 
The slope of the gradient line to the level curve at $(x^\star,y^\star)=(0,1)$, given by $-\frac{g_{V,1}(0,1)}{g_{V,2}(0,1)}$, is either positive, zero, or negative but greater than $1-1/\gamma$. 

If the slope is zero or negative, then by convexity, and the fact that $\sup(G)=(1,1)$, $g_{V,2}(0,1) = g_{V}(0,1) = 1$. Furthermore, in the case the slope is negative 
\begin{align*}
    -g_{V,1}(0,1)/g_{V,2}(0,1) &> 1-1/\gamma& & \Rightarrow &g_{V,1}(0,1) + g_{V,2}(0,1) < g_{V,2}(0,1)/\gamma = 1/\gamma,
\end{align*}
 as desired. In the case the slope is zero, $g_{V,1}(0,1) = 0$ and $g_{V,2}(0,1) = g_{V}(0,1) = 1$, so $\phi(0)=1<1/\gamma$ also.

If the slope is positive, then $g_{V,2}(0,1)>1$ and $g_{V,1}(0,1)<0$.
Consider $\kappa\in(0,1)$ such that $g_{V}(\kappa,1)=1$, and $\tau=1/g_{V}(0,1)\in(0,1)$ such that $g_{V}(0,\tau)=1.$
By convexity, the slope is greater than $(1-\tau)/\kappa$, that is, $-g_{V,1}(0,1)/g_{V,2}(0,1)>(1-\tau)/\kappa$, and thus the following holds
\begin{align*}
    -g_{V,1}(0,1)/g_{V,2}(0,1) &>1/\kappa-1/(\kappa g_{V,2}(0,1))& &\Rightarrow &-g_{V,1}(0,1) > (g_{V,2}(0,1)-1)/\kappa >g_{V,2}(0,1)-1,
\end{align*}
since $\tau=1/g_{V}(0,1)$ and $g_{V}(0,1)=g_{V,2}(0,1)$.
Rearranging the right hand side gives $\phi(0)=g_{V,1}(0,1)+g_{V,2}(0,1)<1<1/\gamma,$ as desired.
Consequently, by the Intermediate Value Theorem, a solution to $\psi(z)=0$ exists for $y>x$, and if there are multiple solutions, then since $\psi$ is non-decreasing we must have $\psi(z)=0$ for all $z\in I_z=(z_\star,z^\star]$ with $\psi(z)<0$ for $z<z_\star$ and $\psi(z)>0$ for $z>z^\star$ with $x_0/y_0=z^\star.$

Now consider the piecewise differentiable case. By the arguments above, $\psi$ is non-decreasing, with $\psi(0)<0$, $\psi(1)>0$.
As in the continuously differentiable case, either $\psi(z)=0$ on an interval $I_z=(z_\star,z^\star],$ with $\psi(z)<0$ for $z<z_\star$ and $\psi(z)>0$ for $z>z^\star$, in which case we take $x_0/y_0=z^\star$; or the solution occurs at a point of non-differentiability.
In the latter case, $z^\star=x_0/y_0$ is a jump point of $\psi$, with $\psi(z)<0$ for $z<z_\star$ and $\psi(z)>0$ for $z>z^\star$.
Consequently, zero belongs to the subdifferential of $\psi$ at $z^\star$.
While these are infinitely many tangent lines to $g_{\bV}$ at $(x_0,y_0)$, only one of these passes through the point $(\gamma,\gamma).$


Finally, the equation to which $(x_0,y_0)$ is a solution is simply the equation of a tangent line to $g_{\bV}$, which passes through $(\gamma,\gamma)$. This can be seen for example by rearranging the first part of equation~\eqref{eq:tang}.
\end{proof}

\subsubsection{Proof for Proposition~\ref{prop:tangentline}}
\label{sec:proofproptangent}

The proof of Proposition~\ref{prop:tangentline} is facilitated by one further Lemma showing convexity of the objective function $h$.

\begin{lemma}
\label{lem:hconvex}
    For convex and twice (piecewise) differentiable $g_{\bV}$, the objective function $h(s)=s+g_{\bV}(x-\gamma s, y-\gamma s)$ is convex.
\end{lemma}
\begin{proof}
For $p \in [0,1]$,
    \begin{align*}
        h(ps_1+(1-p)s_2) &= ps_1+(1-p)s_2 +g_{\bV}\rbr{x-\gamma[ps_1+(1-p)s_2], y-\gamma[ps_1+(1-p)s_2]}\\
        &=ps_1+(1-p)s_2 +g_{\bV}\rbr{p(x-\gamma s_1,y-\gamma s_1) + (1-p)(x-\gamma s_1,y-\gamma s_1)}\\
        & \leq p \left\{s_1 +g_{\bV}(x-\gamma s_1,y-\gamma s_1)\right\} + (1-p) \left\{s_2 +g_{\bV}(x-\gamma s_2,y-\gamma s_2)\right\}\\
        &=ph(s_1)+(1-p)h(s_2),
    \end{align*}
    by convexity and homogeneity of $g_{\bV}$.
\end{proof}

\label{Appendix:tangentplane}
\begin{proof}
The fact that $g_{\bX}=g_{\bV}$ for $\gamma \leq 1/g_{\bV}(1,1)$ follows from Lemma~\ref{lem:mins0} since the minimizer is $\hat{s}=0$.
For $1/g_{\bV}(1,1)<\gamma\leq 1$, we derive an explicit form of the minimizer $\hat{s}\in[0,\min(x,y)/\gamma]$ and the corresponding analytic form of $g_{\bX}$ in~\eqref{eq:gauge_HW}. 
Consider the derivative of the objective function $h$:
    \begin{align*}
        h'(s)
        &=1-\gamma\phi\rbr{\frac{x-\gamma s}{y-\gamma s}}.
    \end{align*}
    By equation~\eqref{eq:tanggamma},
    \[
    \phi(x_0/y_0) = g_{V,1}(x_0,y_0)+g_{V,2}(x_0,y_0) = 1/\gamma.
    \]
    We first consider the condition for $\hat{s}=0$.
    By Lemma~\ref{lem:hconvex}, $h(s)$ is convex, and
    for a convex function on the compact interval $[0,\min(x,y)/\gamma)]$, the minimum occurs at the lower bound if and only if the derivative at $0$ is non-negative, i.e., $h'(0)\ge 0$, which is equivalent to $\phi(x/y)\leq 1/\gamma = \phi(x_0/y_0)$, i.e., $\phi(x/y)\le \phi(x_0/y_0).$ 
    Since $\phi(z)$ is non-decreasing for $z=x/y\in[0,1]$, the inequality in function values implies $x/y \le x_0/y_0.$
    This can be expressed in angular terms as 
    \begin{equation*}
        \frac{x}{x+y}\le \frac{x_0}{x_0+y_0},\quad y\ge x.
    \end{equation*}
    By symmetry, for $y<x$, the corresponding condition becomes $\frac{y}{x+y}\ge \frac{y_0}{x_0+y_0}.$    
    Combining both cases, the condition is
    \begin{align*}
        \hat{s}=
        0,\quad&\text{if}\,\,\,\, \frac{x}{x+y}\le \frac{x_0}{x_0+y_0}\,\,\,\,\text{for}\,\,\,\,y\ge x,\,\,\,\, \text{and}\,\,\frac{y}{x+y}\ge \frac{y_0}{x_0+y_0}\,\,\,\,\text{for}\,\,\,\,y<x.
    \end{align*}
    We now consider the general case where $h'(\hat{s})=0$, which gives
    \begin{equation*}
        \phi\rbr{\frac{x-\gamma \hat{s}}{y-\gamma \hat{s}}}=1/\gamma \qquad \Rightarrow \qquad\phi\rbr{\frac{x-\gamma \hat{s}}{y-\gamma \hat{s}}} = \phi\rbr{\frac{x_0}{y_0}}.
    \end{equation*}
If $\phi$ is strictly increasing over [0,1], then
      \begin{align}
    \frac{x_0}{y_0}=\frac{x-\gamma \hat{s}}{y-\gamma \hat{s}} \qquad \Rightarrow \qquad \hat{s}=\frac{y_0 x-x_0 y}{y_0 \gamma -x_0 \gamma}. \label{eq:s-hat}
    \end{align}
    Otherwise, if $\phi$ is non-decreasing such that $\phi(z) = 1/\gamma$ for all $z\in(z_\star,z^\star=x_0/y_0]$, then any $s$ such that $\frac{x-\gamma s}{y-\gamma s} \in (z_\star,z^\star=x_0/y_0]$ is a solution, so we can take $\hat{s}$ corresponding to the upper bound, i.e., with the same form as in equation~\eqref{eq:s-hat}.

    We have that, for $y \geq x$ and $x/(x+y) \geq x_0/(x_0+y_0)$, or by symmetry for $x>y$ and $y/(x+y)\geq y_0/(x_0+y_0)$,
    \begin{align*}
        g_{\bX}(x,y) &= \hat{s} + g_{\bV}(x-\gamma\hat{s},y-\gamma\hat{s})\\
        &= \frac{y_0 x-x_0 y}{y_0 \gamma -x_0 \gamma} + g_{\bV}\left(x-\frac{y_0 x-x_0 y}{y_0 -x_0},y-\frac{y_0 x-x_0 y}{y_0 -x_0}\right)\\
        &= \frac{y_0 x-x_0 y}{y_0 \gamma -x_0 \gamma} + g_{\bV}\left(\frac{x_0 y-x_0 x}{y_0 -x_0},\frac{y_0 y-y_0 x}{y_0 -x_0}\right)\\
        &= \frac{y_0 x-x_0 y}{y_0 \gamma -x_0 \gamma} + \frac{y-x}{x_0-y_0}g_{\bV}(x_0,y_0) = \frac{y-\rbr{\frac{\gamma-y_0}{\gamma-x_0}}x}{\gamma\rbr{1-\rbr{\frac{\gamma-y_0}{\gamma-x_0}}}},
    \end{align*}
using the homogeneity of $g_{\bV}$ and the fact that $g_{\bV}(x_0,y_0)=1$.
The resulting form of $g_{\bX}$ corresponds precisely to the gauge function derived from the tangent line to the unit level curve $g_{\bV}(x,y)=1$ at the point $(x_0,y_0)$, where this tangent line passes through $(\gamma,\gamma).$
    To verify the identity, for $y\ge x$, recall the equation of the tangent line at $(x_0,y_0)$
    \begin{equation*}
        1=xg_{V,1}(x_0,y_0)+yg_{V,2}(x_0,y_0) \qquad \Rightarrow  \qquad      y=-\frac{g_{V,1}(x_0,y_0)}{g_{V,2}(x_0,y_0)}x+\frac{1}{g_{V,2}(x_0,y_0)}.
    \end{equation*}
    Since this tangent line passes through the point $(\gamma,\gamma)$, we have $1=\gamma g_{V,1}(x_0,y_0)+\gamma g_{V,2}(x_0,y_0).$
    Using this condition, we can rewrite the tangent line equation as
    \begin{equation}
    \label{eq:tangentline_gv}
        y-\gamma=-\frac{g_{V,1}(x_0,y_0)}{g_{V,2}(x_0,y_0)}(x-\gamma).
    \end{equation}
    Since $(x_0,y_0)$ lies on the line, we can substitute this into equation~\eqref{eq:tangentline_gv} to give
    \begin{align*}
    -\frac{g_{V,1}(x_0,y_0)}{g_{V,2}(x_0,y_0)} = \frac{\gamma-y_0}{\gamma-x_0},
    \end{align*}
    and therefore
    \begin{align*}
        y-\gamma=\frac{\gamma-y_0}{\gamma-x_0}(x-\gamma).
    \end{align*}
    The gauge function associated with this tangent line, which passes through $(\gamma,\gamma)$, is thus
    \begin{equation*}
    \frac{y-\rbr{\frac{\gamma-y_0}{\gamma-x_0}}x}{\gamma\rbr{1-\rbr{\frac{\gamma-y_0}{\gamma-x_0}}}}=1.
    \end{equation*}
    In summary, for $\frac{1}{g_{\bV}(1,1)}\le \gamma \le 1$, the explicit form of the minimizer is given by
    \begin{equation*}
    \hat{s}=
    \begin{cases}
        0,\quad&\text{if}\,\,\,\, \frac{x}{x+y}\le \frac{x_0}{x_0+y_0}\,\,\,\,\text{for}\,\,\,\,y\ge x,\,\,\,\, \text{and}\,\,\,\,\frac{y}{x+y}\ge \frac{y_0}{x_0+y_0}\,\,\,\,\text{for}\,\,\,\,y<x\\
        \frac{x_0 y-y_0 x}{\gamma(x_0-y_0)},\quad&\text{if}\,\,\,\, \frac{x}{x+y}> \frac{x_0}{x_0+y_0}\,\,\text{for}\,\,\,\,y \ge x,\,\,\text{and}\,\,\,\,\frac{y}{x+y}< \frac{y_0}{x_0+y_0}\,\,\,\,\text{for}\,\,\,\,y<x.
    \end{cases}
    \end{equation*}
    and the analytic form of the gauge function $g_{\bX}$ in~\eqref{eq:gauge_HW} is
    \begin{align*}
        g_{\bX}(x,y)=
        \begin{cases}
            g_{\bV}(x,y),\quad&\text{if}\,\, \frac{x}{x+y}\le \frac{x_0}{x_0+y_0}\,\,\text{for}\,\,y\ge x,\,\, \text{and}\,\,\frac{y}{x+y}\ge \frac{y_0}{x_0+y_0}\,\,\text{for}\,\,y<x\\
            \frac{y-\rbr{\frac{\gamma-y_0}{\gamma-x_0}}x}{\gamma\rbr{1-\rbr{\frac{\gamma-y_0}{\gamma-x_0}}}},\quad&\text{if}\,\, \frac{x}{x+y}> \frac{x_0}{x_0+y_0}\,\,\text{for}\,\,y \ge x,\,\,\text{and}\,\,\frac{y}{x+y}< \frac{y_0}{x_0+y_0}\,\,\text{for}\,\,y<x.
        \end{cases}
    \end{align*}
\end{proof}

\newpage
\begin{center}
\Large\bf SUPPLEMENTARY MATERIAL
\end{center}

\section{Finding coordinatewise supremum analytically for additively mixed gauge functions}

Consider additively mixed gauge functions ${g}_{\bX}^\star(x,y)=pg_{\bX}^{[1]}(x,y)+(1-p)g_{\bX}^{[2]}(x,y)$ for $(x,y)\in[0,\infty)^2.$
We aim to determine the coordinatewise supremum analytically for specified classes of additive mixtures.
We assume that both $g_{\bX}^{[1]}$ and $g_{\bX}^{[2]}$ are continuous and symmetric; it follows that $c_1^\star=c_2^\star=c^\star.$
We consider examples where both $g_{\bX}^{[1]}$ and $g_{\bX}^{[2]}$ are convex, from which it follows that the objective function $g_{\bX}^\star(z,1):=\tilde{k}^\star(z)$ for $z\in[0,1]$ to be minimized is convex. 
We denote the minimizer of $\tilde{k}^\star$ over the interval $[0,1]$ by $\kappa$.

\subsection{Additively mixing Gaussian and logistic gauge functions}

Consider the objective function to be minimized, derived from the additive mixture of Gaussian and logistic gauge functions
\begin{equation}
\label{eq:Amix_Gausslog}
    \tilde{k}^\star(z)
    =
    \frac{p}{1-\rho^2}\rbr{z+1-2\rho\sqrt{z}}+(1-p)\rbr{\frac{1}{\gamma}(z+1)+\rbr{1-\frac{2}{\gamma}}z},\quad z\in[0,1],
\end{equation}
where $p\in(0,1)$, $\gamma\in(0,1)$, and $\rho\in[0,1).$
We analyze the derivative of the function \eqref{eq:Amix_Gausslog} with respect to $z$ to find the minimizer $\kappa$
\begin{align*}
    \frac{\diff}{\diff z}\tilde{k}^\star(z):=\tilde{k}^{\star'}(z)
    &=
    \frac{p}{1-\rho^2}\rbr{1-\frac{\rho}{\sqrt{z}}} + (1-p)\rbr{1-\frac{1}{\gamma}}.
\end{align*}

We first examine a specific case where $\rho=0$.
The derivative simplifies to $\tilde{k}^{\star'}(z)=p+(1-p)(1-1/\gamma)$, which is a constant. 
The sign of this derivative determines the minimizer $\kappa$
\begin{align*}
    \kappa=
    \begin{cases}
        0, \quad \text{if}\quad \gamma \ge 1-p\\
        1, \quad \text{if}\quad \gamma < 1-p.
    \end{cases}
\end{align*}

We now find the condition for the minimizer $\kappa$ to lie in $(0,1)$.
By setting $\tilde{k}^{\star '}(z)=0$ and solving for $\sqrt{z}$, we have
\begin{equation}
\label{eq:sqrtz}
\sqrt{z}=\frac{\rho}{1-(1-\rho^2)\rbr{\frac{1-p}{p}}\rbr{\frac{1-\gamma}{\gamma}}}=\frac{\rho}{1-K},    
\end{equation}
where $K=(1-\rho^2)\rbr{\frac{1-p}{p}}\rbr{\frac{1-\gamma}{\gamma}}>0$.
For $\kappa$ to be in $(0,1)$, $\sqrt{z}$ in~\eqref{eq:sqrtz} must be in $(0,1)$ as well.
The first conditions for $\sqrt{z}>0$ are $\rho\in(0,1)$ and $K<1$.
The second condition for $\sqrt{z}<1$ is $K<1-\rho$.
The condition $K<1-\rho$ is stricter than the condition $K<1$ since $(1-\rho)<1$.
Thus, the final conditions we need are $K<1-\rho$ and $\rho\in(0,1).$
The condition $K<1-\rho$ in terms of the original parameters is
\begin{align*}
    K<1-\rho \iff (1+\rho)\rbr{\frac{1-p}{p}}\rbr{\frac{1-\gamma}{\gamma}} < 1.
\end{align*}
Consequently, the two conditions for $\kappa\in(0,1)$ are $\rho\in(0,1)$ and $(1+\rho)\rbr{\frac{1-p}{p}}\rbr{\frac{1-\gamma}{\gamma}} < 1$.

Lastly, we find the condition for $\kappa = 1$, which holds if $\sqrt{z}\ge1$, where $\sqrt{z}$ as in~\eqref{eq:sqrtz}.
The condition $\sqrt{z}=\rho/(1-K)\ge 1$ implies $K\ge 1-\rho$. 
For the square root of $z$ to be positive, the denominator must also be positive, that is, $K<1.$
Thus, for the minimizer $\kappa$ to be $1$, we need to satisfy the condition $1-\rho\le K<1.$
In terms of original parameters, the condition is written as
\begin{align*}
    &(1-\rho)\le (1-\rho^2)\rbr{\frac{1-p}{p}}\rbr{\frac{1-\gamma}{\gamma}}<1\\
    &1\le (1+\rho)\rbr{\frac{1-p}{p}}\rbr{\frac{1-\gamma}{\gamma}} \,\,\text{and}\,\,(1-\rho^2)\rbr{\frac{1-p}{p}}\rbr{\frac{1-\gamma}{\gamma}}<1.
\end{align*}

All together, the minimizer is
\begin{align*}
    \kappa=
    \begin{cases}
        0,\quad\text{if}\quad \rho=0\,\,\text{and}\,\,\gamma\ge 1-p\\
        \rho^2/(1-K)^2,\quad\text{if}\quad \rho\in(0,1)\,\,\text{and}\,\,(1+\rho)\rbr{\frac{1-p}{p}}\rbr{\frac{1-\gamma}{\gamma}}<1\\
        1,\quad\text{if}\quad \rho=0\,\,\text{and}\,\,\gamma<1-p\,\,\text{or}\,\,\rho\in(0,1)\,\,\text{and}\,\,(1-\rho)\le(1-\rho^2)\rbr{\frac{1-p}{p}}\rbr{\frac{1-\gamma}{\gamma}}<1.
    \end{cases}
\end{align*}

\subsection{Additively mixing Inverted-logistic and logistic gauge functions}

Consider the objective function to be minimized, derived from the additive mixture of Inverted-logistic and logistic gauge functions
\begin{equation}
    \label{eq:Amix_InvLogLog}
    \tilde{k}^\star(z)=p\rbr{z^{1/\theta}+1}^\theta + (1-p)\rbr{\frac{1}{\gamma}(z+1)+\rbr{1-\frac{2}{\gamma}}z},\quad z\in[0,1],
\end{equation}
where $\theta\in(0,1]$, $p\in(0,1)$, and $\gamma\in(0,1).$
Similarly to before, we analyze the derivative of $\tilde{k}^\star(z)$ in~\eqref{eq:Amix_InvLogLog} with respect to $z$ to find the minimizer $\kappa$
\begin{align*}
    \frac{\diff}{\diff z}\Tilde{k}^\star(z)
    :=\tilde{k}^{\star '}(z)=
    p\left(z^{1/\theta}+1\right)^{\theta-1}z^{1/\theta-1}+(1-p)\rbr{\frac{1-\gamma}{\gamma}}.
\end{align*}

We first consider a specific case where $\theta=1$.
The derivative simplifies to $\tilde{k}^{\star '}(z)=p+(1-p)\rbr{\frac{1-\gamma}{\gamma}}$, which is a constant.
Depending on the sign of the derivative, the minimizer is determined as follows:
\begin{align*}
    \kappa=
    \begin{cases}
        0,\quad\text{if}\quad \gamma \ge 1-p\\
        1,\quad\text{if}\quad \gamma < 1-p.
    \end{cases}
\end{align*}

We now find the condition for the minimizer $\kappa$ to lie in $(0,1).$
By solving for $\tilde{k}^{\star '}(z)=0$, the solution, denoted $z^*$, is given by
\begin{equation*}
    z^*=\left[\left\{\left(\frac{1-p}{p}\right)\left(\frac{1}{\gamma}-1\right)\right\}^{\frac{1}{\theta-1}}-1\right]^{-\theta}
    =\cbr{C^{\frac{1}{\theta-1}}-1}^{-\theta},
\end{equation*}
with $C=\left(\frac{1-p}{p}\right)\left(\frac{1}{\gamma}-1\right)>0$.

The condition $\kappa\in(0,1)$ means we need to consider two inequalities: $z^* >0$ and $z^* <1.$
To satisfy the condition for $z^* > 0$, $C^{\frac{1}{\theta-1}}-1$ must be positive since the exponent is negative for $\theta\in(0,1).$
Solving for $C$, the condition becomes
\begin{align*}
    &C^{\frac{1}{\theta-1}}>1\\
    &C<1 \iff \left(\frac{1-p}{p}\right)\left(\frac{1}{\gamma}\right)<1.
\end{align*}
Thus, the condition for $z^* > 0$ is $\left(\frac{1-p}{p}\right)\left(\frac{1}{\gamma}\right)<1.$
To satisfy condition for $z^* <1$, we now solve 
\begin{align*}
    \cbr{C^{\frac{1}{\theta-1}}-1}^{-\theta}&<1\\
    C^\frac{1}{\theta-1} &> 2\\
    C&<2^{\theta-1}\\
    \left(\frac{1-p}{p}\right)\left(\frac{1}{\gamma}\right)&<2^{\theta-1}.
\end{align*}
For $\kappa$ to be in $(0,1)$, both conditions must hold together: $C<1$ and $C<2^{\theta-1}$.
This implies that $C<\min(1,2^{\theta-1}).$
Since $2^{\theta-1}\le 1$ for $\theta\in(0,1]$, it reduces to the single condition $C<2^{\theta-1}.$

Lastly, for the minimizer $\kappa$ to be $1$, we need to satisfy $z^*\ge 1$.
The condition in terms of original parameters is written as
\begin{align*}
    \rbr{C^{\frac{1}{\theta-1}}-1}^{-\theta} &\ge 1\\
    C^{\frac{1}{\theta-1}}&\le 2\\
    C&\ge 2^{\frac{1}{\theta-1}}\\
    \left(\frac{1-p}{p}\right)\left(\frac{1}{\gamma}\right)&\ge 2^{\theta-1}.
\end{align*}

All together, the minimizer is
\begin{align*}
    \kappa=
    \begin{cases}
        0, \quad\,\,\,\,\, \text{if}\quad \rho=0\,\,\text{and}\,\,\gamma \ge 1-p\\
        z^*,\quad \text{if}\quad \rho\in(0,1)\,\,\text{and}\,\,(1+\rho)\rbr{\frac{1-p}{p}}\rbr{\frac{1-\gamma}{\gamma}}<1\\
        1,\quad\,\,\,\,\, \text{if}\quad \rho=0\,\,\text{and}\,\,\gamma=1-p\,\,\text{or}\,\, \rho\in(0,1)\,\,\text{and}\,\,(1-\rho^2)\rbr{\frac{1-p}{p}}\rbr{\frac{1-\gamma}{\gamma}}\ge 1.
    \end{cases}
\end{align*}

\subsection{Additively mixing rectangular and logistic gauge functions}
Consider the objective function derived from the additive mixture of rectangular and logistic gauge functions
\begin{equation*}
    \tilde{k}^\star(z)=p\max\left(\frac{1}{\theta}(z-1),\frac{1}{\theta}(1-z),\frac{1}{2-\theta}(z+1)\right)
    +(1-p)\rbr{\frac{1}{\gamma}(z+1)+\rbr{1-\frac
    {2}{\gamma}}z},\quad z\in[0,1],
\end{equation*}
where $\theta\in(0,1]$, $\gamma\in(0,1)$, and $p\in(0,1).$

Both $g_{\bX}^{[1]}$ and $g_{\bX}^{[2]}$ are continuous and symmetric but not differentiable everywhere.
The logistic gauge function is strictly monotonic over $z\in[0,1]$.
We first consider a specific case where $\theta=1.$
The derivative simplifies to $\tilde{k}^{\star '}(z)=p+(1-p)\rbr{1-1/\gamma}$, which is a constant.
The sign of the derivative determines the minimizer as follows:
\begin{align*}
    \kappa=
    \begin{cases}
        0, \quad \text{if}\quad \gamma \ge 1-p\\
        1, \quad \text{if} \quad \gamma < 1-p.
    \end{cases}
\end{align*}

Note that the rectangular gauge function exhibits a sharp turning point at $(1-\theta)$, denoted by $z^*=1-\theta$, which arises from the point of intersection $\frac{1}{\theta}(1-z)=\frac{1}{2-\theta}(z+1).$
For $z^*=(1-\theta)\in(0,1)$ to be a turning point, we need 
\[
\frac{\diff}{\diff z}\tilde{k}^{\star '}(z)=\frac{\diff}{\diff z}\cbr{\frac{p}{\theta}(1-z)+(1-p)\cbr{\frac{1}{\gamma}(z+1)+\rbr{1-\frac{2}{\gamma}}z}}=-\frac{p}{\theta}+(1-p)(1-2/\gamma)<0
\] for $z\in(0,1-\theta)$, which is always true, and 
\[\frac{\diff}{\diff z}\tilde{k}^{\star '}(z)=\frac{\diff}{\diff z}\cbr{\frac{p}{2-\theta}(z+1)+(1-p)\cbr{\frac{1}{\gamma}(z+1)+\rbr{1-\frac{2}{\gamma}}z}}=\frac{p}{2-\theta}+(1-p)(1-1/\gamma)>0\] for $z\in(1-\theta,1).$
Thus, the sign of the quantity $\frac{p}{2-\theta}+(1-p)(1-1/\gamma)$ determines the location of the minimizer $\kappa$
\begin{align*}
    \kappa=
    \begin{cases}
        1-\theta, \quad \text{if}\quad \frac{p}{2-\theta}+(1-p)(1-1/\gamma) > 0 \\
        1, \quad \text{if} \quad \frac{p}{2-\theta}+(1-p)(1-1/\gamma) \le 0.
    \end{cases}
\end{align*}
All together, the minimizer is
\begin{align*}
    \kappa=
    \begin{cases}
        0,\qquad\quad\text{if}\quad \theta=1\,\,\text{and}\,\,\gamma\ge 1-p\\
        1-\theta,\,\,\,\quad\text{if}\quad \theta\in(0,1)\,\,\text{and}\,\,\frac{p}{2-\theta}+(1-p)(1-1/\gamma) > 0\\
        1,\qquad\quad\text{if}\quad \theta=1\,\,\text{and}\,\,\gamma<1-p\,\,\text{or}\,\,\theta\in(0,1)\,\,\text{and}\,\,\frac{p}{2-\theta}+(1-p)(1-1/\gamma) \le 0.
    \end{cases}
\end{align*}

\section{Deriving the analytic form of additive mixtures constructed from an additive stochastic representation}

\subsection{Determining the solution $\hat{s}$ analytically for the Gaussian gauge function $g_{V}$}
Consider the additive mixture with the Gaussian gauge function for $g_{\bV}$
\begin{equation*}
    g_{\bX}(x,y)=\min_{s\in[0,\min(x,y)/\gamma]}\cbr{s+g_{\bV}(x-\gamma s,y-\gamma s)}.
\end{equation*}
Define the objective function $h(s)$ to be minimized as
\begin{align*}
    h(s)
    &=
    s+\rbr{x-\gamma s y-\gamma s -2\rho\rbr{(x-\gamma s)(y-\gamma s)}^{1/2}}/(1-\rho^2)\\
    &=
    s+\rbr{x+y-2\gamma s -2\rho\sqrt{(x-\gamma s)(y-\gamma s)}}/(1-\rho^2)\\
    &=
    s+\rbr{x+y-2\gamma s -2\rho\sqrt{l(s)}}/(1-\rho^2),
\end{align*}
where $l(s)=(x-\gamma s)(y- \gamma s)$.
For differentiable $g_{\bV}$, we consider the first-order condition of $h(s)$ to determine the minimizer $\hat{s}$
\begin{align*}
    \frac{\diff}{\diff s}h(s)
    &=
    1-\frac{2\gamma}{1-\rho^2}-\frac{\rho}{1-\rho^2}\frac{1}{\sqrt{l(s)}}l'(s)\\
    &=
    1-\frac{2\gamma}{1-\rho^2}-\frac{\rho(-\gamma(x+y)+2\gamma^2 s)}{(1-\rho^2)\sqrt{(x-\gamma s)(y-\gamma s)}}=0
\end{align*}
After expanding and rearranging the terms, we have
\begin{align*}
    \rbr{1-\rho^2}^2\rbr{1-\frac{2\gamma}{1-\rho^2}}^2\rbr{xy-\gamma(x+y)s+\gamma^2 s^2}
    =
    \rho^2\rbr{\gamma^2(x+y)^2-4\gamma^3(x+y)s+4\gamma^4 s^2}.
\end{align*}
Letting $\hatc=\rbr{1-\rho^2}^2\rbr{1-\frac{2\gamma}{1-\rho^2}}^2=(1-\rho^2-2\gamma)^2$, the equation simplifies to
\begin{align*}
    &\hatc xy-\hatc\gamma(x+y)s+\hatc\gamma^2 s^2=\rho^2\gamma^2(x+y)^2-4\rho^2\gamma^3(x+y)s+4\rho^2\gamma^4s^2\\
    &\rbr{4\rho^2\gamma^4-\hatc\gamma^2}s^2-\rbr{4\rho^2\gamma^3(x+y)-\hatc\gamma(x+y)}s+\rho^2\gamma^2(x+y)^2-\hatc xy=0\\
    &\rbr{4\rho^2\gamma^2-\hatc}\gamma^2s^2-(4\rho^2\gamma^2-\hatc)\gamma(x+y)s+\rho^2\gamma^2(x+y)^2-\hatc xy=0\\
    &\rbr{\gamma^2K}s^2-\gamma K(x+y)s+\rho^2\gamma^2(x+y)^2-\hatc xy=0,
\end{align*}
with $K=4\rho^2\gamma^2-\hatc.$
The resulting equation is a quadratic equation for $s$.
Letting $A:=\gamma^2 K$, $B:=\gamma K(x+y)$, and $C:=\rho^2\gamma^2(x+y)^2-\hatc xy$, the solution lying in $(0,\min(x,y)/\gamma]$ is
\begin{equation}
\label{eq:shat_Ga}
    \hat{s}=\frac{-B+\sqrt{B^2-4AC}}{2A}.
\end{equation}
Substituting the minimizer $\hat{s}$ in~\eqref{eq:shat_Ga} into the gauge function $g_{\bX}(x,y)$ yields the analytic form.

By Proposition 3 in the main manuscript, the solution $\hat{s}$ is linear in $x$ and $y$.
For the solution $\hat{s}$ to be linear, the discriminant $B^2-4AC$ must be a perfect square of a linear function of $x$ and $y.$
We calculate the discriminant
\begin{align*}
    B^2-4AC
    &=\gamma^2 K^2(x+y)^2-4\rbr{\gamma^2 K}\rbr{\rho^2\gamma^2(x+y)^2-\hatc xy}\\
    &=\gamma^2 K\rbr{K(x+y)^2-4\rho^2\gamma^2(x+y)^2+4\hatc xy}\\
    &=\gamma^2 K\rbr{-\hatc(x+y)^2+4\hatc xy}\\
    &=\gamma^2 K \hatc\rbr{-(x+y)^2+4xy}\\
    &=-\gamma^2 K \hatc(x-y)^2.
\end{align*}

The square root of the discriminant is $\sqrt{B^2-4AC}=\gamma\sqrt{-K\hatc}|x-y|.$
For $\frac{1}{g_{\bV}(1,1)}=\frac{1+\rho}{2}<\gamma\le 1$, we will show that $-K\hatc$ is always nonnegative.
Since $\hatc\ge 0$, it is sufficient to show $K\le0$.
The condition $K\le 0$ is equivalent to
\begin{equation*}
    2\rho\gamma\le |1-\rho^2-2\gamma|.
\end{equation*}
Given $2\gamma > 1+\rho$, observe the sign of the term inside the absolute value
\begin{equation*}
    1-\rho^2-2\gamma < 1-\rho^2-(1+\rho)=-\rho^2-\rho \le 0,
\end{equation*}
for $\rho\in[0,1)$.
This implies $1-\rho^2-2\gamma$ is always negative under the condition.
The inequality we need to show is now
\begin{equation*}
    2\rho\gamma \le -1 +\rho^2 +2\gamma.
\end{equation*}
Rearranging the terms, we have a quadratic expression with respect to $\rho$
\begin{equation*}
    \rho^2-2\gamma \rho+(2\gamma-1)\ge 0.
\end{equation*}
The solutions to the quadratic equation $\rho^2-2\gamma\rho+(2\gamma-1)=0$ are $\rho_1=1$ and $\rho_2=2\gamma -1$.
For the quadratic $\rho^2-2\gamma\rho+(2\gamma-1)$ to be nonnegative, $\rho$ must be outside these solutions: $\rho\ge1$ or $\rho\le 2\gamma-1$.
Since $\rho\in[0,1)$, we need $\rho\le 2\gamma-1$, which is guaranteed by the condition $\gamma>(1+\rho)/2.$
Therefore, the inequality always hold under the condition $\gamma>(1+\rho)/2$ and thus $-K\hatc$ is always nonnegative.

Substituting the simplified discriminant back into the quadratic form solution for $\hat{s}$, we have
\begin{equation}
\label{eq:solution_s}
    \hat{s}=\frac{K(x+y)+\sqrt{-K\hatc}|x-y|}{2\gamma K}.
\end{equation}

\subsection{Finding an unique point $(x_0,y_0)$ from the tangent line equation for Gaussian gauge case}

Recall Proposition 3 in the main manuscript. 
For $\frac{1}{g_{\bV}(1,1)}< \gamma \le 1$, the gauge function $g_{\bX}$ gives a transition: $g_{\bX}$ corresponds to $g_{\bV}$ when the minimizer is $\hat{s}=0$ and corresponds to the gauge function associated with the tangent line to the unit level curve $g_{\bV}(x,y)=1$ at the point $(x_0,y_0)$, where the tangent line passes through $(\gamma,\gamma)$.
By Lemma 1 in the main manuscript, the point $(x_0,y_0)$ is uniquely determined by solving the equation for the tangent line to the unit level curve $g_{\bV}(x,y)=1$, where the tangent line passes through $(\gamma,\gamma)$.

Specifically, for $y\ge x$, the point $(x_0,y_0)$ is the solution to the equation $(y-\gamma)=\frac{\diff y}{\diff x}(x-\gamma)$, where $y$ is defined as a function of $x$ via the unit level curve $g_{\bV}(x,y)=1$.
For Gaussian gauge function $g_{\bV}$, the function $y$ is written as
\begin{equation*}
    y = 2\rho\sqrt{(1-\rho^2)(x-x^2)}-x(1-2\rho^2)+(1-\rho^2).
\end{equation*}
The derivative of the function $y$ with respect to $x$ is
\begin{equation*}
    \frac{\diff y}{\diff x} = \rho\rbr{(1-\rho^2)(x-x^2)}^{-1/2}\rbr{(1-\rho^2)(1-2x)}-(1-2\rho^2).
\end{equation*}

Letting $z=(1-\rho^2)(x-x^2)$, the equation of the tangent line that passes through the point $(\gamma,\gamma)$ is given by
\begin{equation*}
    y-\gamma = \frac{\diff y}{\diff x}(x-\gamma).
\end{equation*}
Substituting the derivative and $y$ into the tangent line equation, we rearrange terms
\begin{align*}
    2\rho\sqrt{z}-x(1-2\rho^2)+(1-\rho^2)-\gamma 
    &= \cbr{\rho z^{-1/2}\rbr{(1-\rho^2)(1-2x)}-(1-2\rho^2)}(x-\gamma) \\
    2\rho\sqrt{z}-\rho z^{-1/2}\rbr{(1-\rho^2)(1-2x)}(x-\gamma) &= (2\gamma-1)(1-\rho^2) \\
    2\rho z -z^{1/2}x(1-2\rho^2)+z^{1/2}(1-\rho^2)-z^{1/2}\gamma &= \rho(1-\rho^2)(1-2x)(x-\gamma) - z^{1/2}(1-2\rho^2)x + z^{1/2}(1-2\rho^2)\gamma \\
    2\rho(x-x^2)+(1-\rho^2)^{1/2}(x-x^2)^{1/2}(1-2\gamma) &= \rho(1-2x)(x-\gamma) \\
    2\rho^2 \gamma(1-2\gamma)x + \rho^2\gamma^2 &= (2\gamma-1)^{2}x - (2\gamma-1)^2x^2 - \rho^2(2\gamma-1)^2 x.
\end{align*}
Simplifying the following equation 
\begin{equation*}
    (2\gamma-1)^2 x^2 + (2\rho^2\gamma(1-2\gamma)+\rho^2(2\gamma-1)^2-(2\gamma-1)^2)x + \rho^2\gamma^2 = 0,
\end{equation*}
the resulting equation reduced to a quadratic form
\begin{equation}
\label{eq:quadratic_x0}
    (2\gamma-1)^2x^2+(2\gamma-1)(1-\rho^2-2\gamma)x+\rho^2\gamma^2=0.
\end{equation}
Letting $a:=(2\gamma-1)^2$, $b:=(2\gamma-1)(1-\rho^2-2\gamma)$, and $c:=\rho^2\gamma^2$, the solutions are
\begin{equation*}
    x_0 = \frac{-b \pm \sqrt{b^2-4ac}}{2a}.
\end{equation*}

Let $(x_0,y_0)$ be the point obtained from the quadratic equation in~\eqref{eq:quadratic_x0}.
Using the solution $(x_0,y_0)$, we can verify that the solution $\hat{s}$ in~\eqref{eq:solution_s} coincides with the form $\frac{x_0y-y_0x}{\gamma(x_0-y_0)}$ derived from Proposition 3.
Recall that $\hatc=(1-\rho^2-2\gamma)^2$ and $K=4\rho^2\gamma^2-\hatc.$
The sum of $x_0+y_0$ is $\frac{1-\rho^2-2\gamma}{1-2\gamma}$ and thus its squared is $(x_0+y_0)^2=\hatc/(1-2\gamma)^2.$
The product of $x_0y_0$ is $\frac{\rho^2\gamma^2}{(2\gamma-1)^2}$.
Using $x_0+y_0$ and $x_0y_0$, we can also obtain $(x_0-y_0)^2$ as follows
\begin{align*}
    (x_0-y_0)^2&=(x_0+y_0)^2-4x_0y_0\\
    &=
    \frac{(2\gamma-1)^2(1-\rho^2-2\gamma)^2}{(2\gamma-1)^4}-\frac{4\rho^2\gamma^2}{(2\gamma-1)^2}\\
    &=
    \frac{(1-\rho^2-2\gamma)^2}{(2\gamma-1)^2}-\frac{4\rho^2\gamma^2}{(2\gamma-1)^2}\\
    &=
    \frac{-4\rho^2\gamma^2+\hatc}{(2\gamma-1)^2}\\
    &=
    -\frac{K}{(2\gamma-1)^2}.
\end{align*}

For $y\ge x$, the solution $\hat{s}$ in~\eqref{eq:solution_s} is
\begin{equation*}
    \hat{s}
    =\frac{K(x+y)+\sqrt{-K\hatc}(y-x)}{2\gamma K}
    =\rbr{\frac{K-\sqrt{-K\hatc}}{2\gamma K}}x+\rbr{\frac{K+\sqrt{-K\hatc}}{2\gamma K}y}.
\end{equation*}
We aim to show that this solution $\hat{s}$ is identical to the form 
\[
\frac{-y_0 x +x_0 y}{\gamma(x_0-y_0)}=\rbr{\frac{-y_0}{\gamma(x_0-y_0)}x}+\rbr{\frac{x_0}{\gamma(x_0-y_0)}}y.
\]
For these two linear functions of $x$ and $y$ to be identical, we equate both respective coefficients
\begin{align*}
    \frac{K-\sqrt{-K\hatc}}{2\gamma K}=\frac{-y_0}{\gamma(x_0-y_0)}\\
    \frac{K+\sqrt{-K\hatc}}{2\gamma K}=\frac{x_0}{\gamma(x_0-y_0)}.
\end{align*}
Subtracting these two equations, these identities hold if and only if the following condition is satisfied
\begin{equation*}
    \frac{\sqrt{-K\hatc}}{K}=\frac{x_0+y_0}{x_0-y_0}.
\end{equation*}
Squaring both sides of the condition, we need to check if the following identity holds
\begin{equation*}
    -\frac{\hatc}{K}=\frac{(x_0+y_0)^2}{(x_0-y_0)^2}.
\end{equation*}
We now substitute $(x_0+y_0)^2$ and $(x_0-y_0)^2$ from the previous calculations into the right-hand side
\begin{align*}
    \frac{(x_0+y_0)^2}{(x_0-y_0)^2}
    &=\frac{\frac{\hatc}{(2\gamma-1)^2}}{-\frac{K}{(2\gamma-1)^2}}\\
    &=-\frac{\hatc}{K}.
\end{align*}
Therefore, the identity holds true, implying that the solution $\hat{s}$ in~\eqref{eq:solution_s} obtained from the derivative corresponds to the form $\frac{x_0 y-y_0 x}{\gamma(x_0-y_0)}$ derived from Proposition 3.

\subsection{Finding an unique point $(x_0,y_0)$ from the tangent line for Inverted-logistic gauge function}

Similarly to the Gaussian gauge function, for the Inverted-logistic gauge function $g_{\bV}$, we first solve the unit level curve $g_{\bV}(x,y)=1$ for $y$
\begin{equation*}
    y = \rbr{1-x^{1/\theta}}^\theta.
\end{equation*}
Its derivative with respect to $x$ is given by
\begin{equation*}
    \frac{\diff y}{\diff x} = -\rbr{1-x^{1/\theta}}^{\theta-1} x^{1/\theta-1}.
\end{equation*}
The tangent line equation that passes through the point $(\gamma,\gamma)$ is
\begin{align*}
    y - \gamma &= \frac{\diff y}{\diff x}(x-\gamma)\\
    \rbr{1-x^{1/\theta}}^\theta - \gamma &= -\rbr{1-x^{1/\theta}}^{\theta-1}x^{1/\theta-1}(x-\gamma) \\
    1/\gamma&=(1-x^{1/\theta})^{1-\theta}+x^{1/\theta-1}
\end{align*}
For briefer notation, we let $z:=x^{1/\theta}$ and consider the implicit equation:
\begin{equation}
\label{eq:implicit_eq}
    \frac{1}{\gamma}=(1-z)^{1-\theta}+z^{1-\theta},
\end{equation}
where $\gamma\in(2^{-\theta},1]$, $\theta\in(0,1]$, and $z\in[0,1]$.
For $\theta=1$, the constraint on $\gamma$ is $\gamma\in(1/2,1]$.
When $\theta=1$, the solution to \eqref{eq:implicit_eq} can only exist if $\gamma=1/2$.
Since $\gamma=1/2$ is outside the interval, there is no solution.
When $\theta\in(0,1),$ let $f(z)=(1-z)^{1-\theta}+z^{1-\theta}$.
The function $f(z)$ is concave with a maximum of $2^\theta$ and minimum of $1.$
Rearranging the constraint on $\gamma$ gives $1\le 1/\gamma < 2^{\theta}$, which is equivalent to finding the solution $z$ for which $1\le f(z)< 2^\theta$.
For any value of $\gamma\in(2^{-\theta},1]$, there are two distinct solutions for $z=x^{1/\theta}.$
If $\gamma=1$, it corresponds to the minimum value of $f(z)=1$, which occurs at the boundaries $z=0$ and $z=1$ so $x=0$ and $x=1.$
If $\gamma\in(2^{-\theta},1)$, since the function $f(z)$ is symmetric around $z=1/2$, let one solution for $z_1$ be $z_1\in(0,1/2)$.
The other solution is $z_2=1-z_1.$
Thus, the solutions are $x_1=z_1^\theta$ and $x_2=(1-z_1)^\theta.$
Since the solution $z_1$ lies in $(0,1/2)$, $1-z_1\in(1/2,1)$.
This means $z_1<1-z_1$ and it implies $z_1^{\theta}<(1-z_1)^{\theta}\implies x_1<x_2.$
Consequently, we can take the minimum solution $x_1$, which is the unique solution for $y\ge x.$

\section{Box-plots of estimated dependence measures from simulation study}

We create box-plots of the estimated slopes $\hat{\alpha}$ for the conditional extremes model \citep{heffernan2004conditional,nolde2022linking} over 1000 iterations across five scenarios, for seven additive mixture models, as shown in Figure~\ref{fig:alpha}.
For the asymptotic dependence cases, the box-plots indicate that the median values of the estimated slopes are 1 under the logistic and Dirichlet dependence structures, while there is slightly greater variability in the weakly dependent AD case, as expected.

\begin{figure}[ht!]
\centering
\includegraphics[width=4cm]{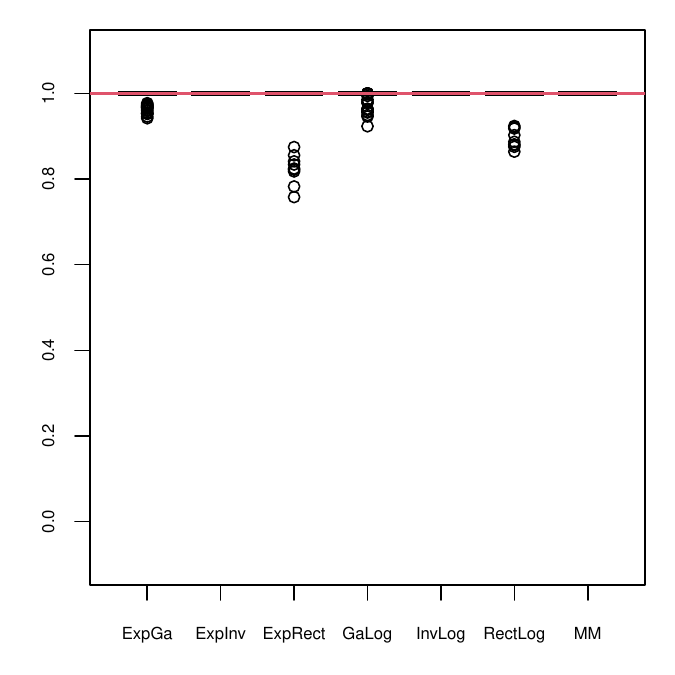}
\includegraphics[width=4cm]{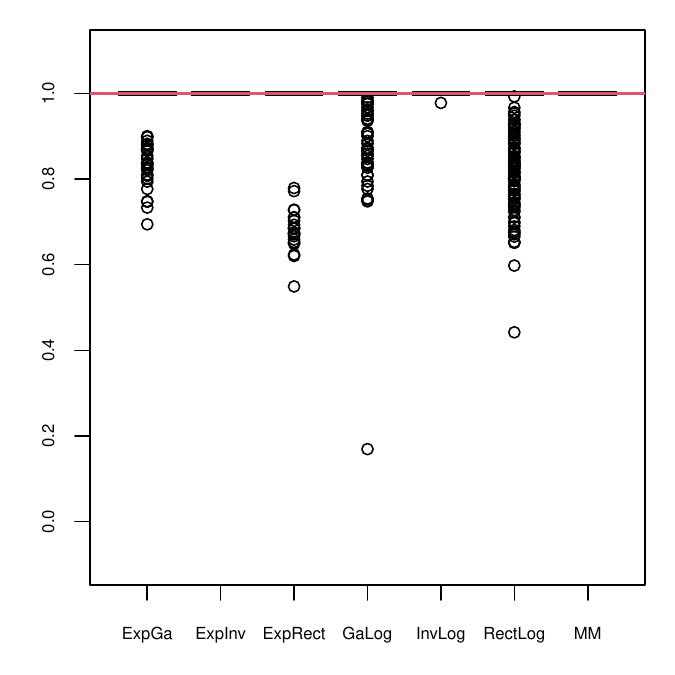}
\includegraphics[width=4cm]{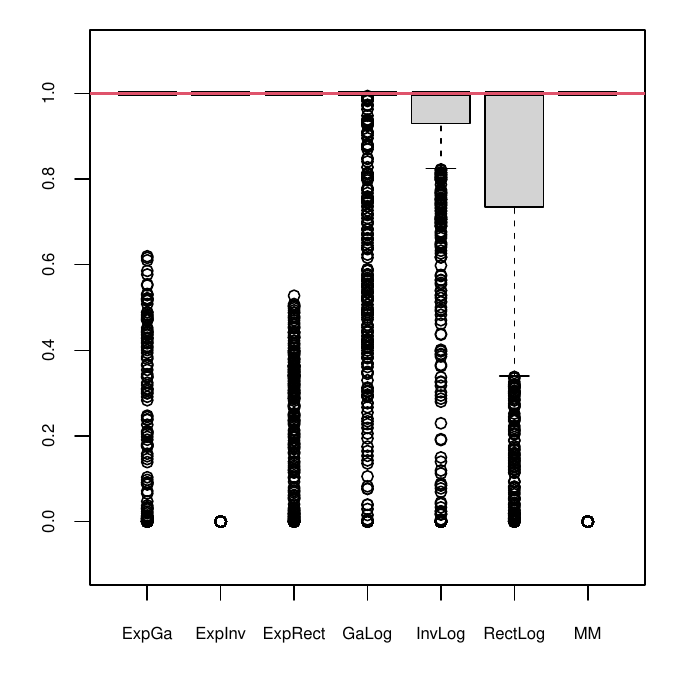}\\
\includegraphics[width=4cm]{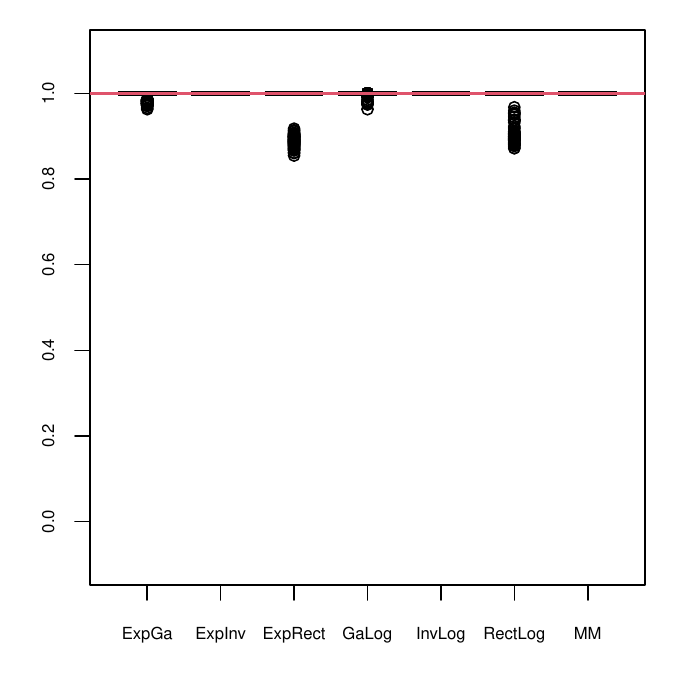}
\includegraphics[width=4cm]{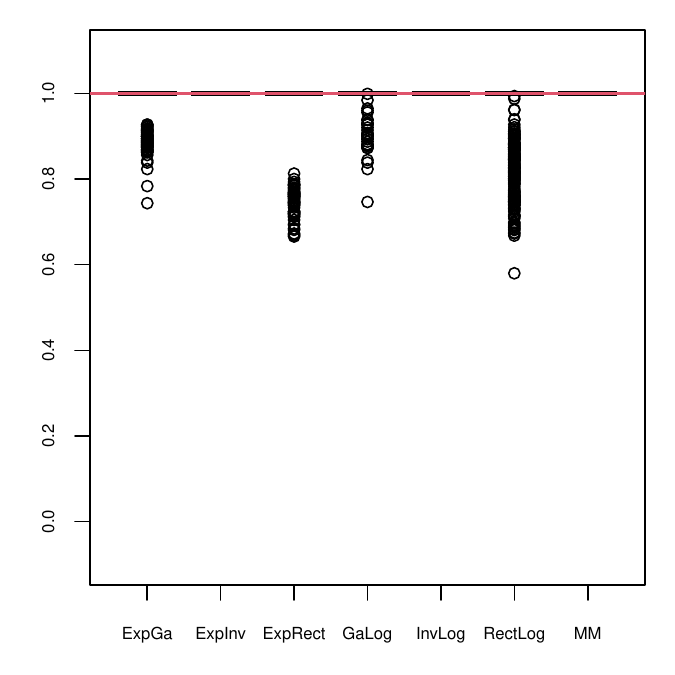}
\includegraphics[width=4cm]{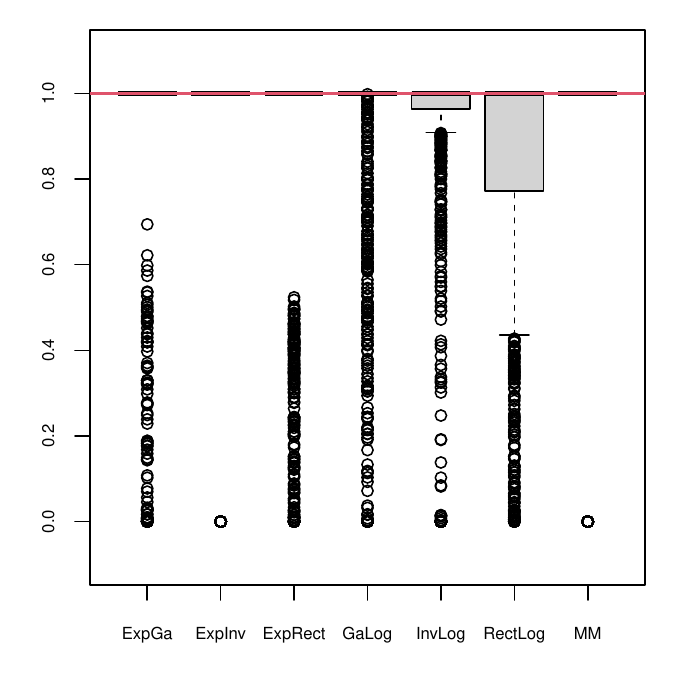}\\
\includegraphics[width=4cm]{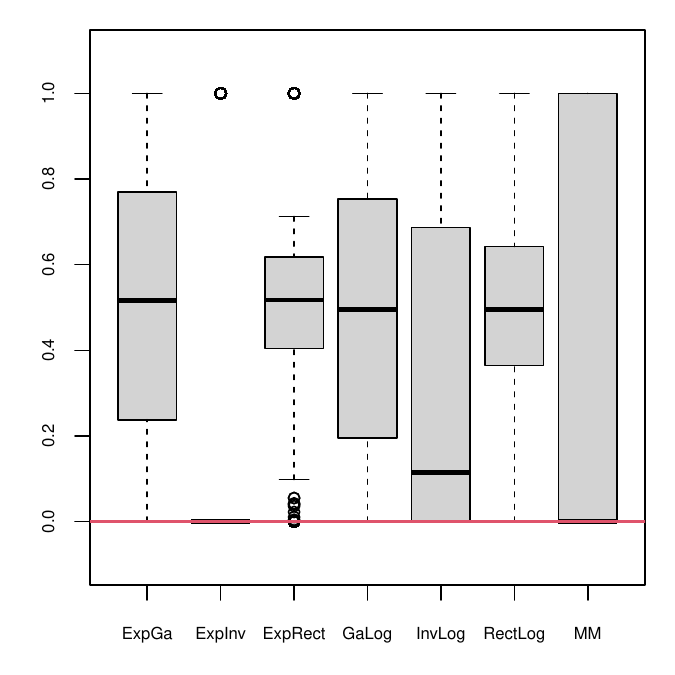}
\includegraphics[width=4cm]{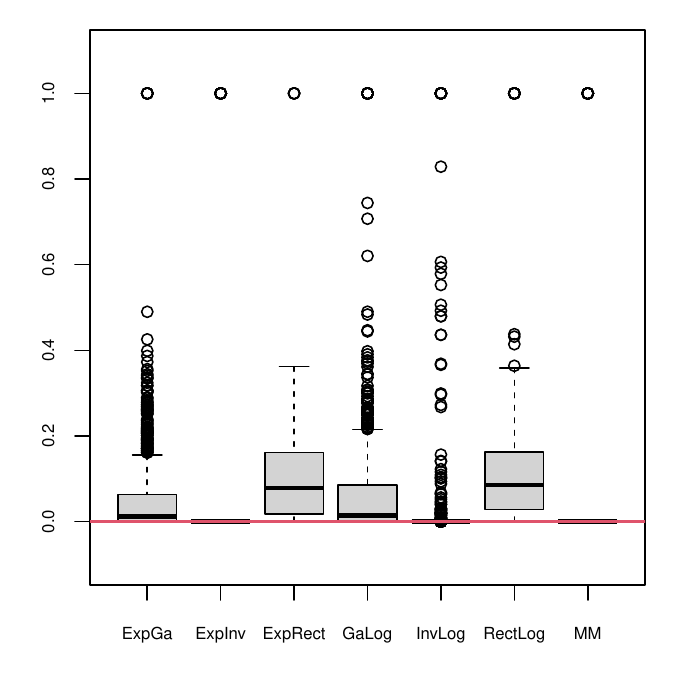}
\includegraphics[width=4cm]{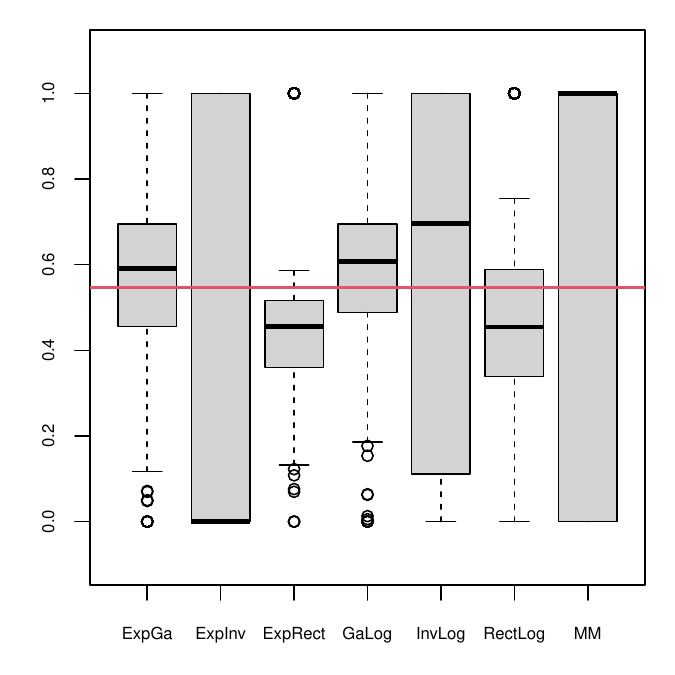}
\includegraphics[width=4cm]{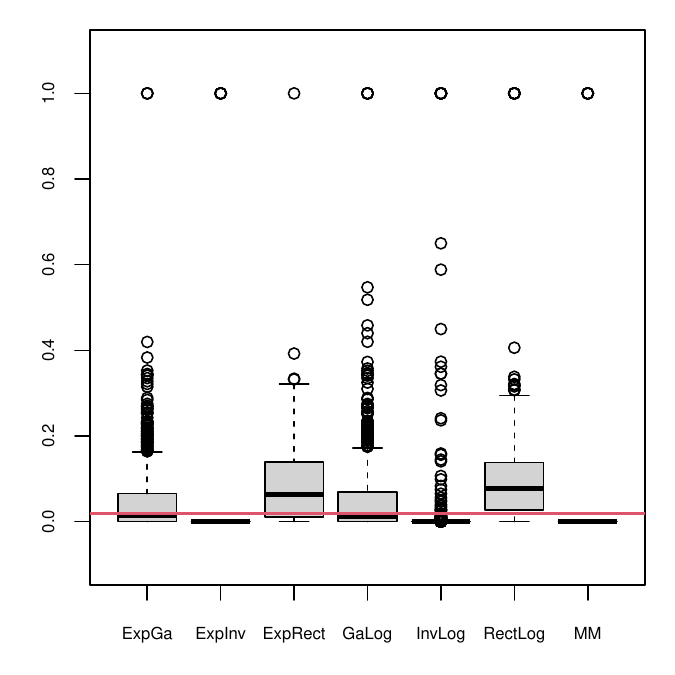}
\caption{\label{fig:alpha} Box-plots of the estimated slopes $\hat{\alpha}$ for the conditional extremes model over 1000 iterations.
From left to right, the top row represents the scenarios for strongly dependent AD (st.d.AD), moderately strongly dependent AD (mst.d.AD), weakly dependent AD (w.d.AD) under the logistic dependence structure.
The second row represents the same scenarios (st.d.AD, mst.d.AD, w.d.AD) under the Dirichlet dependence structure.
In the bottom row, the first two plots show strongly dependent AI (st.d.AI) and weakly dependent AI (w.d.AI) under the Inverted-logistic dependence structure, while the last two plots indicate st.d.AI and w.d.AI under the Gaussian dependence structure.
The box-plots are ordered as follows for each plot: ExpGa, ExpInv, ExpRect, GaLog, InvLog, RectLog, and MM.}
\end{figure}

The box-plots of the estimated residual tail dependence $\hat{\eta}$ are shown in Figure~\ref{fig:eta} over 1000 iterations across five scenarios, for seven additive mixture models.
For the asymptotic dependence cases, the box-plots indicate that the median values of the estimated slopes are 1, although there is relatively greater variability in the weakly dependent AD case, as expected.
Similarly, for the strongly dependent asymptotic independence cases, they show relatively greater variability compared to the weakly dependent AI cases.

\begin{figure}[ht!]
\centering
\includegraphics[width=4cm]{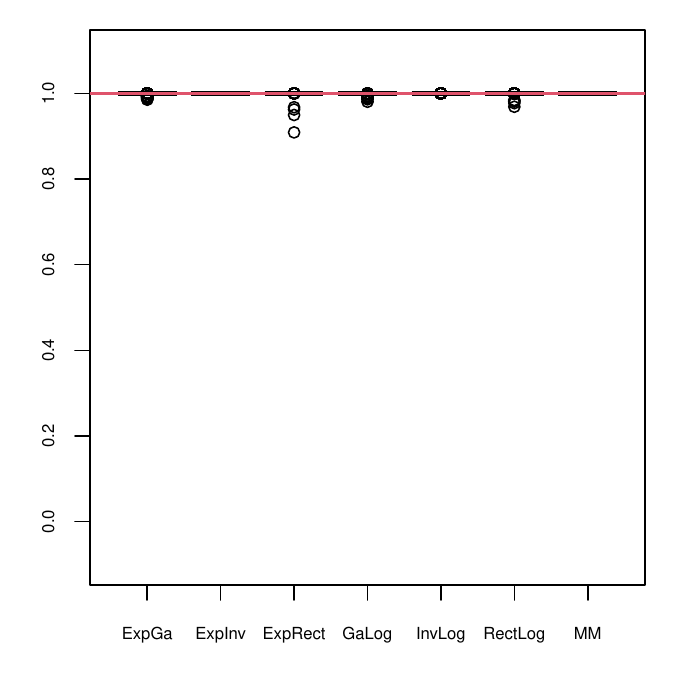}
\includegraphics[width=4cm]{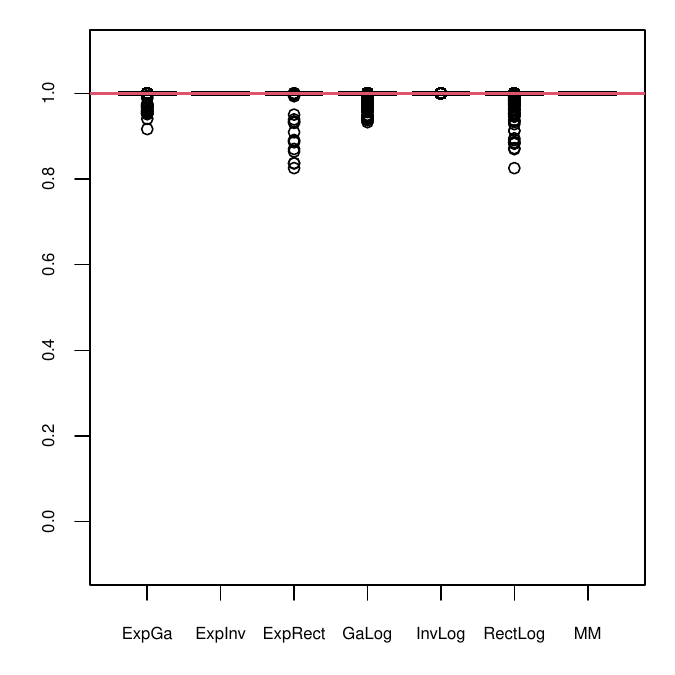}
\includegraphics[width=4cm]{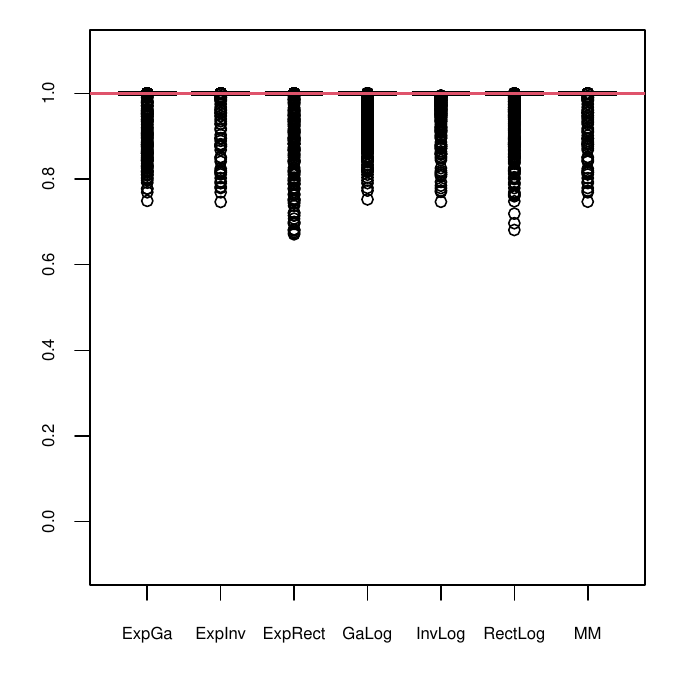}\\
\includegraphics[width=4cm]{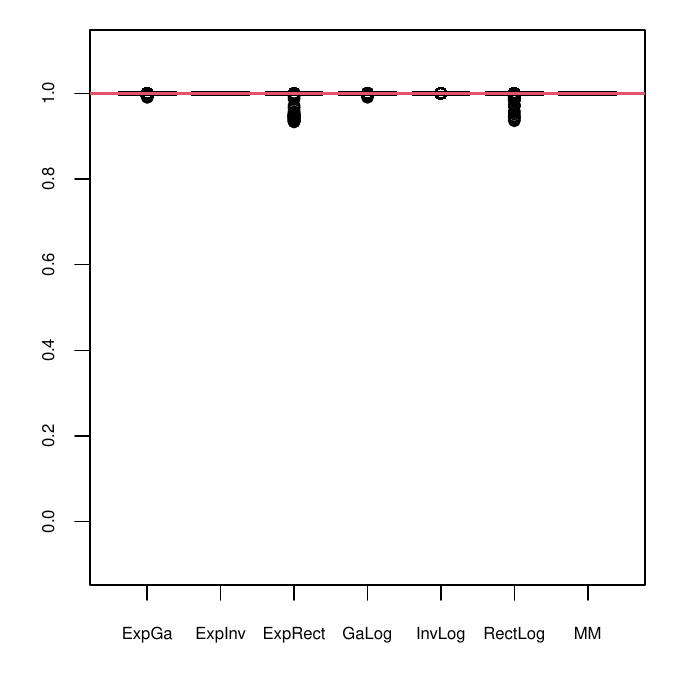}
\includegraphics[width=4cm]{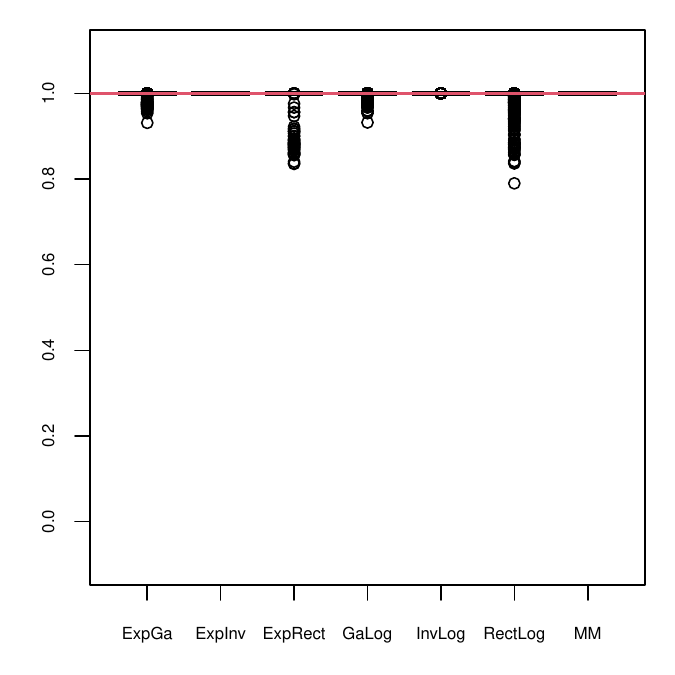}
\includegraphics[width=4cm]{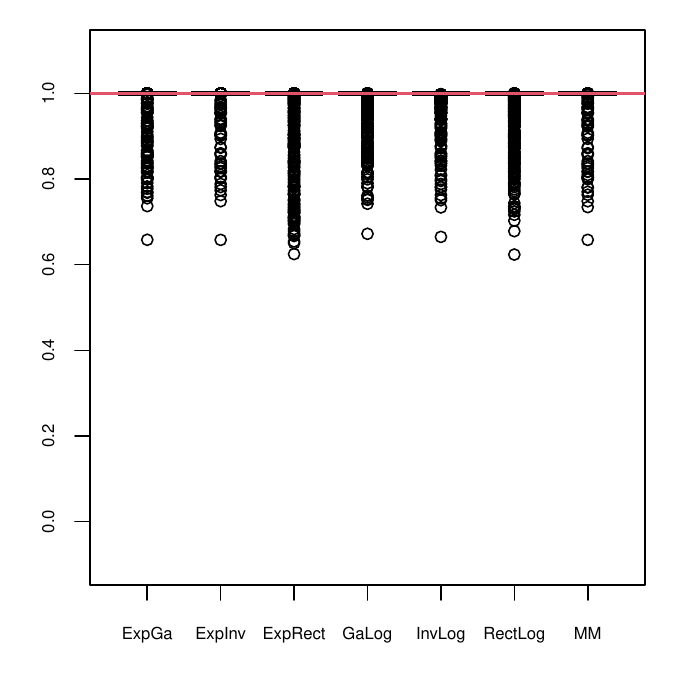}\\
\includegraphics[width=4cm]{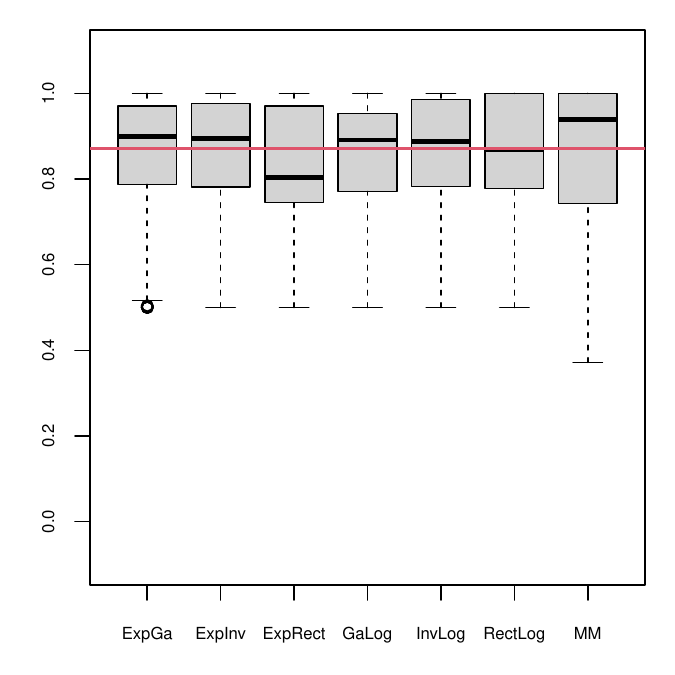}
\includegraphics[width=4cm]{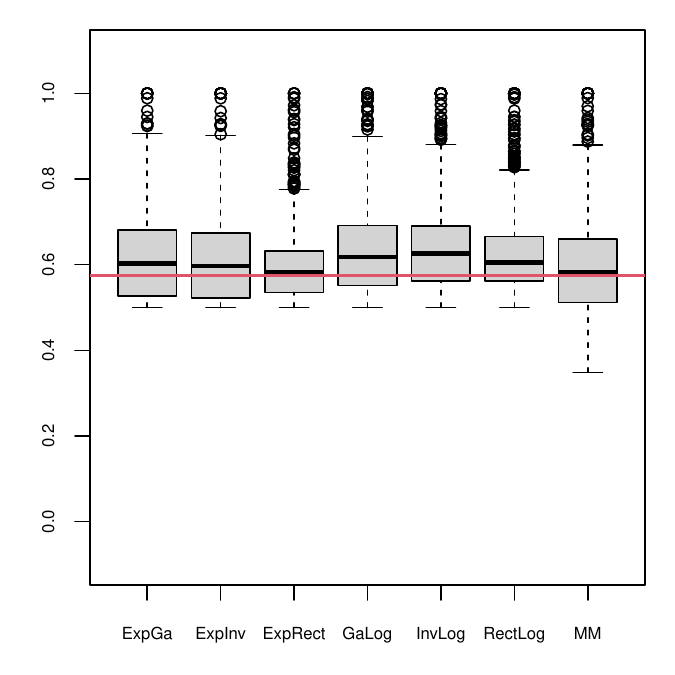}
\includegraphics[width=4cm]{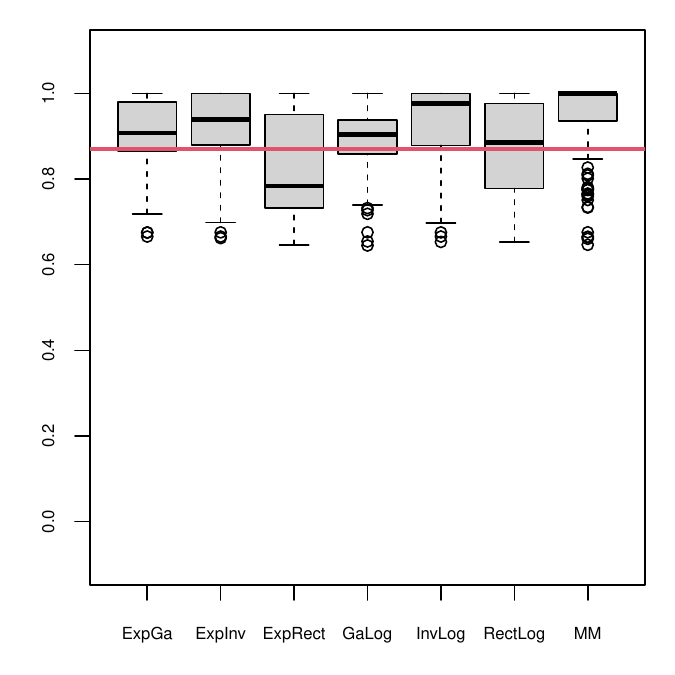}
\includegraphics[width=4cm]
{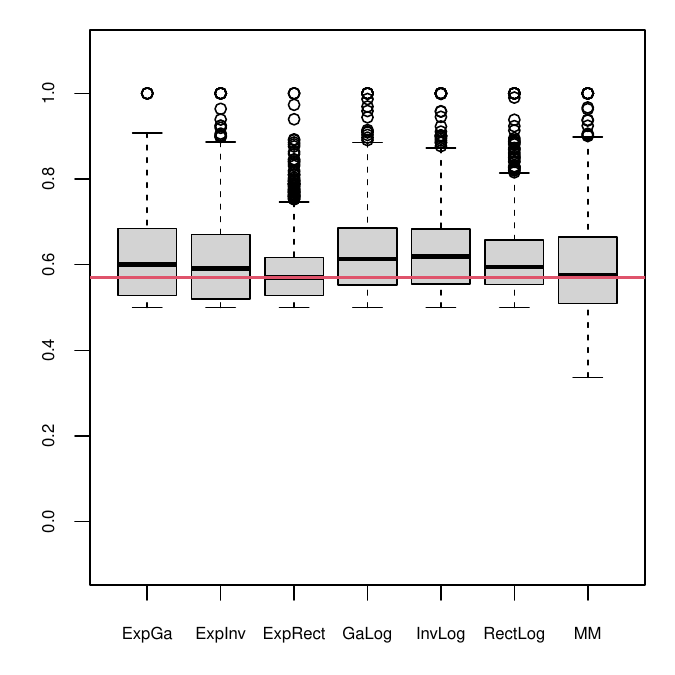}
\caption{\label{fig:eta} Box-plots of the estimated residual tail dependence coefficients $\hat{\eta}$ over 1000 iterations.
The description about the plot settings aligns with that given in Figure~\ref{fig:alpha}.
The red lines indicate the true values of $\eta$ for each case. }
\end{figure}

For estimating the lower-right corner of the extreme set, all additive mixtures provide reasonably accurate estimates of the probability $\prob((X,Y)\in C_{lr})$, where $C_{lr}=(8,\infty)\times(0,7)$, across all cases. Note that a different scale is used on the y-axis.

\begin{figure}[ht!]
\centering
\includegraphics[width=4cm]{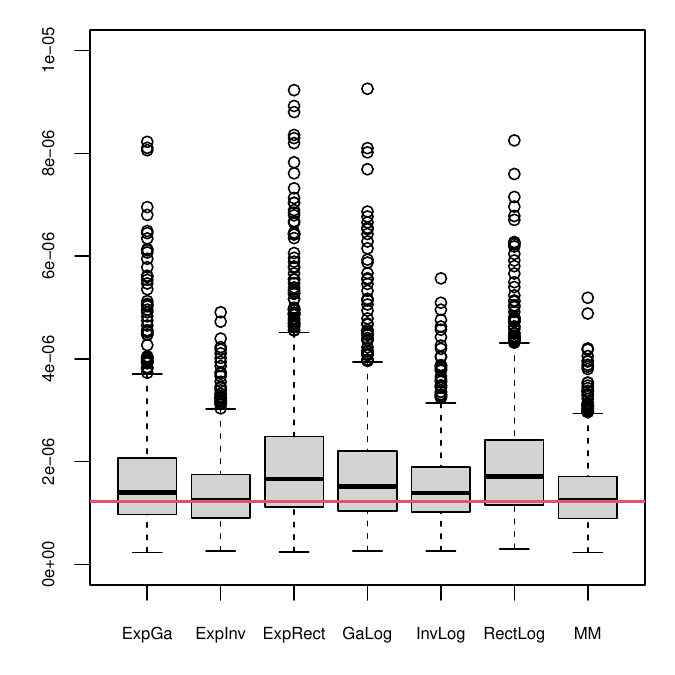}
\includegraphics[width=4cm]{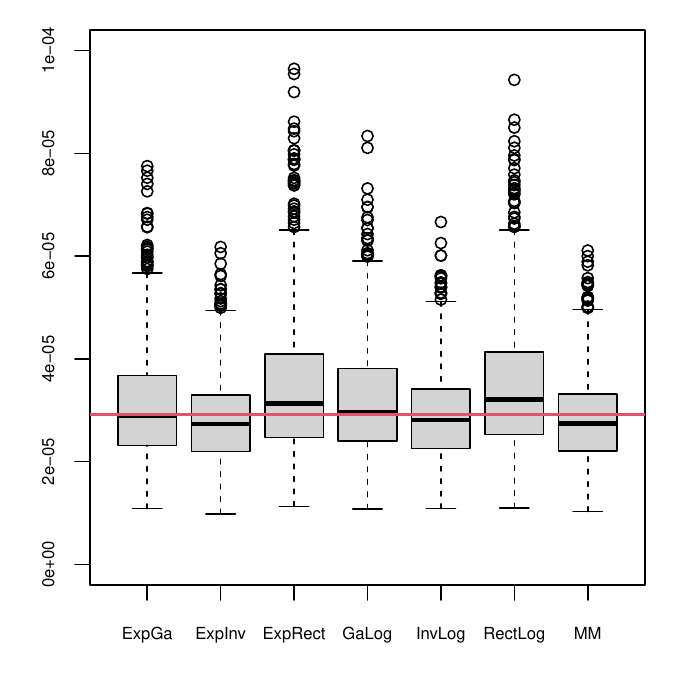}
\includegraphics[width=4cm]{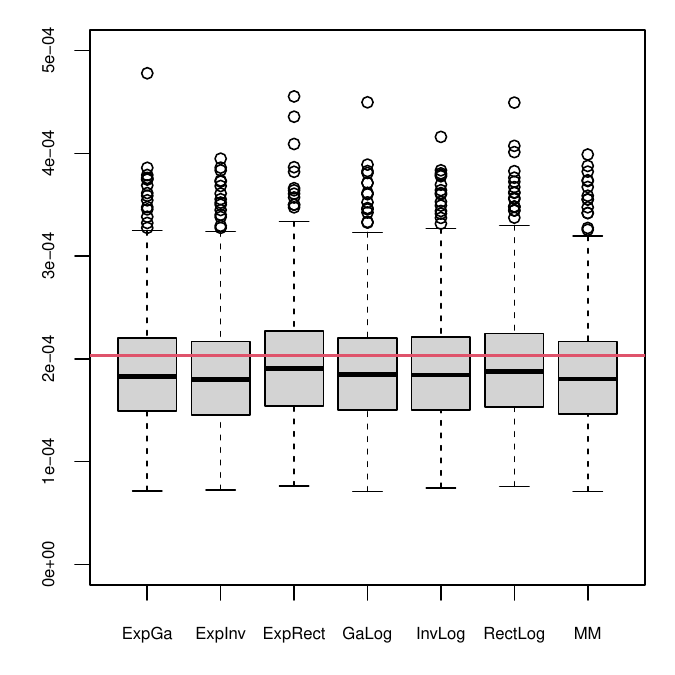}\\
\includegraphics[width=4cm]{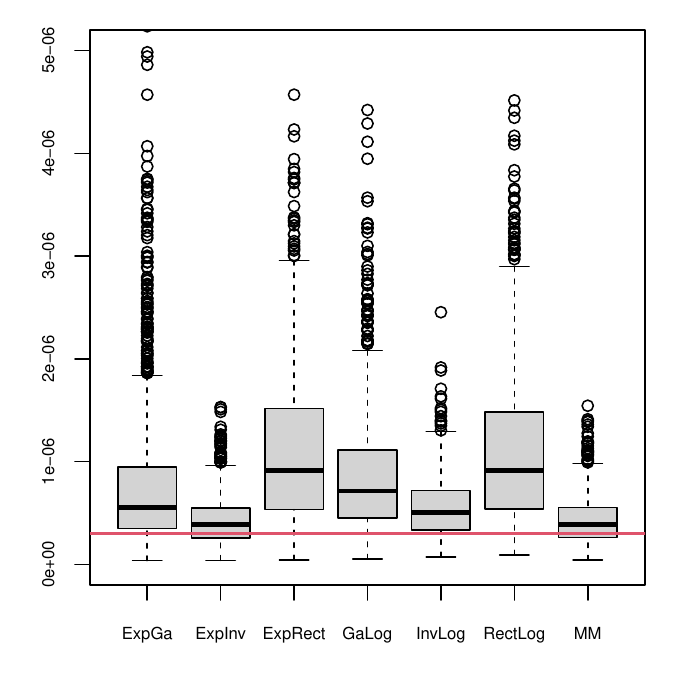}
\includegraphics[width=4cm]{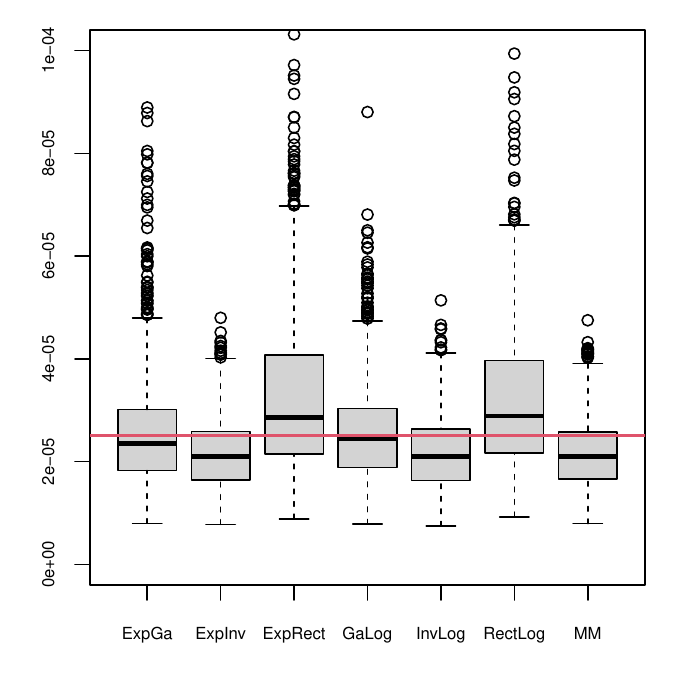}
\includegraphics[width=4cm]{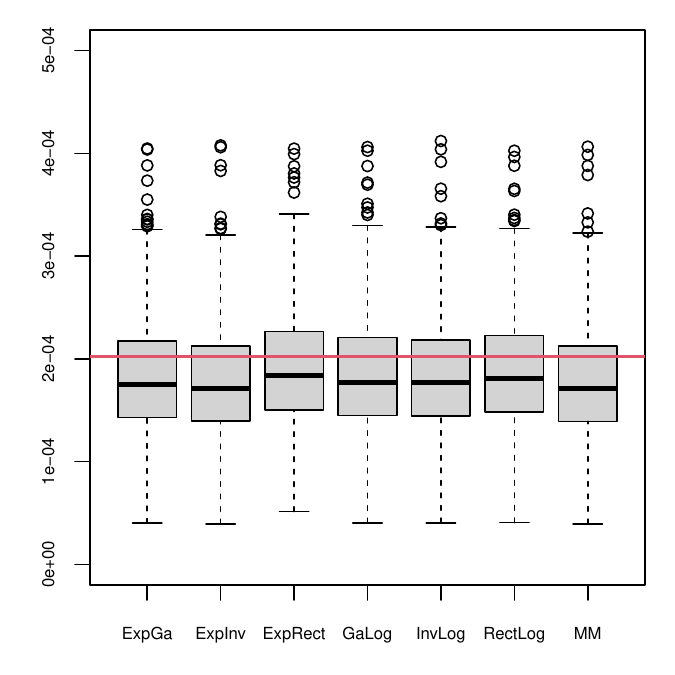}\\
\includegraphics[width=4cm]{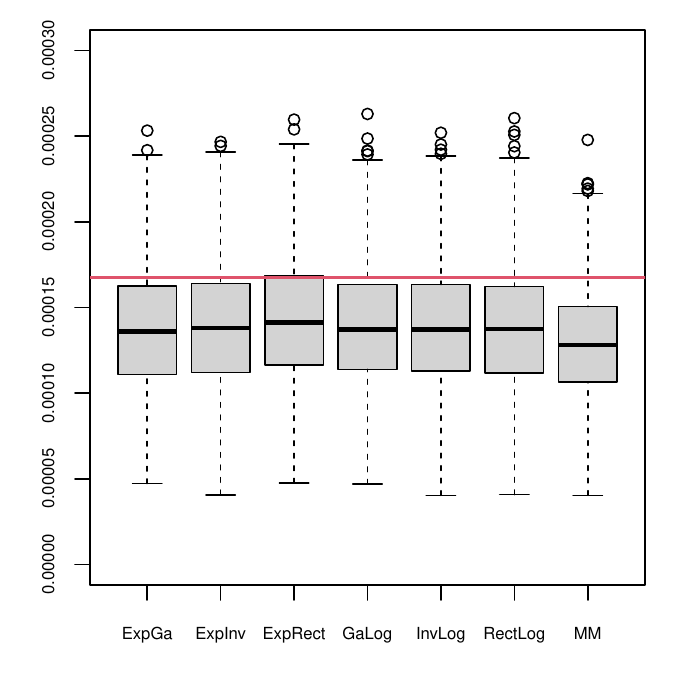}
\includegraphics[width=4cm]{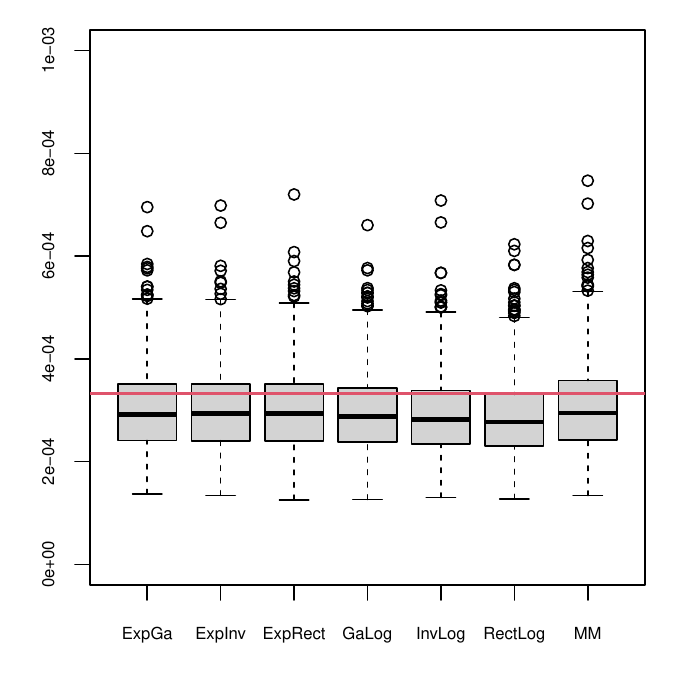}
\includegraphics[width=4cm]{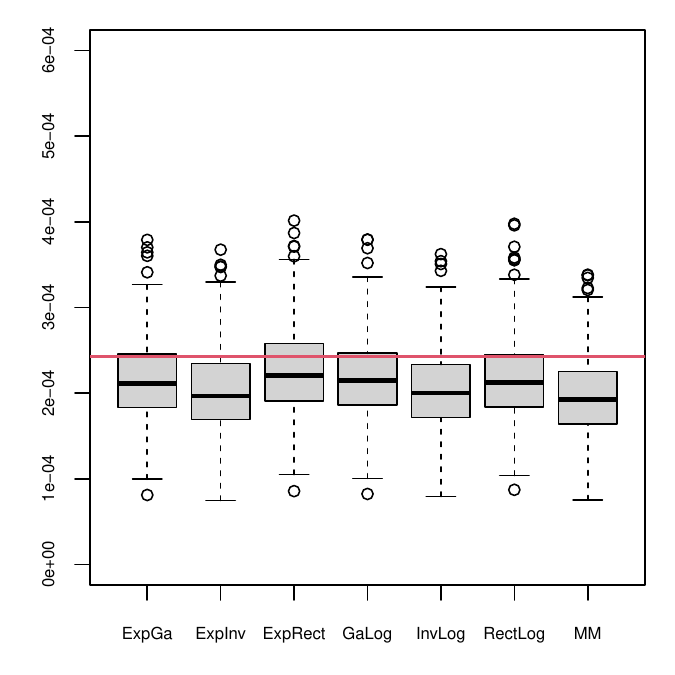}
\includegraphics[width=4cm]{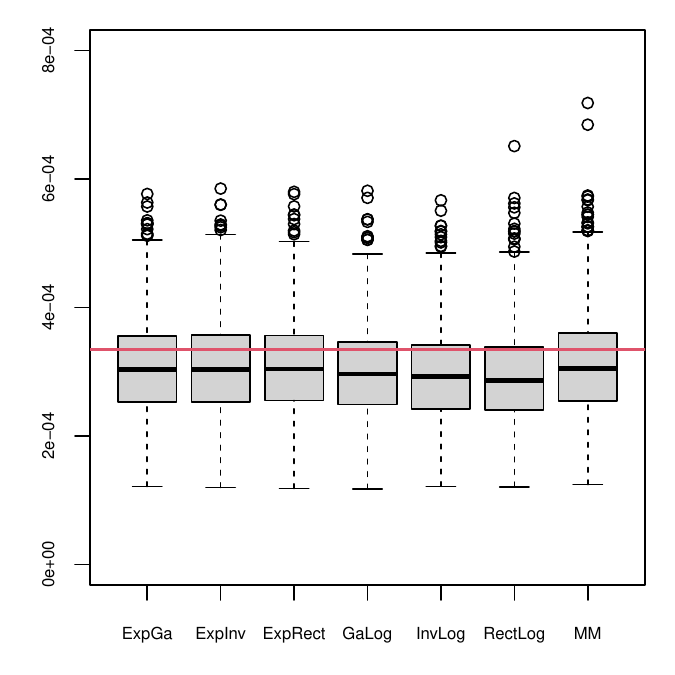}
\caption{\label{fig:probEst} 
Box-plots of the estimated probabilities for the set $C_3$ over 1000 iterations.
The description about the plot setting aligns with that given in Figure~\ref{fig:alpha}.
The red lines indicate the true probabilities for each case.
}
\end{figure}

We also create the empirical gauge function estimates derived from a rolling-windows quantile method over 1000 iterations in Figure~\ref{fig:empiricialGauges}.
The overall shape coincides with the fitted gauge functions.

\begin{figure}[ht!]
\centering
\includegraphics[width=4cm]{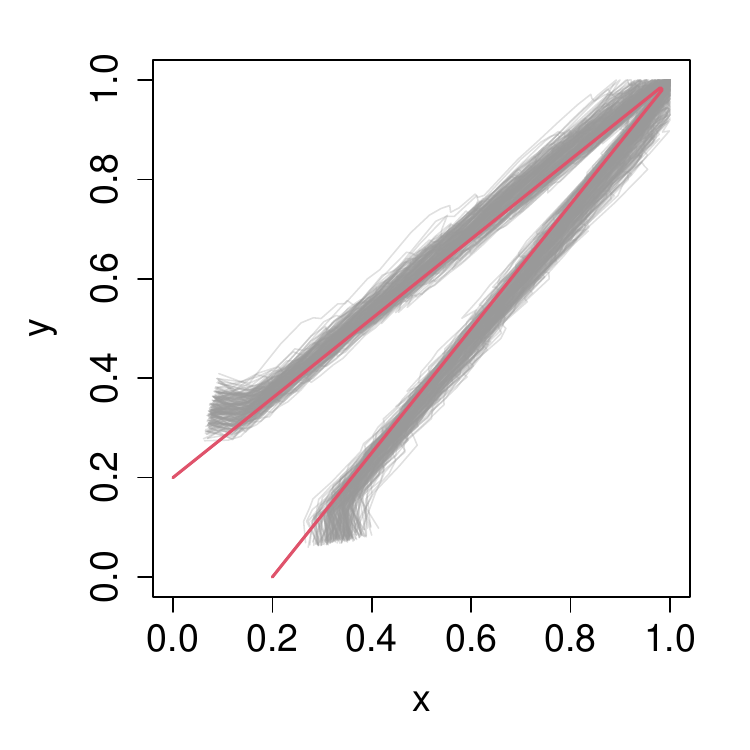}
\includegraphics[width=4cm]{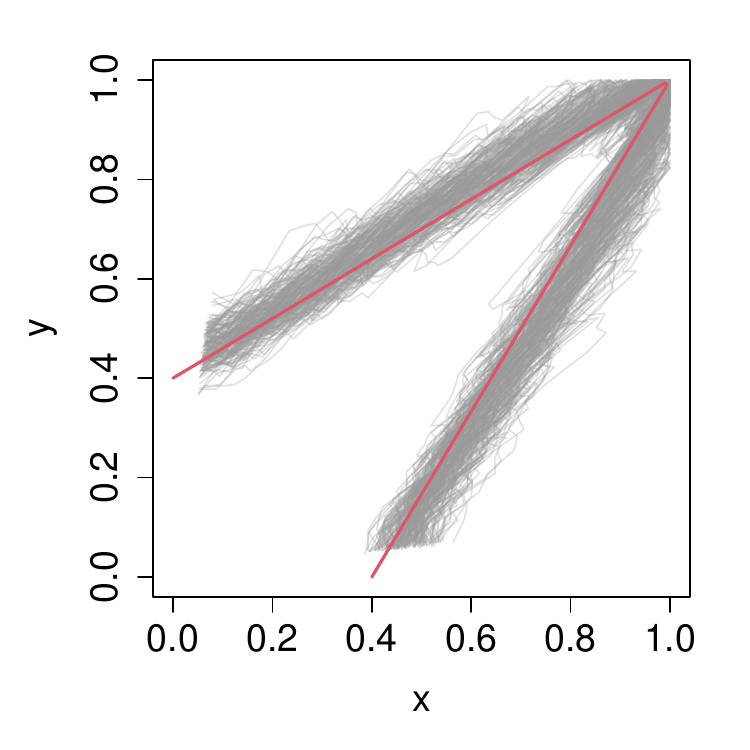}
\includegraphics[width=4cm]{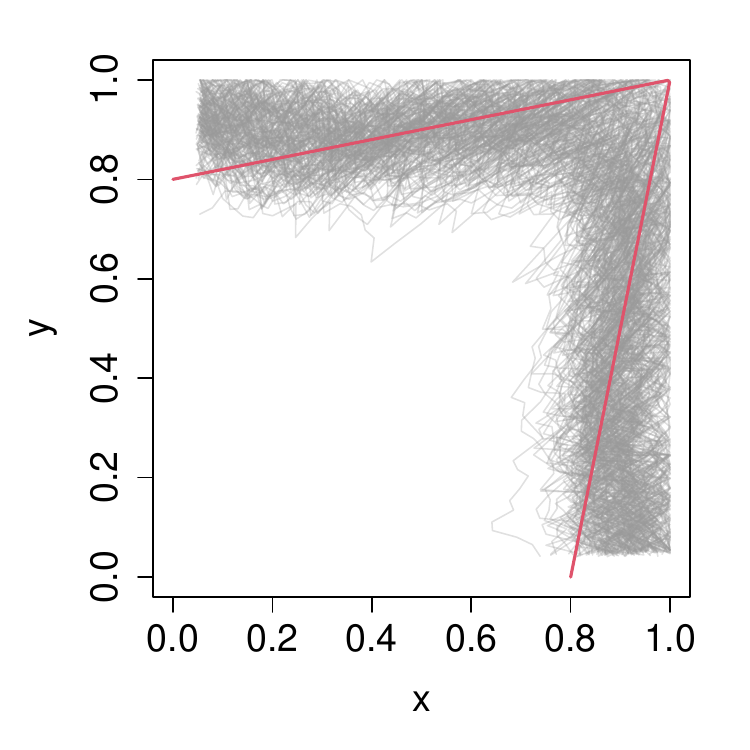}\\
\includegraphics[width=4cm]{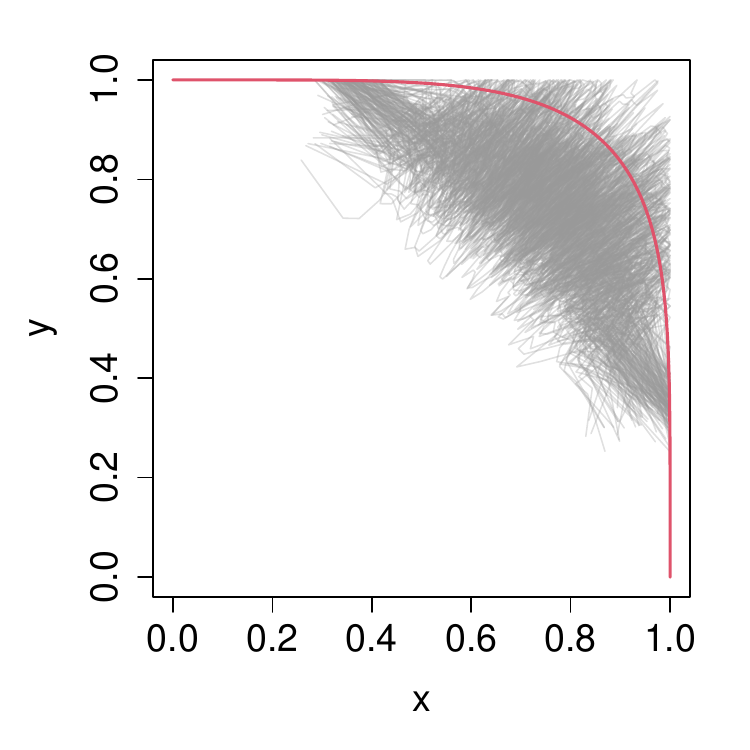}
\includegraphics[width=4cm]{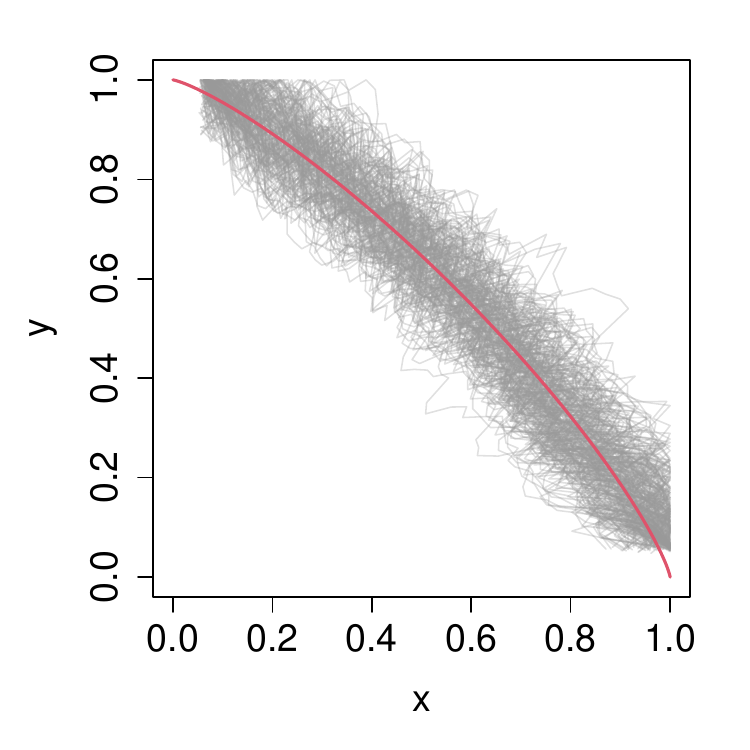}
\includegraphics[width=4cm]{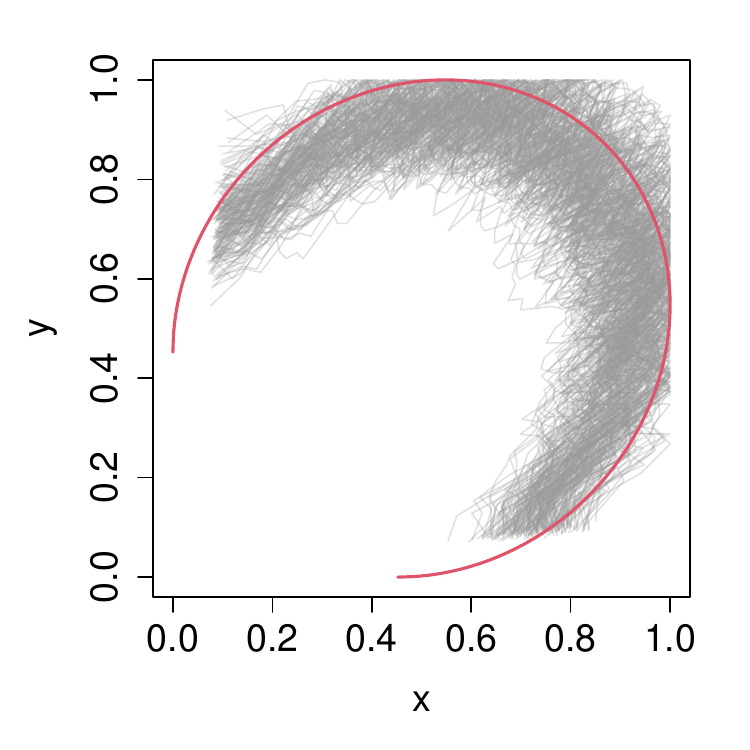}
\includegraphics[width=4cm]{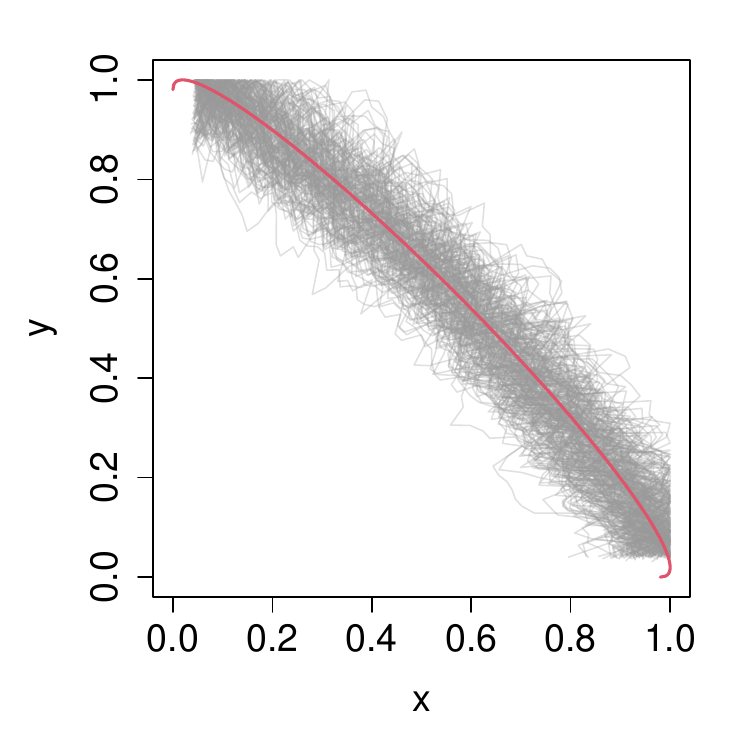}
\caption{\label{fig:empiricialGauges} Empirical gauge function estimates derived from a rolling-windows quantile method over 1000 iterations.
The red lines indicate the fitted gauge functions based on the obtained ML estimates.}
\end{figure}

\section{Numerical results from river flow data}

Based on the geometric criteria, we observed that asymptotic dependence is identified only for the pair $(X_2,X_3).$
Numerical results of the fitted additive mixtures for the pair $(X_2,X_3)$ are included in Table~\ref{tab:X23}.

\begin{table}[ht!]
    \centering
    \begin{tabular}[t]{|c||c|c|c|} 
      \hline
      & ExpGa & ExpInv & ExpRect \\
      \hline\hline
      $\hat{\lambda}$ & 1.43 & 1.43 & 1.43 \\
      \hline
      $\hat{\gamma}_{\text{\tiny HW}}$ & 1.10 & 1.16 & 1.11 \\
      \hline
      $\hat{\theta}_{\text{\tiny HW}}$ & 0.63 & 0.23 & 0.72 \\
      \hline
      nll & 993.7 & 993.7 & 993.7 \\
      \hline
      AIC & 1993.41 & 1993.41 & 1993.41 \\
      \hline
      $\hat{\alpha}$ & 1.00 & 1.00 & 1.00 \\
      \hline
      $\hat{\eta}$ & 1.00 & 1.00 & 1.00 \\
      \hline
      \end{tabular}
    \begin{tabular}[t]{|c||c|c|c|} 
    \hline
    & GaLog & InvLog & RectLog \\
    \hline\hline
    $\hat{\lambda}$ & 1.36 & 1.43 & 1.34 \\
    \hline
    $\hat{\theta}_{\text{\tiny Amix}}$ & 0.81 & 0.99 & 0.34 \\
    \hline
    $\hat{\gamma}_{\text{\tiny log}}$ & 0.50 & 0.71 & 0.89 \\
    \hline
    $\hat{p}$ & 0.73 & 0.16 & 0.16 \\
    \hline
    nll & 993.59 & 993.7 & 993.6 \\
    \hline
    AIC & 1995.17 & 1995.41 & 1995.19 \\
    \hline
    $\hat{\alpha}$ & 0.87 & 1.00 & 1.00 \\
    \hline
    $\hat{\eta}$ & 0.99 & 1.00 & 1.00 \\
    \hline
    \end{tabular}
    \begin{tabular}[t]{|c||c|} 
    \hline
    & MM \\
    \hline\hline
    $\hat{\lambda}$ & 1.43 \\
    \hline
    $\hat{\theta}_{\text{\tiny MM}}$ & 0.86 \\
    \hline
    nll & 993.7 \\
    \hline
    AIC & 1991.41 \\
    \hline
    $\hat{\alpha}$ & 1.00 \\
    \hline
    $\hat{\eta}$ & 1.00 \\
    \hline
    \end{tabular}
\caption{Summary of the fitted results for the pair $(X_2,X_3)$.
The symbols $\hat{\lambda}_{\text{\tiny tg}}$  $\hat{\gamma}_{\text{\tiny HW}}$, $\hat{\theta}_{\text{\tiny HW}}$, $\hat{\theta}_{\text{\tiny Amix}}$, $\hat{\gamma}_{\text{\tiny log}}$, $\hat{p}$, and $\hat{\theta}_{\text{\tiny MM}}$ represent the ML estimates for the shape parameter of the truncated gamma distribution, the scale parameter for the HW model, and the correlation parameter for the Gaussian gauge function $g_{\bV}$, the parameter associated with the first gauge function $g_{\bX}^{[1]}$ in the additive mixture model, the parameter for the logistic gauge function, the weight parameter, and the parameter associated with the max-min gauge function, respectively.
`nll' refers to the negative log-likelihood value, $\hat{\alpha}$ denotes the estimated slope for the conditional extremes model, and $\hat{\eta}$ corresponds to the estimated residual tail dependence coefficient.
}
\label{tab:X23}
\end{table}

\end{document}